\documentclass[onecolumn,12pt]{IEEEtran}

\usepackage{amsmath,amssymb,array}
\usepackage{siunitx}
\usepackage{longtable,tabularx}
\usepackage{algorithm}
\usepackage{algorithmic}
\usepackage{multirow}
\usepackage{color}
\usepackage{bm}
\usepackage[pdftex]{graphicx}

\DeclareMathOperator{\Tr}{\mathrm{Tr}}

\newtheorem{theorem}{Theorem}
\newtheorem{definition}{Definition}
\newtheorem{lemma}{Lemma}%
\newtheorem{remark}{Remark}%
\newtheorem{example}{Examle}%
\newtheorem{proposition}{Propositon}%
\newcommand{\argmax}{\mathop{\rm argmax}\limits}
\newcommand{\argmin}{\mathop{\rm argmin}\limits}
\newcommand{\interior}{\mathop{\rm int}\limits}

\def\bX{\mathbb X}
\def\bY{\mathbb Y}

\def\Ker{\mathop{\rm Ker}}
\def\im{\mathop{\rm Im}}

\def\rs{\mathop{\rm s}}
\def\ra{\mathop{\rm a}}

\def\cX{\mathcal X}
\def\cY{\mathcal Y}

\def\bs{\bm{s}}
\def\bx{\bm{x}}
\def\by{\bm{y}}
\def\bS{\bm{S}}
\def\bX{\bm{X}}
\def\bY{\bm{Y}}

\def\X{{\mathcal{X}}}
\def\Y{{\mathcal{Y}}}

\def\H{{\mathcal{H}}}
\def\cH{{\mathcal{H}}}

\def\S{{\mathcal{S}}}

\def\Pro{\mathop{\Gamma}\nolimits}

\def\QED{\mbox{\rule[0pt]{1.5ex}{1.5ex}}}

\newcommand{\qed}{\hfill \QED}

 \newenvironment{proofof}[1]{\vspace*{5mm} \par \noindent
{\it Proof of #1:\hspace{2mm}}}{\qed
}
\def\Label#1{\label{#1}\ [\ \text{#1}\ ]\ }
\def\Label{\label}
\begin{document}
\title{Bregman divergence based em algorithm 
and its application to classical and quantum rate distortion theory
}

\author{
Masahito~Hayashi~\IEEEmembership{Fellow,~IEEE}\thanks{Masahito Hayashi is with 
Shenzhen Institute for Quantum Science and Engineering, Southern University of Science and Technology, Nanshan District,
Shenzhen, 518055, China,
International Quantum Academy (SIQA), Futian District, Shenzhen 518048, China,
Guangdong Provincial Key Laboratory of Quantum Science and Engineering,
Southern University of Science and Technology, Nanshan District, Shenzhen 518055, China,
and
Graduate School of Mathematics, Nagoya University, Nagoya, 464-8602, Japan.
(e-mail:hayashi@sustech.edu.cn, masahito@math.nagoya-u.ac.jp)}}
\date{}
\markboth{M. Hayashi: Bregman divergence based em algorithm}{}

\maketitle

\begin{abstract}
We formulate em algorithm in the framework of Bregman divergence,
which is a general problem setting of information geometry.
That is, we address the minimization problem of 
the Bregman divergence between an exponential subfamily and a mixture subfamily
in a Bregman divergence system.
Then, we show the convergence and its speed under several conditions. 
We apply this algorithm to rate distortion and its variants including the quantum setting,
and show the usefulness of our general algorithm.
In fact, existing applications of Arimoto-Blahut algorithm to rate distortion theory
make the optimization of 
the weighted sum of the mutual information and the cost function by using 
the Lagrange multiplier.
However, in the rate distortion theory, it is needed to minimize the 
mutual information under the constant constraint for the cost function.
Our algorithm directly solves this minimization.
In addition, we have numerically checked the convergence speed of our algorithm 
in the classical case of rate distortion problem.
\end{abstract}


\section{Introduction}\Label{S1}
Em algorithm is known as a useful algorithm in various areas including machine learning and neural network \cite{Amari,Fujimoto,Allassonniere}.
Its basic idea can be backed to the reference \cite{Bol}.
In information theory, the Arimoto-Blahut algorithm \cite{Blahut,Arimoto} is known as a powerful tool to calculate
various information-theoretical optimization problems including mutual information.
Both algorithms are composed of iterative steps.
In this paper, we apply em algorithm to rate distortion and its variants including the quantum setting.

Although em algorithm has several variants,
the most general form is given as the minimum divergence between 
a mixture family and an exponential family \cite{Amari}.
However, the convergence speed of em algorithm is not known in general.
Moreover, it has a possibility to converge to a local minimum \cite{Amari,Fujimoto,Allassonniere}.
Therefore, it is needed to guarantee the convergence to the global minimum
and clarify the convergence speed.
In this paper, to address these problems in a unified viewpoint,
similar to the paper \cite{Fujimoto},
we formulate em algorithm in a framework of Bregman divergence, which is given from 
a general smooth convex function as a general problem setting of information geometry \cite{Amari-Nagaoka,Amari-Bregman}.
In this general framework, we derive a necessary condition for the global convergence,
and discuss the convergence speed.
When an additional condition is satisfied, this algorithm has exponential convergence.
This additional condition is easily satisfied when the iteration is close to the true value.
Hence, this algorithm rapidly converges around the true value under a certain condition.

When an information-theoretical optimization problem is written in the above form,
em algorithm can be applied to it.
As a typical example, we consider the rate distortion problem,
which is written as a minimization of the mutual information under a linear constraint
to a given distribution.
That is, the objective distribution of this problem belongs to a certain mixture family.
Mutual information is written as the minimum divergence between 
a given distribution and the set of independent distributions, which forms an exponential family.
Hence, this minimization is given as the minimization of 
the divergence between the given mixture family and the exponential family
composed of independent distributions.
The minimization for the rate distortion problem
was studied by Blahut \cite{Blahut} and various papers \cite{Csiszar,Cheng,YSM}.
However, to remove the constraint, they change the objective function 
by using a Lagrange multiplier.
That is, they minimize the weighted sum of the original objective function
and the cost function, whereas the Lagrange multiplier corresponds to 
the weight coefficient.
When the Lagrange multiplier is suitably chosen, 
the solution of this modified minimization given the solution of the 
original minimization.
However, no preceding paper showed how to choose 
the Lagrange multiplier.
Therefore, it was required to develop how to find the suitable 
the Lagrange multiplier.
Fortunately, the set of conditional distribution with a linear constraint
forms a mixture family.
Hence, our method can directly solve the required minimization 
with a linear constraint.
Then, we apply these obtained general results to 
several variants \cite{Cover,LZP} of the rate distortion problem including the quantum setting \cite{DHW}.


The remaining part of this paper is organized as follows.
Section \ref{Section.IG-structure} formulates general basic properties for Bregman divergence.
Section \ref{S3} explains how the set of probability distributions and the set of quantum states 
satisfy the condition for Bregman divergence.
Section \ref{Sec:em} states em algorithm in the framework of Bregman divergence,
and derives its various properties. 
Section \ref{S5} applies the above general results to classical rate distribution and its variants.
Section \ref{S6} applies them to its quantum extension.

\section{Bregman divergence: Information geometry based on convex function}\Label{Section.IG-structure}
In this section, we prepare general basic properties for Bregman divergence.
Originally, information geometry was studied as the geometry of probability distributions.
This structure can be generalized as a geometry of a smooth strictly convex function, which is called 
Bregman Divergence.
This section discusses several useful properties of Bregman Divergence.

\subsection{Legendre transform}\Label{S2-0}
In this paper, 
a sequence $a= (a^i)_{i=1}^k$ with an upper index expresses
an vertical vector 
and 
a sequence $b= (b_i)_{i=1}^k$ with an lower index expresses
an horizontal vector as
\begin{align}
a= \left(
\begin{array}{c}
a^1 \\
a^2 \\
\vdots \\
a^k
\end{array}
\right), \quad
b= (b_1, b_2,\ldots, b_k).
\end{align}

 Let $\Theta$ be an open convex set in $\mathbb{R}^d$ and 
$F:\Theta \rightarrow \mathbb{R}$ be a $C^\infty$-class strictly convex function.
We introduce another parametrization 
$\eta =(\eta_1, \ldots, \eta_d)\in \mathbb{R}^d$ as
\begin{align}
\eta_j := \partial_j F(\theta),\Label{du1}
\end{align}
where $\partial_j$ expresses the partial derivative for the $j$-th variable.
We introduce the vector $\nabla^{(e)} [F](\theta):=(\partial_j F(\theta))_{j=1}^d$.
Hence, the relation \eqref{du1} is rewritten as
\begin{align}
\eta = \nabla^{(e)} [F](\theta).\Label{M1}
\end{align}
Therefore, $\nabla^{(e)}$ can be considered as a horizontal vector.

Since $F$ is a $C^\infty$-class strictly convex function,
this conversion is one-to-one.
the parametrization 
$\eta_j$ is called the mixture parameter.
We denote the open set of vectors $\eta(\theta) =(\eta_1, \ldots, \eta_d)$ given in \eqref{du1},
by $\Xi$.
For $\eta \in \Xi$,
we define the {\it Legendre transform} $F^*={\cal L}[F]$ of $F$ 
\begin{align}
F^*(\eta)=\sup _{\theta \in \Theta} \langle \eta,\theta\rangle -F(\theta).
\Label{MN1}
\end{align}

We have \cite[Section 3]{Fujimoto}\cite[Section 2.2]{hayashi}
\begin{align}
\partial^j F^*
(\eta(\theta) )=\theta^j, \Label{du2}
\end{align}
where $\partial^j$ expresses the partial derivative for the $j$-th variable under the mixture parameter.
We introduce the vector $\nabla^{(m)} [F^*](\eta):=(\partial^j F^*(\eta))_{j=1}^d$.
Hence, the relation \eqref{du2} is rewritten as
\begin{align}
\theta = \nabla^{(m)} [F^*] (\eta(\theta) ).\Label{M2}
\end{align}

In later discussion, we address subfamilies related to
$m$ vectors $v_1, \ldots, v_m \in \mathbb{R}^d$.
For a preparation for such cases, we prepare 
the following two equations, which  will be used for calculations based on mixture parameters.
Then, we define a $d \times m $ matrix $V$ as $(v_1 \ldots v_m )$.
The multiplication function of V from the left (right) hand side  
is denoted by $L[V]$ ($R[V]$).
Since 
\begin{align}
\partial_j (F\circ L[V])(\theta) =
\frac{\partial F}{\partial \theta^j}(V \theta) =
\sum_{i} v_{j}^i \partial_i F(V \theta) =
(R[V]\circ (\nabla^{(e)} [F] )\circ L[V] (\theta))_j,
\end{align}
we have 
\begin{align}
\nabla^{(e)} [F\circ L[V]] &=   R[V]\circ (\nabla^{(e)} [F] )\circ L[V].\Label{E12}
\end{align}

In the same way, we can show 
\begin{align}
\nabla^{(m)} [ F^* \circ R[V]]
=L[V] \circ \nabla^{(m)} [ F^*] \circ R[V]
\Label{MDF}.
\end{align}
Also, we have
\begin{align}
& (F^* \circ R[V])^* (\theta')
=
\sup_{\eta} \langle \eta ,\theta'\rangle
- \sup _{\theta \in \Theta} \langle \eta V,\theta\rangle -F( \theta) \nonumber\\
=&
\sup_{\eta} \inf_{\theta \in \Theta}
\langle \eta ,\theta'- V\theta\rangle
 +F( \theta) 
=\inf_{\theta: \theta'= V\theta} F( \theta) .
\Label{VCA}
\end{align}

When $V$ is one-to-one, we 
define the function $F\circ L[V^{-1}]$ on $L[V] (\Theta)$. Since
\begin{align}
& F^* \circ R[V] (\eta)
=
\sup _{\theta \in \Theta} \langle \eta V,\theta\rangle -F( \theta)
=
\sup _{\theta \in \Theta} \langle \eta ,V \theta\rangle -F( \theta)\nonumber\\
=&
\sup _{\theta' \in L[V] (\Theta)} \langle \eta , \theta\rangle -F(V^{-1} \theta')
=
(F\circ L[V^{-1}])^*(\eta) ,
\end{align}
we have 
\begin{align}
(F\circ L[V^{-1}])^* &= F^* \circ R[V] 
\end{align}
Combining the above two relations, we have
\begin{align}
\nabla^{(m)} [ (F \circ L[V^{-1}])^*]
=\nabla^{(m)} [ F^* \circ R[V]]
=L[V] \circ \nabla^{(m)} [ F^*] \circ R[V].
\Label{L14}
\end{align}

\subsection{Exponential subfamily}\Label{S2-1}
A subset $\mathcal{E} \subset \Theta$ is called an {\it exponential subfamily} 
generated by $l$ linearly independent vectors $v_1,\ldots, v_l \in \mathbb{R}^d$
at $\theta_0 \in \Theta$ 
when the subset $\mathcal{E}$ is given as
\begin{align}
    \mathcal{E} &= \left\{ \phi_{{\cal E}}^{(e)}(\bar{\theta})  \in \Theta \left| 
    \bar{\theta} \in \Theta_{{\cal E}}    \right. \right\}.
\end{align}
In the above definition,
$\phi_{{\cal E}}^{(e)} (\bar{\theta})$ is defined 
for $ \bar{\theta} =( \bar{\theta} ^1, \ldots,  \bar{\theta}^l ) \in 
\mathbb{R}^l$ as
\begin{align}
\phi_{{\cal E}}^{(e)}(\bar{\theta}) &:= \theta_0 + \sum_{j=1}^l \bar{\theta}^j v_j  
\end{align}
and the set $\Theta_{{\cal E}}$ is defined as 
\begin{align}
\Theta_{{\cal E}} := \{ \bar{\theta} \in \mathbb{R}^l |  
\phi_{{\cal E}}^{(e)}(\bar{\theta}) \in \Theta\}.
\Label{const1-UE}
\end{align}
Since $\Theta$ is an open set, the set $\Theta_{{\cal E}}$ is an open set.
In the following, we restrict the domain of $\phi_{{\cal E}}^{(e)}$
to $\Theta_{{\cal E}}$.
We define the inverse map 
$\psi_{{\cal E}}^{(e)}:=(\phi_{{\cal E}}^{(e)})^{-1}: {\cal E}\to \Theta_{{\cal E}}$.

For an exponential subfamily ${\cal E}$, we define 
the function $F_{{\cal E}}$ as
\begin{align}
F_{{\cal E}}(\bar{\theta}):= F (\phi_{{\cal E}}^{(e)}(\bar{\theta})).
\end{align}
In fact, even in an exponential subfamily ${\cal E}$,
we can employ the mixture parameter
${\psi}_{{\cal E},j}^{(m)} (\phi_{{\cal E}}^{(e)}(\bar{\theta})):= 
\partial_j F_{{\cal E}} (\bar{\theta})$
because the map $\bar{\theta}\mapsto F_{{\cal E}} (\bar{\theta})$ is also
a $C^\infty$-class strictly convex function.
We define the set $\Xi_{{\cal E}}:= \{ (\partial_j F_{{\cal E}} (\bar{\theta}))_{j=1}^l \}_{\bar{\theta} \in \Theta_{{\cal E}}}$.
We define the inverse map 
$\phi_{{\cal E}}^{(m)}:=(\psi_{{\cal E}}^{(m)})^{-1}: 
\Xi_{{\cal E}}\to {\cal E}$.

\subsection{Mixture subfamily}\Label{S2-2}
For $d$ linearly independent vectors $u_1, \ldots, u_d \in \mathbb{R}^d$, and 
a vector $a=(a_1, \ldots, a_{d-k} )^T \in \mathbb{R}^{d-k}$, 
a subset $\mathcal{M} \subset \Theta$ is called a {\it mixture subfamily} 
generated by the constraint 
\begin{align}
\sum_{i=1}^d u^i_{k+j} \partial_i F(\theta) =a_j\Label{const1}
\end{align}
 for $j=1, \ldots,d-k$
when the subset $\mathcal{M}$ is written as
\begin{align}
    \mathcal{M} = \left\{ \theta  \in \Theta \left| \hbox{ Condition \eqref{const1} holds.} \right. \right\} .
\end{align}
We define a $d \times d $ matrix $U$ as $(u_1 \ldots u_d )$.
To make a parametrization in the above mixture subfamily ${\cal M}$,
we set the new natural parameter $\bar{\theta}=(\bar{\theta}^1, \ldots, \bar{\theta}^d)$ as
$\theta=U \bar{\theta} $,
and introduce the new mixture parameter
\begin{align}
\bar{\eta}_i
=\partial_j (F \circ U) (\bar{\theta})
\Label{Dif}
\end{align}
Since $\bar{\eta}_{k+i}= a_i$ for $i=1, \ldots, d-k$ in ${\cal M}$,
the initial $k$ elements $\bar{\eta}_1, \ldots, \bar{\eta}_{k}$
gives a parametrization for ${\cal M}$.
For the parametrization, we define the map $\psi_{\cal M}^{(m)}$ 
as $\psi_{\cal M}^{(m)} ( U \bar{\theta}):= (\partial_j (F \circ U) (\bar{\theta}))_{j=1}^k$.
We define the set $\Xi_{{\cal M}}:= 
\{ \psi_{\cal M}^{(m)} (\theta) |  {\theta} \in {\cal M}\}$ of 
the new mixture parameters,
and the inverse map 
$\phi_{{\cal M}}^{(m)}:=(\psi_{{\cal M}}^{(m)})^{-1}: 
\Xi_{{\cal M}}\to {\cal M}$.
Since $\Theta$ is an open set, 
the set $\Xi_{{\cal M}}$ is an open subset of $\mathbb{R}^k$.
When an element $\bar{\eta} \in \Xi_{{\cal M}}$ satisfies 
$\bar{\eta}_j =\partial_j (F \circ U) (\bar{\theta})$ for $j=1, \ldots, k$,
we have
\begin{align}
\partial^i (F\circ U)^*( \bar{\eta},a) = \bar{\theta}^i\Label{CO1}
\end{align}
 for $i=1, \ldots, d$.
Since $\bar{\eta} \mapsto(F\circ U)^*( \bar{\eta},a) $ is strictly convex,
the map $\bar{\eta} \mapsto (\partial^i (F\circ U)^*( \bar{\eta},a))_{i=1}^k$
is one-to-one. 
Hence, the initial $k$ elements $\bar{\theta}^1, \ldots, \bar{\theta}^{k}$
give a parametrization for ${\cal M}$.
That is, we have 
\begin{align}
( (U^{-1} \theta)^i )_{i=1}^k
= (\partial^i (F\circ U)^*( \psi_{\cal M}^{(m)} ( {\theta}) ,a))_{i=1}^k.\Label{NAY}
\end{align}
We define the set $\Theta_{\cal M}:= \{ ( (U^{-1} \theta)^i )_{i=1}^k | {\theta} \in {\cal M} \}$. This set is written as
\begin{align}
\Theta_{\cal M}= 
\Bigg\{(\theta^1, \ldots, \theta^k) \in \mathbb{R}^k \Bigg|
\begin{array}{l}
\exists (\theta^{k+1}, \ldots, \theta^d) \in \mathbb{R}^{d-k} \hbox{ such that}  
\\
\sum_{i=1}^d u^i_{k+j} \partial_i F( U(\theta^1, \ldots, \theta^d)) =a_j \\
\hbox{ for } j=1, \ldots,d-k.
\end{array}
\Bigg\}.
\Label{const1-U}
\end{align}

When the mixture subfamily ${\cal M}$ is an exponential subfamily generated by $u_1, \ldots, u_k$,
we retake $\theta_0$ such that $(U^{-1} \theta_0)^i=0$ for $i=1, \ldots, k$.
Then, the subsets $\Theta_{{\cal M}}$ and $\Xi_{{\cal M}}$ are the same subsets defined in Subsection \ref{S2-1}.

\subsection{Bregman Divergence and $e$- and $m$- projections}
\begin{definition}[Bregman Divergence]
Let $\Theta$ be an open set in $\mathbb{R}^d$ and 
$F:\Theta \rightarrow \mathbb{R}$ be a $C^\infty$-class strictly convex function.
The Bregman divergence $D^F$ is defined by
\begin{equation}
    D^{F}(\theta_1 \| \theta_2):= 
    \langle \nabla^{(e)}[F](\theta_1), \theta_1 - \theta_2\rangle 
    - F(\theta_1)+F(\theta_2)      
    ~ (\theta_1, \theta_2 \in \Theta).
\Label{XZL}
\end{equation}
\end{definition}
We call the triplet $(\Theta,F,D^F)$ a Bregman divergence system.
In the one-parameter case, we have the following lemma.
\begin{lemma}\Label{NLO}
Assume that $d=1$.
$\frac{\partial }{\partial \theta_1} D^{F}(\theta_1 \| \theta_2)
=\frac{d^2}{d\theta^2}F(\theta_1)(\theta_1-\theta_2)$.
Hence, when $D^{F}(\theta_1 \| \theta_2)$
is monotonically increasing for $\theta_1$ in $(\infty, \theta_2]$,
and is
monotonically decreasing for $\theta_1$ in $(\theta_2,-\infty)$.
\end{lemma}

By using the Hesse matrix $J_{i,j}(\theta):=
\frac{\partial^2 F}{\partial \theta^i \partial \theta^j}(\theta)$, 
this quantity can be written as
\begin{equation}
    D^{F}(\theta_1 \| \theta_2)= 
\int_{0}^1 \sum_{i,j}(\theta_1^i - \theta_2^i)(\theta_1^j - \theta_2^j)
J_{i,j} 
(\theta_2 +t(\theta_1 - \theta_2))t dt.\Label{KPOT}
\end{equation}
This expression shows the inequality
\begin{align}
D^{F}(\theta_1 \| \theta_2) \ge     
D^{F}(\theta_1 \| \theta_2 +t(\theta_1 - \theta_2))
+ D^{F}(\theta_2 +t(\theta_1 - \theta_2) \| \theta_2)\Label{KPOT2}
\end{align}
for $t\in (0,1)$.

For an invertible matrix $U$, we have
\begin{equation}
    D^{F}(\theta_1 \| \theta_2)= 
    D^{F\circ U}(U^{-1}(\theta_1) \| U^{-1}(\theta_2)).\Label{COP}
\end{equation}
Since 
\begin{align}
\frac{\partial}{\partial \theta_2^i}\frac{\partial}{\partial \theta_2^j}
D^{F}(\theta_1 \| \theta_2)=
J_{i,j}
(\theta_2),
\Label{BJ1} \end{align}
$D^{F}(\theta_1 \| \theta_2)$ is convex function with respect to the second parameter $\theta_2$.

When $\theta_2$ is given as $\theta_1+ \Delta \theta $, and
the norm of $\Delta \theta$ is small,
The relation \eqref{KPOT} shows that
\begin{align}
D^{F}(\theta_1 \| \theta_1+ \Delta \theta)=
\sum_{i,j}
\frac{1}{2}J_{i,j}
(\theta_1)
(\Delta \theta)^i (\Delta \theta)^j+o( \|\Delta \theta\|^2).\Label{MLA}
\end{align}
Since the relations \eqref{du1} and \eqref{MN1} imply
\begin{align}
F^*(\eta)
=\sum_{i=1}^d\theta^i \eta(\theta_i)-F(\theta)
=\langle  \eta(\theta),\theta\rangle -F(\theta),
\end{align}
we have
\begin{align}
&    D^{F^*}(\nabla^{(e)} [F](\theta_2) \| \nabla^{(e)} [F](\theta_1))
=    D^{F^*}(\eta(\theta_2) \| \eta(\theta_1))\nonumber \\
=& \langle \eta(\theta_2)-\eta(\theta_1), \theta_2\rangle -F^*(\eta(\theta_2))
+F^*(\eta(\theta_1))  \nonumber \\
=& \langle \eta(\theta_1),\theta_1-\theta_2\rangle -F(\theta_1)+F(\theta_2)   
=    D^{F}(\theta_1 \| \theta_2).
\Label{XI1}
\end{align}
Therefore, when $\theta_2$ is fixed and 
$D^{F}(\theta_1 \| \theta_2)$ is a convex function for a  mixture parameter $\eta(\theta_1)$.
We define the matrix $J^*(\theta):= (J^{i,j,*}(\theta))_{i,j}$ as
\begin{align}
J^{i,j,*}(\theta):=\frac{\partial^2 F^*}{\partial \eta_i \partial \eta_j}(\eta)
\end{align}
with $\eta=\eta(\theta)$, which is the inverse matrix $J(\theta)^{-1}$ of $J(\theta)$. 
Applying the formula \eqref{KPOT} to $F^*$, we have
\begin{align}
&   D^{F}(\theta_1 \| \theta_2)
=    D^{F^*}(\eta(\theta_2) \| \eta(\theta_1)) \nonumber \\
=&
 \int_0^1 \sum_{i=1}^d \sum_{j=1}^d 
(\eta(\theta_2)-\eta(\theta_1))_i
(\eta(\theta_2)-\eta(\theta_1))_j
J^{i,j,*}(\theta(s)) s ds \Label{NXP},
\end{align}
where $\theta(s)$ is defined as
$\eta(\theta(s))= \eta(\theta_1)+ s (\eta(\theta_2)-\eta(\theta_1))$.
Similar to \eqref{KPOT2}, we have
\begin{align}
D^{F^*}(\eta(\theta_2) \| \eta(\theta_1))
=&D^{F}(\theta_2 \| \theta_1)
\ge
D^{F^*}(\eta(\theta_2) \| \eta(\theta(s)))+D^{F^*}(\eta(\theta(s)) \| \eta(\theta_1)) \nonumber\\
=&
D^{F}(\theta_2 \| \theta(s))+D^{F}(\theta(s) \| \theta_1). \Label{KPOT3}
\end{align}

In fact, when we restrict both inputs into an exponential subfamily ${\cal E}$, 
we have the following characterization.
That is, the restriction of the Bregman divergence system $(\Theta,F,D^F)$ to ${\cal E}$
can be considered as the Bregman divergence system $(\Theta_{{\cal E}},F_{{\cal E}},D^{F_{\cal E}})$
because we have
\begin{align}
D^{F}((\phi_{{\cal E}}^{(e)}(\bar{\theta}_1) \| (\phi_{{\cal E}}^{(e)}(\bar{\theta}_2))
=
D^{F_{{\cal E}}}(\bar{\theta}_1 \| \bar{\theta}_2)\Label{NBSO}
\end{align}
for $\bar{\theta}_1 , \bar{\theta}_2 \in \Theta_{{\cal E}}$.

Using a simple calculation, we can show the following proposition.

\begin{proposition}[Pythagorean Theorem \cite{Amari-Nagaoka}]\Label{MNL}
Let $\mathcal{E} \subset \Theta$ be an
exponential subfamily  generated by $l$ vectors $v_1,\ldots, v_l 
\in \mathbb{R}^d$ at $\theta_0 \in \Theta$, 
and 
$\mathcal{M} \subset \Theta$ be a mixture subfamily 
generated by the constraint $\sum_{i=1}^d v^i_j \eta_i(\theta)=a_j$ for 
$j=1, \ldots, l$.
Assume that an intersection $\theta^*$ of $\mathcal{E}$ and $\mathcal{M}$ exists.
For any $\theta \in \mathcal{E}$ and $\theta' \in \mathcal{M}$, we have
\begin{align}
D^F(\theta\|\theta')=D^F(\theta\|\theta^*)+D^F(\theta^*\|\theta').\Label{AKO9}
\end{align}
\end{proposition}

\begin{proof}
To show the relation \eqref{AKO9},
we choose an invertible matrix $U=(u_1 \ldots u_d )$ such that $u_i=v_i$ for $i=1, \ldots, l$.
Using the formula \eqref{COP}, we have
\begin{align}
&D^F(\theta\|\theta')=D^{F\circ U}(U^{-1}(\theta)\|U^{-1}(\theta'))\nonumber \\
=&\sum_{i=1}^d \frac{\partial }{\partial \theta^i} F\circ U(\theta) 
( (U^{-1}\theta)^i-(U^{-1} \theta')^i)
-F(\theta)+F(\theta') \nonumber \\
\stackrel{(a)}{=}&\sum_{i=1}^d \frac{\partial }{\partial \theta^i} F\circ U(\theta) 
((U^{-1}\theta)^i-(U^{-1}{\theta^*})^i)
-F(\theta)+F(\theta^*) \nonumber \\
&+\sum_{i=1}^l \frac{\partial }{\partial \theta^i} F\circ U(\theta^*) 
((U^{-1}{\theta^*})^i-(U^{-1}{\theta'})^i)
-F(\theta^*)+F(\theta') \nonumber \\
=&
D^F(\theta\|\theta^*)+D^F(\theta^*\|\theta'),\Label{AKO3}
\end{align}
where $(a)$ follows from the following facts;
Since $\theta^*$ and $\theta'$ belong to the same exponential family $\mathcal{E}$,
$(U^{-1}{\theta^*})^i=(U^{-1}{\theta'})^i$ for $i=l+1, \ldots, d$.
Since $\theta^*$ and $\theta$ belong to the same mixture family $\mathcal{M}$,
$\frac{\partial }{\partial \theta^i} F\circ U(\theta) 
=\frac{\partial }{\partial \theta^i} F\circ U(\theta^*) $
for $i=1, \ldots, l$.
\end{proof}

\begin{lemma}\Label{LA1}
Let $\mathcal{E}$
be an exponential family generated by $l$ vectors $v_1,\ldots, v_l 
\in \mathbb{R}^d$.
The following conditions are equivalent for 
the exponential subfamily $\mathcal{E}$, $\theta^* \in \mathcal{E}$,
and $\theta_0 \in \Theta$.
\begin{description}
\item[(E0)]
The element $\theta^*\in \mathcal{E}$ achieves a local minimum
for the minimization $\min_{\hat{\theta} \in \mathcal{E} }  D^F (\theta_0 \| \hat{\theta})$.
\item[(E1)]
The element $\theta^*\in \mathcal{E}$ achieves the minimum value  
for the minimization $\min_{\hat{\theta} \in \mathcal{E} }  D^F (\theta_0 \| \hat{\theta})$.
\item[(E2)]
Let $\mathcal{M} \subset \Theta$ be the mixture subfamily 
generated by the constraint $\sum_{i=1}^d v^i_j \eta_i(\theta)=
\sum_{i=1}^d v^i_j \eta_i(\theta_0)$ for $j=1, \ldots, l$.
The element $\theta^*\in \mathcal{E}$ belongs to
the intersection 
$\mathcal{M}\cap \mathcal{E}$.
\end{description}
Further, when there exists an element $\theta^*\in \mathcal{E}$ to satisfy the above condition,
such an element is unique. 
\end{lemma}

In the following, we denote the above mixture family $\mathcal{M}$ by $\mathcal{M}_{\theta_0\to \mathcal{E}}$.
Then, 
$\theta^* \in \mathcal{E}$ is called
the $e$-{\it projection} of $\theta$ onto an exponential subfamily $\mathcal{E}$,
and is denoted by $\Pro^{(e),F}_{\mathcal{E}} (\theta)$
because the points $\theta$ and $\theta^* $ are connected via the mixture family 
$\mathcal{M}_{\theta_0\to \mathcal{E}}$.
We call the minimum $\min_{\hat{\theta} \in \mathcal{E} }  D^F (\theta \| \hat{\theta})$
the projected Bregman divergence
between $\theta$ and $\mathcal{E}$.

\begin{proof}
Assume that (E0) holds. 
When an element $\hat{\theta}\in \mathcal{E}$ belongs to the neighbor hood of $\theta^*$,
we have
\begin{align}
&D^F(\theta_0\|\hat{\theta})-D^F(\theta_0\|{\theta}^*) \nonumber \\
=
&\sum_{i=1}^l \frac{\partial }{\partial \theta^i} F\circ U(\theta_0) 
((U^{-1}{\theta^*})^i-(U^{-1}{\hat{\theta}})^i)
-F(\theta^*)+F(\hat{\theta}) \nonumber  \\
=
&\sum_{i=1}^l 
\Big( \frac{\partial }{\partial \theta^i} F\circ U(\theta^*) 
-\frac{\partial }{\partial \theta^i} F\circ U(\theta_0) \Big)
((U^{-1}{\theta^*})^i-(U^{-1}{\hat{\theta}})^i) \nonumber \\
&+\sum_{i=1}^l \frac{\partial }{\partial \theta^i} F\circ U(\theta_0) 
((U^{-1}{\theta^*})^i-(U^{-1}{\hat{\theta}})^i)
-F(\theta^*)+F(\hat{\theta})  \nonumber \\
=
&\sum_{i=1}^l 
\Big( \frac{\partial }{\partial \theta^i} F\circ U(\theta^*) 
-\frac{\partial }{\partial \theta^i} F\circ U(\theta_0) \Big)
((U^{-1}{\theta^*})^i-(U^{-1}{\hat{\theta}})^i)\nonumber  \\
&+D^F(\theta^*\|\hat{\theta}).
\end{align}
In the following, assuming 
$(\frac{\partial }{\partial \theta^i} F\circ U(\theta^*) )_{i=1}^l
\neq 
(\frac{\partial }{\partial \theta^i} F\circ U(\theta_0) )_{i=1}^l$,
we derive the contradiction.
Since $\theta^*$ is an inner element of $\mathcal{E}$,
we choose 
an element $\hat{\theta}\in\mathcal{E}$ as
$\theta^*+ x \Delta \theta $ such that
$T:=\sum_{i=1}^l 
\Big( \frac{\partial }{\partial \theta^i} F\circ U(\theta^*) 
-\frac{\partial }{\partial \theta^i} F\circ U(\theta_0) \Big)
(\Delta \theta)^i <0$.
Then, due to \eqref{MLA}, the divergence $D^F(\theta^*\|\hat{\theta})$ behaves as the order 
$O(x^2)$.
Hence, choosing sufficiently small $x$, we have
$D^F(\theta_0\|\hat{\theta})-D^F(\theta_0\|{\theta}^*) 
=Tx + O(x^2)<0$, which implies contradiction.
Hence, we have 
$(\frac{\partial }{\partial \theta^i} F\circ U(\theta^*) )_{i=1}^l
=(\frac{\partial }{\partial \theta^i} F\circ U(\theta_0) )_{i=1}^l$,
which implies that 
$\theta^*$ is an intersection 
between $\mathcal{M}$ and $\mathcal{E}$.
Hence, (E2) holds.

Assume that (E2) holds.
Let  $\theta^*$ an intersection
between $\mathcal{M}$ and $\mathcal{E}$.
Then, the relation \eqref{AKO9} guarantees that
the element $\theta^*$ realizes the minimum 
$\min_{\hat{\theta} \in \mathcal{E} }  D^F (\theta_0 \| \hat{\theta})$.
Hence, (E1) holds.
Further, (E1) implies (E0).

When there are two different 
intersections 
between $\mathcal{M}$ and $\mathcal{E}$,
the above discussion and the relation \eqref{AKO9} guarantee
that the divergence between two intersections must be zero, which yields contradiction.
Thus, the intersection
between $\mathcal{M}$ and $\mathcal{E}$ 
should be unique.
\end{proof}

Exchanging the roles of the exponential family and the mixture family, we have the following lemma.

\begin{lemma}\Label{LA2}
We choose $l$ vectors $v_1,\ldots, v_l \in \mathbb{R}^d$.
Let $\mathcal{M}$
be an mixture family generated by 
generated by the constraint $\sum_{i=1}^d v^i_j \eta_i(\theta)=
\sum_{i=1}^d v^i_j \eta_i(\theta_0)$ for $j=1, \ldots, l$.
The following conditions are equivalent for 
the mixture subfamily $\mathcal{M}$, $\theta^{**} \in \mathcal{M}$,
and $\theta_0 \in \Theta$.
\begin{description}
\item[(M0)]
The element $\theta^{**}\in \mathcal{M}$ achieves a local minimum
for the minimization $\min_{\hat{\theta} \in \mathcal{M} }  D^F ( \hat{\theta}\|\theta_0 )$.
\item[(M1)]
The element $\theta^{**}\in \mathcal{M}$ achieves the minimum value  
for 
the minimization $\min_{\hat{\theta} \in \mathcal{M} }  D^F ( \hat{\theta}\|\theta_0 )$.
\item[(M2)]
Let $\mathcal{E} \subset \Theta$ be the mixture subfamily 
generated by $l$ vectors $v_1,\ldots, v_l \in \mathbb{R}^d$ at $\theta_0 \in \Theta$.
The element $\theta^{**}\in \mathcal{M}$ belongs to
the intersection 
$\mathcal{M}\cap \mathcal{E}$.
\end{description}
Further, when there exists an element $\theta^{**}\in \mathcal{M}$ to satisfy the above condition,
such an element is unique. 
\end{lemma}

In the following, we denote the above exponential family $\mathcal{E}$ by $\mathcal{E}_{\theta_0\to \mathcal{M}}$.
Then, 
$\theta^{**} \in \mathcal{M}$ is called
the {\it $m$-projection} of $\theta$ onto an mixture subfamily $\mathcal{M}$,
and is denoted by $\Pro^{(m),F}_{\mathcal{M}} (\theta)$
because the points $\theta$ and $\theta^{**} $ are connected via 
the exponential family $\mathcal{E}_{\theta_0\to \mathcal{M}}$.
When $ \mathcal{M}$ is an exponential subfamily and a mixture subfamily,
we can define both projections
$\Pro^{(e),F}_{\mathcal{M}}$ and $\Pro^{(m),F}_{\mathcal{M}}$, and these projections are different maps.
Hence, the subscripts $(e)$ and $(m)$ are needed.

\begin{lemma}\Label{Th5}
Let $\mathcal{E} \subset \Theta$ be an
exponential subfamily  generated by $l$ vectors $v_1,\ldots, v_l 
\in \mathbb{R}^d$ at $\theta_0 \in \Theta$.
For $\theta_* \in \Theta$,
the element $\Pro^{(e),F}_{\mathcal{E}} (\theta_*)=\theta^* \in {\cal E}$ 
is uniquely characterized as
$\sum_{j=1}^d v_i^j \partial_j F(\theta^*)
= \sum_{j=1}^d v_i^j \partial_j F(\theta_*)$, i.e.,
$R[V] \circ \nabla[F](\theta^*)=R[V] \circ \nabla[F](\theta_*)$.
That is, 
the mixture parameter of the element $\Pro^{(e),F}_{\mathcal{E}} (\theta_*)=\theta^* \in {\cal E}$ 
is given by the above condition.
\end{lemma}

\begin{proof}
We choose the mixture subfamily ${\cal M}$ 
generated by the constraint 
\begin{align}
\sum_{j=1}^d v^j_i \partial_j F(\theta)=
 \sum_{j=1}^d v_i^j \partial_j F(\theta_*)
\Label{MFA}
\end{align}
 for $i=1, \ldots,l$.
Due to Pythagorean theorem (Proposition \ref{MNL}),
the point $\theta^*$ is characterized by the intersection between
${\cal M}$  and ${\cal E}$. 
Hence, the constraint \eqref{MFA} for ${\cal M}$ guarantees the desired statement.
\end{proof}

\begin{lemma}\Label{Th6}
Let $l$ vectors $u_1,\ldots, u_d \in \mathbb{R}^d$ be linearly independent.
Let $\mathcal{M} \subset \Theta$ be a mixture subfamily 
generated by the constraint 
\begin{align}
\sum_{i=1}^d u^i_j \partial_i F (\theta)=a_j \Label{BO1}
\end{align}
for $j=k+1, \ldots, d$.
When the maximum $\max_{\theta\in {\cal M}} D^F(\theta\|\theta_{**})$ exists,
we obtain the following characterizations for $\Pro^{(m),F}_{\mathcal{M}} (\theta_{**})$.
\begin{description}
\item[(A1)]
The point $\Pro^{(m),F}_{\mathcal{M}} (\theta_{**})=\theta^{**} \in {\cal M}$ 
is uniquely characterized as
\begin{align}
(U^{-1}{\theta}^{**})^i
=(U^{-1}{\theta}_{**})^i\Label{MFA2}
\end{align}
 for $i=1, \ldots, k$,
 where $U$ is defined in the same way as Subsection \ref{S2-2}.
\item[(A2)]
We choose the exponential subfamily ${\cal E}$ 
generated by $d-k$ vectors $u_{k+1},\ldots, u_{d} \in \mathbb{R}^{d}$
at $\theta_{**}$.
The intersection between
${\cal M}$  and ${\cal E}$ is composed of 
the unique element $\Pro^{(m),F}_{\mathcal{M}} (\theta_{**})$. 

\item[(A3)]
The point $\Pro^{(m),F}_{\mathcal{M}} (\theta_{**})=\theta^{**} \in {\cal M}$ 
is uniquely characterized as
$\theta_{**}+ \sum_{j'=1}^{d-k} \bar{\tau}^{j'} u_{k+j'}$,
where $(\bar{\tau}^1, \ldots, \bar{\tau}^{d-k})$ is the unique element to satisfy 
\begin{align}
\frac{\partial}{\partial \tau^{j}} F \Big(\theta_{*}+ \sum_{j'=1}^{l} \tau^{j'} u_{k+j'} \Big) =a_j
\Label{const1-T}
\end{align}
for $j=1, \ldots, d-k$.

\end{description}
\end{lemma}

\begin{proof}
To characterize elements of ${\cal M}$, we employ the parameter $\bar\eta$ defined in \eqref{Dif}.
Then, the set ${\cal M}$ is given as $\{ (\bar\eta_1, \ldots, \bar\eta_k, a_1, \ldots, a_{d-k})|
(\bar\eta_1, \ldots, \bar\eta_k) \in \mathbb{R}^{k}\}$ under this parameterization.
Then, using \eqref{XI1}, we have
\begin{align}
& D^F( \phi_{\mathcal{M}}^{(m)}(\bar\eta_1, \ldots, \bar\eta_k, a_1, \ldots, a_{d-k}) \|  \theta_{**}) \nonumber \\
=& D^{ (F\circ U)^*}(\psi_{\mathcal{M}}^{(m)}(\theta_{**})
\|(\bar\eta_1, \ldots, \bar\eta_k, a_1, \ldots, a_{d-k}))
\end{align}
Since the map $(\bar\eta_1, \ldots, \bar\eta_k)\mapsto
D^{ (F\circ U)^*}(\psi_{\mathcal{M}}^{(m)}(\theta_{**})
\|(\bar\eta_1, \ldots, \bar\eta_k, a_1, \ldots, a_{d-k}))$ is smooth and convex,
the minimum
$\min_{ (\bar\eta_1, \ldots, \bar\eta_k)} 
D^{ (F\circ U)^*}( \psi_{\mathcal{M}}^{(m)}(\theta_{**})
\|(\bar\eta_1, \ldots, \bar\eta_k, a_1, \ldots, a_{d-k}))$
is realized when
\begin{align}
\partial^i(F\circ U)^*(\bar\eta_1, \ldots, \bar\eta_k, a_1, \ldots, a_{d-k})
=\partial^i(F\circ U)^*(\psi_{\mathcal{M}}^{(m)}(\theta_{**})).\Label{NA1}
\end{align}
 for $i=1, \ldots, k$.
Since \eqref{NA1} is equivalent to \eqref{MFA2} due to \eqref{CO1}, 
we obtain (A1).

The exponential subfamily ${\cal E}$ is characterized as
$\{\theta| (U^{-1} \theta)^{i}=(U^{-1}{\theta}_{**})^i \hbox{ for }
i= 1, \ldots, k\}$.
Then, we find that the intersection between
${\cal M}$  and ${\cal E}$ is not empty and contains $\theta^{**}$.
Further, 
when an element $\theta$ belongs to the intersection between ${\cal M}$  and ${\cal E}$,
the Pythagorean theorem (Proposition \ref{MNL}) guarantees that 
the element $\theta$ realizes the maximum $\max_{\theta\in {\cal M}} D^F(\theta\|\theta_{**})$.
Hence, the intersection between ${\cal M}$  and ${\cal E}$ is composed of 
the unique element $\Pro^{(m),F}_{\mathcal{M}} (\theta_{**})$. 
Hence, we obtain (A2).

Due to (A2), the unique element $\Pro^{(m),F}_{\mathcal{M}} (\theta_{**})$
is characterized as 
an element in ${\cal E}=
\{ \theta_{**}+ \sum_{j'=1}^{d-k} \tau^{j'} u_{k+j'}|$\par\noindent$ (\tau^1, \ldots, \tau^{d-k}) \in \mathbb{R}^{d-k} \}$
to satisfy \eqref{const1-T}.
Hence, we obtain (A3).
\end{proof}

Due to Lemmas \ref{LA1} and \ref{LA2},
it is important to find a sufficient condition for (E2) and (M2).
To discuss this issue for a convex function $F$ and $\Theta$,
we fix $l$ linearly independent vectors $v_1,\ldots, v_l \in \mathbb{R}^d$.
Then, we consider the following conditions;
\begin{description}
\item[(M3)]
We denote the exponential family generated by 
the $l$ linearly independent vectors $v_1,\ldots, v_l \in \mathbb{R}^d$ at $\theta_0 \in \Theta$ by $\mathcal{E}(\theta_0) $.
The $l$-dimensional parameter space $\Theta_{\mathcal{E}(\theta_0)}$ 
does not depend on $ \theta_0 \in \Theta$.
while the space 
$\Theta_{\mathcal{E}(\theta_0)}$ 
is defined in the way as \eqref{const1-UE}.
In this case, this set is denoted by $\Xi(v_1,\ldots, v_l)$.
\item[(E3)]
We denote the mixture family generated by the constraint
$\sum_{i=1}^d v^i_{j} \partial_i F(\theta) =a_j
$ for $j=1, \ldots,l$
by $\mathcal{M}(a_1, \ldots, a_l)$.
The $d-l$-dimensional parameter space
$\Theta_{\mathcal{M}(a_1, \ldots, a_l)}$ 
does not depend on $ (a_1, \ldots, a_l) \in \mathbb{R}^l$ unless
$\mathcal{M}(a_1, \ldots, a_l)$ is empty
while the space 
$\Theta_{\mathcal{M}(a_1, \ldots, a_l)}$ 
is defined in the way as \eqref{const1-U}.
In this case, this set is denoted by $\Theta(v_1,\ldots, v_l)$.
\end{description}

Under the above condition, we have the following lemmas. 
\begin{lemma}\Label{Lem7}
Assume that 
the $l$ linearly independent vectors $v_1,\ldots, v_l \in \mathbb{R}^d$
satisfy Condition (M3).
Given $(a_1, \ldots, a_l)\in \Xi(v_1,\ldots, v_l)$, we define the mixture family 
$\mathcal{M}(a_1, \ldots, a_l)$ by using the condition \eqref{BO1}.
Then, for $\theta_0 \in \Theta$,
the projected point $\Pro^{(m),F}_{\mathcal{M}(a_1, \ldots, a_l)} (\theta_{0}) $ exists.
\end{lemma}
\begin{proof}
When the assumption holds,
for $\theta_0 \in \Theta$,
the exponential family $\mathcal{E}(\theta_0)$ contains an element whose mixture parameter is $(a_1, \ldots, a_l)$.
Hence, due to Lemma \ref{LA2}, 
the exponential family $\mathcal{E}(\theta_0)$ and
the mixture family $\mathcal{M}(a_1, \ldots, a_l)$ have a unique intersection.
Therefore, 
the projected point $\Pro^{(m),F}_{\mathcal{M}(a_1, \ldots, a_l)} (\theta_{0}) $ exists
unless $\mathcal{M}(a_1, \ldots, a_l)$ is empty. 
\end{proof}

\begin{lemma}\Label{Lem8}
Assume that 
the $l$ linearly independent vectors $v_1,\ldots, v_l \in \mathbb{R}^d$
satisfy Condition (E3).
Then, for $ (b^1, \ldots, b^{d-l}) \in \mathbb{R}^{d-l}$ and
$\theta_0 \in \Theta$,
the projected point $\Pro^{(e),F}_{\mathcal{E}(b^1, \ldots, b^{d-l})} (\theta_{0}) $ exists
unless $\mathcal{E}(b^1, \ldots, b^{d-l})$ is empty
where the exponential family $\mathcal{E}(b^1, \ldots, b^{d-l})$ is defined
as $\{  (\sum_{i=1}^{d-l} u_i^j b^i+ \sum_{i=1}^{l} u_i^j \theta^i)_{j=1}^d |
(\theta^1,\ldots,\theta^{l}) \in \mathbb{R}^{l} 
\}\cap \Theta$. 
\end{lemma}

\begin{proof}
Assume that the assumption holds.
For $\theta_0 \in \Theta$,
we define the mixture family $\mathcal{M}(\theta_0)$ by using the constraint;
$\sum_{j=1}^d v^j_i \partial_j F(\theta) =\sum_{j=1}^d v^j_i \partial_j F(\theta_0) $ for $i=1, \ldots, l$.
Then, the mixture family $\mathcal{M}(\theta_0)$ contains an element whose 
natural parameter is $(b^1, \ldots, b^{d-l})$.
Hence, due to Lemma \ref{LA1}
the mixture $\mathcal{M}(\theta_0)$ and
the exponential family $\mathcal{E}(b^1, \ldots, b^{d-l})$ have a unique intersection.
Therefore, 
the projected point $\Pro^{(e),F}_{\mathcal{E}(b^1, \ldots, b^{d-l})} (\theta_{0}) $ exists
unless $\mathcal{E}(b^1, \ldots, b^{d-l})$ is empty.
\end{proof}

In addition, we introduce the following conditions for 
the Bregman divergence system $(\Theta,F,D^F)$.
\begin{description}
\item[(M4)]
Any $l$ linearly independent vectors $v_1,\ldots, v_l \in \mathbb{R}^d$
satisfy Condition (M3) for $l=1, \ldots, d-1$.
\item[(E4)]
Any $l$ linearly independent vectors $v_1,\ldots, v_l \in \mathbb{R}^d$
satisfy Condition (E3) for $l=1, \ldots, d-1$.
\end{description}
When (M4) holds, the $m$-projection 
$\Pro^{(m),F}_{\mathcal{M}}$ 
can be defined for any mixture subfamily $\mathcal{M}$.
Also, 
when (E4) holds, the $e$-projection 
$\Pro^{(e),F}_{\mathcal{E}}$ 
can be defined for any exponential subfamily $\mathcal{E}$.
Therefore, these two conditions are helpful for the analysis of these projections.

\begin{table*}[htb]
\begin{center}
\caption{Summary of dimensions}
\begin{tabular}{|c|c|} \hline
Symbol & Space \\
\hline 
$d$ & Dimension of the whose space \\ 
\hline 
$l$ & Dimension of Exponential family ${\cal E}$ \\
\hline 
$k$ & Dimension of Mixture family ${\cal M}$ \\
\hline 
\end{tabular}
\Label{table1}
\end{center}
\end{table*}

\subsection{Evaluation of Bregman divergence without Pythagorean theorem}\Label{SXNO}
Next, we evaluate Bregman divergence when we cannot use Pythagorean theorem.
For this aim, we focus on $J(\theta)^{-1}$, i.e.,
the inverse of the Hesse matrix $J(\theta)$ defined for the parameters of $\Theta$.
Then, we introduce the following quantity
$\gamma(\hat{\Theta} |{\Theta})$
for a subset $\hat{\Theta}$ of $\Theta$.
\begin{align}
\gamma(\hat{\Theta}|{\Theta}):= &
\inf\{\gamma|
\gamma J(\theta_1)^{-1}\ge J(\theta_2)^{-1}
\hbox{ for }\theta_1,\theta_2 \in \hat{\Theta} \} \\
\end{align}
We say that a subset $\hat{\Theta}$ of $\Theta$
is a {\it star subset} for an element $\theta_1 \in \hat{\Theta}$
when 
$\lambda \eta(\theta)+(1-\lambda)\eta(\theta_1) \in \eta(\hat{\Theta})$ for 
$\theta \in \hat{\Theta}$ and $\lambda \in (0,1)$.

Then, we have the following theorem.
\begin{theorem}\Label{XAM} 
We assume that Condition (M4) holds.
Then, for 
a star subset with $\hat{\Theta}$ for $ \theta_1 \in \hat{\Theta}$,
$\theta_2 \in \hat{\Theta}$, and $\theta_3 \in \Theta$, we have
\begin{align}
& D^F(\theta_1\|\theta_2) \nonumber \\
\le & 
D^F(\theta_1\|\theta_3)+\gamma (\hat{\Theta}|{\Theta}) D^F(\theta_2\|\theta_3)+ 
2 \gamma (\hat{\Theta}|{\Theta})\sqrt{ D^F(\theta_1\|\theta_3)D^F(\theta_2\|\theta_3)}.
\Label{BLT}
\end{align}
\end{theorem} 
The proof of Theorem \ref{XAM} is given in Appendix \ref{A1}.

\subsection{Bregman divergence system for mixture subfamily}\Label{APB}
When ${\cal E}$ is an exponential subfamily, 
the triplet 
$(\Theta_{{\cal E}},F_{{\cal E}},D^{F_{\cal E}})$ is 
a Bregman divergence system as explained in \eqref{NBSO}.
However, 
when ${\cal M}$ is a mixture subfamily and it is not an exponential subfamily, 
it is not so trivial to recover a Bregman divergence system.
We use the symbols defined in Subsection \ref{S2-2}.
Any element in ${\cal M}$ can be parameterized by an element $\bar\theta \in \Theta_{\cal M}$.
Therefore, there uniquely exists an vector $\kappa(\bar\theta) \in \mathbb{R}^{d-k}$ such that
$U( \bar\theta,  \kappa(\bar\theta)) \in {\cal M}$.
Then, we define the map 
$\phi_{\mathcal{M}}^{(e)}: \Theta_{{\cal M}}\to {\cal M}$
as $\phi_{\mathcal{M}}^{(e)}(\bar\theta):=U( \bar\theta,  \kappa(\bar\theta)) $,
and its inverse map 
$\psi_{{\cal M}}^{(e)}:=(\phi_{{\cal M}}^{(e)})^{-1}: {\cal M}\to \Theta_{{\cal M}}$.

A convex function $F_{\mathcal{M}}(\bar\theta)$
is defined as
\begin{align}
F_{\mathcal{M}}(\bar\theta)
:=&(F\circ U)( \bar\theta,  \kappa(\bar\theta)) - \sum_{i=k+1}^d \partial_i (F\circ U)( \bar\theta,  \kappa(\bar\theta)) \kappa^{i-k}(\bar\theta) \nonumber \\
=&(F\circ U)( \bar\theta,  \kappa(\bar\theta)) - \sum_{i=k+1}^d a_i \kappa^{i-k}(\bar\theta).
\Label{XPA}
\end{align}
$(F\circ U)^*|_{\Xi_{{\cal M}}}$ is a convex function.
Due to \eqref{NAY},
the Legendre transform of $(F\circ U)^*|_{\Xi_{{\cal M}}}$ is $F_{\mathcal{M}}(\bar\theta)$.
Hence, $F_{\mathcal{M}}(\bar\theta)$ is a convex function.

Also, we have
\begin{align}
&\frac{\partial}{\partial {\bar\theta}^j } 
F_{\mathcal{M}}(\bar\theta)\nonumber \\
=&
\partial_j (F\circ U) ( \bar\theta,  \kappa(\bar\theta))
+\sum_{i=k+1}^d  \partial_i (F\circ U) ( \bar\theta,  \kappa(\bar\theta)) \partial_j \kappa^{i-k}(\bar\theta) 
-\sum_{i=k+1}^d a_i  \kappa^{i-k}(\bar\theta)\nonumber  \\
=&\partial_j ( F\circ U) ( \bar\theta,  \kappa(\bar\theta)).
\end{align}
Thus,
\begin{align}
&D^{F_{\mathcal{M}}}(\bar\theta_1\|\bar\theta_2)\nonumber \\
=&
F_{\mathcal{M}}(\bar\theta_1) -F_{\mathcal{M}}(\bar\theta_2)-
\sum_{j=1}^{k} \frac{\partial}{\partial {\bar\theta}^j } 
F_{\mathcal{M}}(\bar\theta) (\bar\theta_1^j-\bar\theta_2^j) \nonumber \\
=&
F\circ U(\bar\theta_1, \kappa(\bar\theta_1)) -F\circ U (\bar\theta_2,\kappa(\bar\theta_2))\nonumber \\
&-
\sum_{j=1}^{d} 
\partial_j (F\circ U) ((\bar\theta_1, \kappa(\bar\theta_1))^j-(\bar\theta_2, \kappa(\bar\theta_2))^j) \nonumber \\
=& D^{F\circ U}((\bar\theta_1, \kappa(\bar\theta_1))\|(\bar\theta_2, \kappa(\bar\theta_2))) 
=D^{F}(U(\bar\theta_1, \kappa(\bar\theta_1))\| U(\bar\theta_2, \kappa(\bar\theta_2))) \nonumber \\
=& D^{F}(
\phi_{\mathcal{M}}^{(e)}(\bar\theta_1)\|
\phi_{\mathcal{M}}^{(e)}(\bar\theta_2)) .
\end{align}
Therefore,
the Bregman divergence in the Bregman divergence system $(\Theta_{{\cal M}}, F_{\mathcal{M}}, D^{F_{\mathcal{M}}})$
equals the Bregman divergence in the Bregman divergence system $(\Theta, F, D^{F})$
for two elements in ${\cal M}$.

A subset $\mathcal{E} \subset \mathcal{M}$ 
is called an $l$-dimensional  {\it exponential subfamily} of ${\cal M}$
generated by $l$ linearly independent vectors $v_1,\ldots, v_l \in \mathbb{R}^k$
at $\theta_0 \in \Theta_{\mathcal{M}}$ with $l \le k$
when the subset $\mathcal{E}$ is given as
\begin{align}
    \mathcal{E} &= \left\{ \left. \phi_{{\cal M}}^{(e)} \Big({\theta}_0+ \sum_{i=1}^l \bar{\theta}^i v_i\Big)  
    \right| 
    \bar{\theta} \in \mathbb{R}^{l}    \right\}\cap \mathcal{M} .
\end{align}

A subset $\mathcal{M}_1 \subset \Theta$ is called an $l$-dimensional  {\it mixture subfamily} of ${\cal M}$
generated by the additional constraints 
\begin{align}
\sum_{i=1}^k v^i_{j} \eta_i =a_j\Label{const9}
\end{align}
 for $j=1, \ldots,l$ with $v_1,\ldots, v_l \in \mathbb{R}^k$
when the subset $\mathcal{M}_1$ is written as
\begin{align}
    \mathcal{M}_1 = 
    \left\{ \left.\phi_{\mathcal{M}}^{(m)}(\eta)  \in \mathcal{M} \right| 
\eta \in \Xi_{\mathcal{M}} \hbox{ satisfies Condition \eqref{const9}.}  \right\} 
\end{align}

\subsection{Closed convex mixture subfamily}\Label{4Y}
A closed subset $\mathcal{M}$ of a mixture subfamily $\hat{\mathcal{M}}$
is called a {\it closed mixture subfamily}.
The mixture subfamily $\hat{\mathcal{M}}$ is called the 
{\it extended mixture family} of $\mathcal{M}$
when $\hat{\mathcal{M}}$ and $\mathcal{M}$ have the same dimension.
When a closed mixture subfamily $\mathcal{M}$ is a convex set with respect to the mixture parameter, it is called a {\it closed convex mixture subfamily}.

We define the boundary set $\partial\mathcal{M}:=
\mathcal{M}\setminus \interior {\mathcal{M}}$, where $\interior {\mathcal{M}}$ is the 
interior of $\mathcal{M}$.
For an element $\theta \in \partial\mathcal{M}$, 
a $d-1$-dimensional mixture family $\mathcal{M}'$ is called 
a tangent space of $\mathcal{M}$ at $\theta$
when 
$\mathcal{M}'\cap \mathcal{M}\neq \emptyset$
and 
$\mathcal{M}'\cap \interior {\mathcal{M}}= \emptyset$.
When $\mathcal{M}$ is a closed convex mixture subfamily,
any element $\theta \in \partial\mathcal{M}$ has a tangent space.
When $\mathcal{M}$ is composed of one element, 
we consider that $\hat{\mathcal{M}}:=\mathcal{M}$,
$\partial \mathcal{M}=\emptyset$, and $\interior \mathcal{M}= \mathcal{M}$.

\begin{lemma}\Label{LMGO}
Assume that the Bregman divergence system $(\Theta,F,D^F)$ satisfies 
Condition (M4).
For any element $\theta \in \Theta$ and any closed convex mixture subfamily 
$\mathcal{M}$,
there uniquely exists a minimum point 
\begin{align}
\Pro^{(m),F}_{\mathcal{M}} (\theta):=
\argmin_{\theta' \in \mathcal{M} } D^F(\theta'\|\theta).
\end{align}
In addition, any element $\theta' \in \mathcal{M}$ satisfies the inequality
\begin{align}
D^F(\theta'\| \theta)\ge 
D^F(\theta'\| \Pro^{(m),F}_{\mathcal{M}} (\theta))
+D^F(\Pro^{(m),F}_{\mathcal{M}} (\theta)\| \theta).\Label{MGO2}
\end{align}
Further, we denote 
the extended mixture family of ${\cal M}$ by $\hat{{\cal M}}$. 
When $\theta$ belongs to $\hat{{\cal M}}\setminus {\cal M}$,
then,
\begin{align} 
\Pro^{(m),F}_{\mathcal{M}} (\theta)\in \partial {\cal M}.
\Label{MGOF}
\end{align}
\end{lemma}

\begin{proof}
\noindent{\bf Step 1:}\quad
We choose a sequence $\theta^{(n)} \in \partial\mathcal{M}$
such that
\begin{align}
\lim_{n \to \infty}
D^F(\theta^{(n)}\|\theta)=
\inf_{\theta' \in \mathcal{M} } D^F(\theta'\|\theta).
\end{align}
Since 
$\{ \theta' \in \mathcal{M} | D^F(\theta'\|\theta)\le
D^F(\theta^{(1)}\|\theta)\}$ is a compact subset,
there exists a subsequence of $n_m$ such that 
$\theta^{(n_m)}$ converges.
Since $\mathcal{M} $ is a closed subset,
$\theta^*:= \lim_{m \to \infty}\theta^{(n_m)}$ belongs to 
$\mathcal{M} $.

We define the vector $v:=(\theta^*-\theta) \in \mathbb{R}^d$
and the real numbers 
$a^{*}:=
\sum_{i=1}^d v^i \partial_i F(\theta^{*})\in \mathbb{R}$
and 
$b^{*}:=
\sum_{i=1}^d v^i \partial_i F(\theta)\in \mathbb{R}$.
When $a^{*}> b^{*}$,
as shown in Step 2, any element $\theta' \in \mathcal{M}$ satisfies 
\begin{align}
\sum_{i=1}^d v^i \partial_i F(\theta')
\ge a^*.\Label{JO2}
\end{align}
Otherwise, 
any element $\theta' \in \mathcal{M}$ satisfies 
\begin{align}
\sum_{i=1}^d v^i \partial_i F(\theta')
\le a^*.\Label{JO3}
\end{align}

\noindent{\bf Step 2:}\quad
We show only \eqref{JO2} by contradiction because
the relation \eqref{JO3} can be shown in the same way.
We choose an element $\theta' \in \mathcal{M}$ such that
\eqref{JO2} does not hold.
We denote the mixture parameters of $\theta^* $ and $\theta'$
by $\eta^*$ and $\eta'$.
Since $\mathcal{M}$ is convex with respect to the mixture parameter,
$\theta( \eta^*+t( \eta'-\eta^*))$ belongs to $\mathcal{M}$ for $t \in [0,1]$.

We denote the one-dimensional exponential subfamily 
$\{  \theta+ t(\theta^*-\theta)  \}_{t \in \mathbb{R}} \cap \Theta$
by $\mathcal{E}_1$.
We define $a(t):= \sum_{i=1}^d v^i \partial_i F(\theta( \eta^*+t( \eta'-\eta^*)))$.
We denote the $d-1$-dimensional mixture subfamily 
$\{\theta''\in \Theta|\sum_{i=1}^d v^i \partial_i F(\theta'')=a(t) \}$
by $\mathcal{M}(t)$.
Condition (M4) guarantees that the intersection $\mathcal{M}(t)\cap \mathcal{E}_1$
is composed of only one element. We denote the element by 
$\theta(t)$. 
Then, we have
\begin{align}
D^F(\theta( \eta^*+t( \eta'-\eta^*)) \| \theta)=
D^F(\theta( \eta^*+t( \eta'-\eta^*)) \| \theta(t))+
D^F(\theta(t) \| \theta)
\Label{JO2-L}.
\end{align}
We assume that $t>0$ is sufficiently small. 
Since \eqref{JO2} does not hold, the formula \eqref{KPOT} implies 
$D^F(\theta^* \| \theta)-D^F(\theta(t) \| \theta)=O(t)$.
However, 
$D^F(\theta( \eta^*+t( \eta'-\eta^*)) \| \theta(t))=O(t)$.
The combination of these relations shows that
\begin{align}
D^F(\theta( \eta^*+t( \eta'-\eta^*)) \| \theta)
< D^F(\theta^* \| \theta),
\end{align}
which yields the contradiction.

\begin{figure}[htbp]
\begin{center}
  \includegraphics[width=0.7\linewidth]{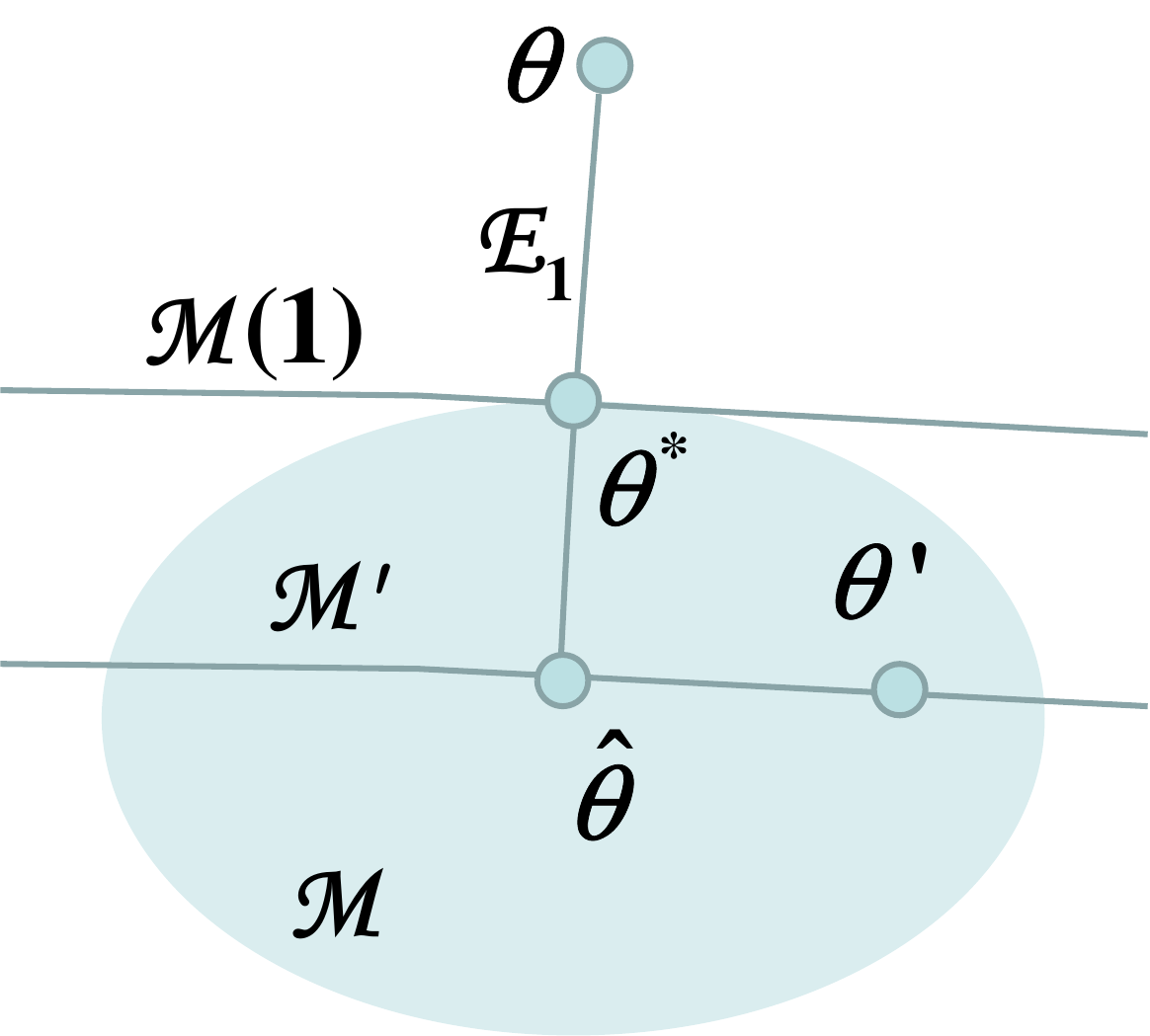}
  \end{center}
\caption{Figure for Step 3 of the proof of Lemma \ref{LMGO}.}
\Label{Fproof}
\end{figure}   

\noindent{\bf Step 3:}\quad
We show the uniqueness of the minimum point and \eqref{MGO2}
only in the case when $a^{*}> b^{*}$.
We can show the other case in the same way.

We define $a':= \sum_{i=1}^d v^i \partial_i F(\theta')$.
We denote the $d-1$-dimensional mixture subfamily 
$\{\theta''\in \Theta|\sum_{i=1}^d v^i \partial_i F(\theta'')=a' \}$
by $\mathcal{M}'$.
Condition (M4) guarantees that the intersection $\mathcal{M}'\cap \mathcal{E}_1$
is composed of only one element. We denote the element by 
$\hat{\theta}$. 
Hence, we have
\begin{align}
D^F(\theta'\| \theta)
\stackrel{(a)}{=}&
D^F(\theta'\| \hat{\theta})+D^F(\hat{\theta}\| \theta)\nonumber  \\
\stackrel{(b)}{\ge}&
D^F(\theta'\| \hat{\theta})
+D^F(\hat{\theta}\| \theta^*) 
+D^F(\theta^*\| \theta)\nonumber  \\
\stackrel{(c)}{=}&
D^F(\theta'\| \theta^*)
+D^F(\theta^*\| \theta),\Label{MGO}
\end{align}
where $(a)$ and $(c)$ follow from Proposition \ref{MNL}, 
and $(b)$ follows from \eqref{KPOT2}.
The relation \eqref{MGO} implies that
$D^F(\theta' \| \theta)> D^F(\theta^*\| \theta)$.
Hence, we obtain the uniqueness of the minimum point.
Also, \eqref{MGO}  implies \eqref{MGO2}.

\noindent{\bf Step 4:}\quad
We show \eqref{MGOF} by contradiction
only in the case when $a^{*}> b^{*}$
because we can show the other case in the same way.
We assume that \eqref{MGOF} does not hold.
We parameterize the exponential family ${\cal E}_1$ as $\{\theta_t\}$ such that
$\theta_0=\theta$ and $\theta_1=\theta_*=
\Pro^{(m),F}_{\mathcal{M}} (\theta)$.
Since $\theta_1 \in \int {\cal M}$,
there is an element $t_0 \in (0,1)$ such that 
$\theta_t \in \partial {\cal M}$.
Hence, Lemma \ref{NLO} guarantees that 
$D^F(\theta_t\| \theta_0)>  D^F(\theta_1\| \theta_0)$,
which contradicts that $\theta_1=\Pro^{(m),F}_{\mathcal{M}} (\theta)$.
\end{proof}

We say that a set of closed convex mixture subfamilies 
$\{ {\mathcal{M}}_{\lambda} \}_{\lambda \in \Lambda}$
covers the boundary $\partial \mathcal{M}$ of a closed convex mixture family $\mathcal{M}$
with subsets $\Lambda_\lambda \subset \Lambda$
and $\lambda \in \Lambda_*:=\Lambda \cup \{0\}$
when 
the following two conditions hold;
The relation 
\begin{align}
\partial {\mathcal{M}}_{\lambda}
=\cup_{\lambda'\in \Lambda_{\lambda}} {\mathcal{M}}_{\lambda'}
\end{align}
holds unless $\partial {\mathcal{M}}_{\lambda} =\emptyset$.
That is, when $\partial {\mathcal{M}}_{\lambda} =\emptyset$,
$\Lambda_{\lambda}$ is the empty set.
The relation 
\begin{align}
{\mathcal{M}}_{\lambda'}
\not\subset {\mathcal{M}}_{\lambda''}
\end{align}
holds for two elements $\lambda',\lambda''\in \Lambda_{\lambda}$.
That is, $0 \in \Lambda_*$ is considered as the index to express ${\cal M}$.
Hence, we define ${\cal M}_0:={\cal M}$.

Also, we define the subset $\bar{\Lambda}_{\lambda}:= 
\{\lambda' \in \Lambda| \exists 
\lambda_2, \ldots, \lambda_{x-1} \hbox{ such that }
\lambda_{i+1} \in \Lambda_{\lambda_i}
\hbox{ with }\lambda_1=\lambda, 
\lambda_x=\lambda' \}.$
In addition, we define the depth $D(\lambda)$ 
of an element $\lambda \in \Lambda$ as follows.
The depth $D(\lambda)$ of an element $\lambda$ is zero
when $\Lambda_{\lambda}$ is the empty set.
Otherwise, 
the depth $D(\lambda)$ of an element $\lambda$ is defined as
$1+\max_{\lambda' \in\Lambda_{\lambda} } D(\lambda')$.
Then, the depth of $\mathcal{M}$ is defined to be
the depth $D(0)$.

\begin{lemma}\Label{CLP}
The sets $\interior M_{\lambda}$ are disjoint, i.e.,
\begin{align}
\interior M_{\lambda} \cap \interior M_{\lambda'}= \emptyset
\hbox{ for }\lambda\neq \lambda' \in \Lambda_*.\Label{CMT}
\end{align}
Also, we have
\begin{align}
\partial \mathcal{M} = \cup_{\lambda' \in 
{\Lambda}} \interior \mathcal{M}_{\lambda'}.\Label{CMT2}
\end{align}
\end{lemma}
\begin{proofof}{Lemma \ref{CLP}}
We show the following statement by induction for 
depth $D(\lambda)$;
The relations 
\begin{align}
\interior \mathcal{M}_{\lambda'} \cap 
\interior \mathcal{M}_{\lambda''}
&= \emptyset
\hbox{ for }\lambda'\neq \lambda'' \in \Lambda_{\lambda}.\Label{CMT9} \\
\partial \mathcal{M}_{\lambda} &= \cup_{\lambda'  \in \bar\Lambda_{\lambda}} 
\interior \mathcal{M}_{\lambda'}.\Label{CMT10}
\end{align}

The relations \eqref{CMT9} and \eqref{CMT10} are trivial when
$D(\lambda)=0$.
In the following, we show 
the relations \eqref{CMT9} and \eqref{CMT10} for 
$D(\lambda)=k$
when they hold for $D(\lambda)\le k-1$.

The convexity of $\mathcal{M}_\lambda$ guarantees that
$\mathcal{M}_{\lambda'}\cap \mathcal{M}_{\lambda''}
\in \partial \mathcal{M}_{\lambda'},\partial \mathcal{M}_{\lambda''}$
for $\lambda',\lambda'' \in \Lambda_\lambda$.
Hence, we have \eqref{CMT9}.
For $\lambda' \in \Lambda_{\lambda}$,
the assumption of induction implies
\begin{align}
\partial \mathcal{M}_{\lambda'}
= \cup_{\lambda'' \in \bar{\Lambda}_{\lambda'}} 
\interior \mathcal{M}_{\lambda''}.\Label{CMT8}
\end{align}
Thus,
\begin{align}
\partial \mathcal{M}_{\lambda}
=&
\cup_{\lambda'\in \Lambda_{\lambda}} {\mathcal{M}}_{\lambda'}
=\cup_{\lambda'\in \Lambda_{\lambda}} 
\big(\interior {\mathcal{M}}_{\lambda'} \cup \partial {\mathcal{M}}_{\lambda'}\big) \nonumber \\
=& 
\cup_{\lambda'\in \Lambda_{\lambda}} 
\Big(\interior {\mathcal{M}}_{\lambda'} \cup 
\big(\cup_{\lambda'' \in \bar{\Lambda}_{\lambda'}} 
\interior \mathcal{M}_{\lambda''} \big)\Big) \nonumber \\
=& \cup_{\lambda' \in 
\bar{\Lambda}_{\lambda}} \interior \mathcal{M}_{\lambda'}.\Label{CMT6}
\end{align}
\end{proofof}

Under the above case, 
the point $\Pro^{(m),F}_{\mathcal{M}} (\theta)$ for $\theta \in \Theta$ can be characterized as follows.

\begin{lemma}\Label{LO10}
Assume that a set of closed convex mixture subfamilies 
$\{ {\mathcal{M}}_{\lambda} \}_{\lambda \in \Lambda}$
covers the boundary $\partial \mathcal{M}$ of a closed convex mixture family $\mathcal{M}$
with subsets $\Lambda_\lambda \subset \Lambda$
and $\lambda \in \Lambda_*:=\Lambda \cup \{0\}$.
We denote the extended mixture subfamily of 
${\mathcal{M}}_{\lambda}$ by $\hat{\mathcal{M}}_{\lambda}$
for $\lambda \in \Lambda_*$.
For $\theta \in \Theta$, we define $\lambda_0:= \argmin_{\lambda '
\in \Lambda_*} 
\{D^F(\Pro^{(m),F}_{\hat{\mathcal{M}}_{\lambda'}} (\theta) \|  \theta)|
\Pro^{(m),F}_{\hat{\mathcal{M}}_{\lambda'}} (\theta) \in
\interior \mathcal{M}_{\lambda'}
\}$. 
Then, we have
$\Pro^{(m),F}_{\mathcal{M}} (\theta)=
\Pro^{(m),F}_{\hat{\mathcal{M}}_{\lambda_0}} (\theta)$.
\end{lemma}
\begin{proof}
Due to Lemma \ref{CLP}, 
there uniquely exists $\lambda_0 \in \Lambda$ such that
$\Pro^{(m),F}_{\mathcal{M}} (\theta) \in \interior \mathcal{M}_{\lambda_0}$.
Then,
$\Pro^{(m),F}_{\mathcal{M}} (\theta)=\argmin_{\theta' \in \interior \mathcal{M}_{\lambda_0}} D^F(\theta' \| \theta )$.
Since Lemma \ref{LA2} guarantees that 
$\argmin_{\theta' \in \interior \mathcal{M}_{\lambda_0}} 
D^F(\theta' \| \theta )=\argmin_{\theta' \in  \hat{\mathcal{M}}_{\lambda_0}} D^F(\theta' \| \theta )$,
we have 
$\Pro^{(m),F}_{\mathcal{M}} (\theta)=\Pro^{(m),F}_{\hat{\mathcal{M}}_{\lambda_0}} (\theta)$.

When $\lambda \in \Lambda$ satisfies the condition 
$\Pro^{(m),F}_{\hat{\mathcal{M}}_{\lambda}} (\theta) \in
\interior \mathcal{M}_{\lambda}$,
we have 
$D^F(\Pro^{(m),F}_{\hat{\mathcal{M}}_{\lambda_0}} (\theta) \|  \theta)
=D^F(\Pro^{(m),F}_{\mathcal{M}} (\theta) \|  \theta)
\le D^F(\Pro^{(m),F}_{\hat{\mathcal{M}}_{\lambda}} (\theta) \|  \theta)$
because $\Pro^{(m),F}_{\hat{\mathcal{M}}_{\lambda}} (\theta) \in
\mathcal{M}$.
Hence, we obtain the desired statement.
\end{proof}

\section{Examples of Bregman divergence}\Label{S3}
\subsection{Classical system}\Label{4A}
We consider the set of probability distributions on the finite set ${\cal X}=\{1, \ldots, n\}$.
We focus on $d$ linearly independent functions $f_1, \ldots, f_{d}$ defined on ${\cal X}$,
where the linear space spanned by $f_1, \ldots, f_{d}$ does not contain a constant function and $d \le n-1$.
Then, we define the $C^{\infty}$ strictly convex function 
$\mu$ on $\mathbb{R}^{d}$ as
$\mu(\theta):= \log (\sum_{x \in {\cal X}} \exp (  \sum_{j=1}^{d} \theta^j f_j(x)  )$, which yields
the Bregman divergence system $(\mathbb{R}^{d}, \mu, D^\mu)$. 
When $d=n-1$,
any probability distribution with full support on ${\cal X}$
can be written as $P_\theta$, which is defined as
$P_\theta(x):= 
\exp \Big(  (\sum_{j=1}^{n-1} \theta^j f_j(x)  ) - \mu(\theta) \Big)$.
It is known that the KL divergence 
equals the Bregman divergence of the potential function $\mu$ \cite[Section 3.4]{Amari-Nagaoka}, 
i.e., we have
\begin{align}
D^{\mu}(\theta\|\theta')= D(P_\theta\|P_{\theta'})
\Label{MGA}
\end{align}
for $\theta \in \mathbb{R}^{d}$, where
the KL divergence $D(q\| p)$ is defined as
\begin{align}
D(q\| p)=\sum_{\omega} p(\omega)
(\log p(\omega) - \log q(\omega)).
\end{align}

\begin{example}\Label{ALO}
When $d=n-1$, 
the Bregman divergence system $(\mathbb{R}^{d}, \mu, D^\mu)$ 
describes the set ${\cal P}_{\X}$ 
of distributions on $\X$ with full support and
the KL divergence.
\end{example}

\begin{example}
When $\X$ is given as $\X_1 \times \X_2$ with $n_i=|\X_i|$,
$f_i$ is a function on $\X_1$ or $\X_2$, and $d=n_1+n_2-2$,
the Bregman divergence system $(\mathbb{R}^{d}, \mu, D^\mu)$ 
describes the set ${\cal P}_{\X_1}\times {\cal P}_{\X_2}$ of independent distributions on  $\X_1 \times \X_2$.
\end{example}

\begin{example}
When $\X$ is given as $\X_1 \times \X_2\times \X_3$ with $n_i=|\X_i|$,
$f_i$ is a function on $\X_1,\X_2$ or $\X_2,\X_3$, and $d=n_2 (n_1+n_3-2)+n_2-1$,
the Bregman divergence system $(\mathbb{R}^{d}, \mu, D^\mu)$ 
describes the set ${\cal P}_{X_1-X_2-X_3}$ of distributions on 
$\X_1 \times \X_2 \times X_3$ to satisfy the Markovian condition $X_1-X_2-X_3$.
\end{example}

When the parameter $\theta$ is limited to $(\bar\theta,\underbrace{0, \ldots, 0}_{d-l})$ with 
$\bar\theta \in \mathbb{R}^l$,
the set of distributions $P_\theta$ forms an exponential subfamily.
Also, when the linear space spanned by $d-k$ linearly independent functions $g_1, \ldots, g_{d-k}$
does not contain a constant function,
for $d-k$ constants $a_1, \ldots, a_{d-k}$, 
the following set of distributions forms a mixture subfamily; 
\begin{align}
\Big\{P_\theta \Big|
\sum_{x \in {\cal X}}g_i(x) P_\theta(x)=a_i \hbox{ for } i=1, \ldots, d-k \Big\}.
\end{align}

When we make linear constraints as explained in Subsection \ref{APB},  
changing the potential function ${\mu}$ in the way as \eqref{XPA},
we can recover \eqref{MGA}.

For the possibility of the projection, we have the following lemma.
\begin{lemma}\Label{LOS}
The Bregman divergence system $(\mathbb{R}^{d}, \mu, D^\mu)$ defined in this subsection satisfies
Conditions (E4) and (M4).
\end{lemma}

To show this lemma, we prepare the following lemma.
\begin{lemma}\Label{LOS2}
For $(\theta^1, \ldots, \theta^{d-l})\in \mathbb{R}^{d-l}$
and $\xi:=(\xi^{1}, \ldots, \xi^{l})\in \mathbb{R}^{l}$,
we define 
\begin{align}
\tau_{(\theta^1, \ldots, \theta^{d-l})}(\xi):=
\Bigg(\frac{\sum_{x \in \mathcal{X}}
f_{d-j}(x) \exp \Big(  \sum_{i=1}^{d} \theta^i f_{i}(x)  \Big) }{\mu(\theta)}
\Bigg)_{j=1}^{l} \in \mathbb{R}^l
\end{align}
with $\theta^{d-l+i}=\xi^i$.
Then, the set $\mathcal{T}_{1,(\theta^1, \ldots, \theta^{d-l})}:=
\{ \tau_{(\theta^1, \ldots, \theta^{d-l})}(\xi) | \xi \in \mathbb{R}^{l} \}$
equals the inner $\mathcal{T}_2$ of the convex full of 
$\{ (f_{d-j}(x))_{j=1}^{l}\}_{x \in \mathcal{X}}$.
\end{lemma}
\begin{proofof}{Lemma \ref{LOS2}}
In this proof, $\mathcal{T}_{1,(\theta^1, \ldots, \theta^{d-l})}$ 
and $\tau_{(\theta^1, \ldots, \theta^{d-l})}(\xi)$
are simplified to $\mathcal{T}_1$ and $\tau(\xi)$.
Since $\mathcal{T}_1\subset \mathcal{T}_2$ is trivial, we show only the opposite relation.

\noindent{\bf Step 1:}\quad
Any element in the boundary of the convex full of $\{ (f_{d-j}(x))_{j=1}^{l}\}_{x \in \mathcal{X}}$
is written as $\Big(\sum_{i=1}^{l'} p_i f_{d-j}(x_i)\Big)_{j=1}^{l}$ with extremal points $(f_{d-j}(x_i))_{j=1}^{l}$
with at most $l$ elements $x_i \in \mathcal{X}$ and at most  $l$ positive numbers $p_i$, where $i=1, \ldots, l'\le l$.
There exists an element $\xi_* \in \mathbb{R}^{l}$ such that
$ \max_{x \in \mathcal{X}} \sum_{j=1}^l \xi_*^j f_{d-j}(x)=
\sum_{j=1}^l \xi_*^j f_{d-j}(x_i)=1$ and 
$\sum_{j=1}^l \xi_*^j f_{d-j}(x)<1$ unless 
$(f_{d-j}(x))_{j=1}^{l}$ is written as a convex combination of 
$\{(f_{d-j}(x_i))_{j=1}^{l}\}_{i=1}^{l'}$.
For any $x_i$, there exists an element $\xi(x_i)_* \in \mathbb{R}^{l}$ such that
$ \max_{x \in \mathcal{X}} \sum_{j=1}^l \xi(x_i)_*^j f_{d-j}(x)=
\sum_{j=1}^l \xi(x_i)_*^j f_{d-j}(x_i)=1$,
$\sum_{j=1}^l \xi(x_i)_*^j f_{d-j}(x)<1$ for $x \neq x_i$
and
$\sum_{j=1}^l \xi(x_i)_*^j f_{d-j}(x_{i'})>0$ for $i' \neq i$.
Then, there exist elements $t_i>0$ such that
\begin{align}
\frac{\exp\Big(\sum_{i'=1}^{l'}\sum_{j=1}^l t_{i'}\xi(x_{i'})_*^j f_{d-j}(x_i) + \sum_{j=1}^{d-l}\theta^j f_j(x_i)   \Big)}
{\sum_{i''=1}^{l'}
\exp\Big(\sum_{i'=1}^{l'}\sum_{j=1}^l t_{i'}\xi(x_{i'})_*^j f_{d-j}(x_{i''}) + \sum_{j=1}^{d-l}\theta^j f_j(x_{i''})   \Big)}
=p_i.
\end{align}
Hence, we have
\begin{align}
\tau \Big(t \xi(x_0) + \sum_{i'=1}^{l'} t_{i'}\xi(x_{i'})_* \Big)
\to \Big(\sum_{i=1}^{l'} p_i f_{d-j}(x_i)\Big)_{j=1}^{l}
\end{align}
as $t \to \infty$.

\noindent{\bf Step 2:}\quad
Conversely, for any $\xi \in \mathbb{R}^{l}$, we can choose 
at most $l$ elements $x_1, \ldots, x_{l'} \in \mathcal{X}$ such that\par
\noindent$ \max_{x \in \mathcal{X}} \sum_{j=1}^l \xi^j f_{d-j}(x)=
\sum_{j=1}^l \xi^j f_{d-j}(x_i)$ and 
$\sum_{j=1}^l \xi^j f_{d-j}(x)<\sum_{j=1}^l \xi^j f_{d-j}(x_i)$ for $x \notin 
\{x_1, \ldots, x_{l'}\}$.
Then, we have 
\begin{align}
\tau (t \xi )
\to \Bigg(\sum_{i=1}^{l'} 
\frac{\exp\Big(\sum_{j=1}^{d-l}\theta^j f_j(x_i)   \Big)}
{\sum_{i'=1}^{l'}\exp\Big(\sum_{j=1}^{d-l}\theta^j f_j(x_{i'})   \Big)}
 f_{d-j}(x_i)\Bigg)_{j=1}^{l}
\end{align}
as $t \to \infty$.

\noindent{\bf Step 3:}\quad
We consider the compact set 
$\mathcal{T}(t):=
\{\tau (\xi )\}_{ \max_{x\in \mathcal{X}}  \sum_{j=1}^l \xi^j f_{d-j}(x)=t }$ for a large real number $t>0$.
The analysis on Steps 1 and 2 shows that 
the set $\mathcal{T}(t)$ approaches to the boundary of the convex full of $\{ (f_{d-j}(x))_{j=1}^{l}\}_{x \in \mathcal{X}}$
when $t$ approaches infinity.
Since map $\tau $ is continuous,
the image $\mathcal{D}(t)$
of $\{\xi \in \mathbb{R}^{l}
|\max_{x\in \mathcal{X}}  \sum_{j=1}^l \xi^j f_{d-j}(x)\le t
\}$ for the map $\tau$
is a compact subset whose boundary is close to the boundary of the convex full of $\{ (f_{d-j}(x))_{j=1}^{l}\}_{x \in \mathcal{X}}$.
Therefore, 
$\cup_{t>0}\mathcal{D}(t)$ equals the convex full of $\{ (f_{d-j}(x))_{j=1}^{l}\}_{x \in \mathcal{X}}$.

\end{proofof}

\begin{proofof}{Lemma \ref{LOS}}
It is sufficient to show Conditions (E3) and (M3) for any set of vectors 
$v_1, \ldots, v_l$, where $l=1, \ldots, d-1$.
This fact can be shown as follows.
First, we show (E3). For this aim, we choose an invertible matrix $U$ such that
$u_{d-i}=v_i$ for $i=1, \ldots,l$. 
For simplicity, we rewrite $\sum_{i=1}^{d} u_j^i f_i$ by $f_j$.
Also, we choose $(a_1, \ldots, a_l) \in \mathbb{R}^l$ such that
$\mathcal{M}(a_1, \ldots, a_l) $ is not empty.
We show that $\mathcal{M}(a_1, \ldots, a_l)$ is $\mathbb{R}^{d-l}$.
Due to Lemma \ref{LOS2},
for $(\theta^1, \ldots, \theta^{d-l})\in \mathbb{R}^{d-l}$,
there exists $(\theta^{d-l+1}, \ldots, \theta^{d})\in \mathbb{R}^{l}$
such that
\begin{align}
\frac{\sum_{x \in \mathcal{X}}
f_{d-j}(x) \exp \Big(  \sum_{i=1}^{d} \theta^i f_{i}(x)  \Big) }{\mu(\theta)}=
a_j \Label{MLJ}
\end{align}
for $j=1, \ldots, l$.
The above condition is equivalent to 
\begin{align}
\frac{\partial \mu}{\partial \theta^j}(\theta)=a_j.\Label{MLJ2}
\end{align}
This condition implies the relation 
$\mathcal{M}(a_1, \ldots, a_l)=\mathbb{R}^{d-l}$.
Hence, we have Condition (E3). 

Next, we show (M3). 
The relation \eqref{MLJ} means that 
the set $\Xi_{\mathcal{E}(\theta_0)}$ does not depend on $ \theta_0 \in \Theta$
because 
the choice of $(\theta^1, \ldots, \theta^{d-l})\in \mathbb{R}^{d-l}$ corresponds to 
the choice of $ \theta_0 \in \Theta$
in the relation \eqref{MLJ2}.
Hence, we have Condition (M3). 
\end{proofof}

\subsection{Classical system with fixed marginal distribution}\Label{4A2}
We consider the set of probability distributions on the finite set $\X \times \Y$ with $n_1=|\X|$ and $n_2=|\Y|$.
In particular, the marginal distribution on $\X$
is restricted as $P_X(x)=p_x$.
We focus on $d$ linearly independent functions $\bar{f}_1, \ldots, \bar{f}_{n_2-1}$ defined on ${\cal Y}$,
where the linear space spanned by $\bar{f}_1, \ldots, \bar{f}_{n_2-1}$
does not contain a constant function.
Then, we define the $C^{\infty}$ strictly convex function 
$\bar\mu$ on $\mathbb{R}^{n_1(n_2-1)}$ as
$\bar\mu(\bar\theta):= 
\sum_{x \in \X}p_{x}\mu_x(\bar\theta)$,
where
$\mu_x(\bar\theta):=
\log (\sum_{y \in {\cal Y}} \exp (  
\sum_{j=1}^{n_2-1} \bar\theta^{(x-1)(n_2-1)+j} \bar{f}_j(y)  )$, which yield
the Bregman divergence system $(\mathbb{R}^{n_1(n_2-1)}, \bar\mu, D^{\bar\mu})$. 

A probability distribution with full support on $\X \times \Y$
with the marginal distribution $p_x$ 
can be written as $P_\theta$, which is defined as
as $P_{\bar\theta}(x,y):=p_x 
\exp \Big(  (\sum_{j=1}^{n_2-1} \bar\theta^{(x-1)(n_2-1)+j} \bar{f}_j(y)  ) - \mu_x(\bar\theta) \Big)$.
The KL divergence 
equals the Bregman divergence of the potential function $\bar\mu$, i.e., 
we have
\begin{align}
&D^{\bar\mu}(\bar\theta\|\bar\theta_0)\nonumber \\
=&\sum_{x,j}p_x 
\Big(\frac{\partial \mu_x(\bar\theta)}{\partial  \bar\theta^{(x-1)(n_2-1)+j}}
(\bar{\theta}^{(x-1)(n_2-1)+j}-{{\bar{\theta}}_0^{(x-1)(n_2-1)+j}}) 
-\mu_x(\bar\theta)+ \mu_x(\bar\theta_0)
\Big) \nonumber\\
=&
D(P_{\bar\theta}\|P_{\bar\theta_0})
\Label{MGA4}
\end{align}
for $\theta \in \mathbb{R}^{d}$.

Next, we consider 
the Bregman divergence system $(\mathbb{R}^{n_1 n_2-1}, \mu, D^{\mu})$ defined in Subsection \ref{4A}
with $f_1, \ldots, f_{n_1n_2-1}$ defined as follows;
$f_{(i-1)(n_2-1)+j}(x,y):= \delta_{i,x}\bar{f}_x(y) $ for 
$i=1, \ldots, n_1$ and $j=1, \ldots, n_2-1$.
$f_{n_1(n_2-1)+i}(x,y):= \delta_{i,x}$ for 
$i=1, \ldots, n_1-1$.
We define the mixture subfamily $\mathcal{M}$ by the constraint
\begin{align}
\frac{\partial \mu}{\partial \theta^{n_1(n_2-1)+i}}=
p_i \Label{XOP}
\end{align}
for $i=1, \ldots, n_1-1$.
We apply the discussion given in Subsection \ref{APB} to 
the mixture subfamily $\mathcal{M}$.
The matrix $U$ is the identity matrix.
The mixture subfamily $\mathcal{M}$ is parameterized 
by the natural parameter 
$\bar\theta=(\bar\theta^1, \ldots, \bar\theta^{n_1(n_2-1)})$.
The function $\kappa$ is chosen as
$\kappa^{n_1(n_2-1)+i}(\bar\theta):=\mu_x(\bar\theta)$.
Then, the parameter $(\bar\theta, \kappa(\bar\theta))$
satisfies the condition \eqref{XOP}.
Hence, the mixture subfamily $\mathcal{M}$ coincides with 
the Bregman divergence system $(\mathbb{R}^{n_1(n_2-1)}, \bar\mu, D^{\bar\mu})$.

As an extension of Lemma \ref{LOS}, we have the following lemma.
\begin{lemma}\Label{LOST}
The Bregman divergence system $(\mathbb{R}^{n_1(n_2-1)}, \bar\mu, D^{\bar\mu})$ defined in this subsection satisfies
Conditions (E4) and (M4).
\end{lemma}

\begin{proof}
Condition (E4) holds because the parameter set is $\mathbb{R}^{n_1(n_2-1)}$.
Since the Bregman divergence system $(\mathbb{R}^{n_1 n_2-1}, \mu, D^{\mu})$ satisfies
Condition (M4),
its mixture subfamily $\mathcal{M}$
satisfies Condition (M4).
Hence, $(\mathbb{R}^{n_1(n_2-1)}, \bar\mu, D^{\bar\mu})$
satisfies Condition (M4).
\end{proof}

\subsection{Quantum system}\Label{4B}
In the quantum system, we focus on the $n$-dimensional Hilbert space ${\cal H}$ \cite{hayashi}.
We choose $d$ linearly independent Hermitian matrices $X_1, \ldots, X_{d}$ on ${\cal H}$,
where the linear space spanned by $X_1, \ldots, X_{d}$ does not contain the identify matrix.
Then, we define the $C^{\infty}$ strictly convex function 
$\mu$ on $\mathbb{R}^{d}$ as
$\mu(\theta):= \log (\Tr \exp (  \sum_{j=1}^{d} \theta^j X_j  )$.
A quantum state on ${\cal H}$ is given as
a positive semi definite Hermitian matrix $\rho$ with the condition $\Tr \rho=1$, which is called 
a density matrix.
We denote the set of density matrices by ${\cal S}({\cal H})$.
Any density matrix with full support on ${\cal H}$
can be written as $\rho_\theta$, which is defined as
as $\rho_\theta := 
\exp \Big(  (\sum_{j=1}^{d} \theta^j X_j  ) - \mu(\theta) \Big)$.
It is known that the relative entropy 
equals the Bregman divergence of the potential function $\mu$ \cite[Section 7.2]{Amari-Nagaoka}, 
i.e., we have
\begin{align}
D^{\mu}(\theta\|\theta')= D(\rho_\theta\|\rho_{\theta'})
\Label{MGA2}
\end{align}
for $\theta \in \mathbb{R}^{d}$, where
the relative entropy $D(\rho\| \rho')$ is defined as
\begin{align}
D(\rho\| \rho')=\Tr \rho (\log \rho - \log \rho').
\end{align}

\begin{example}
When $d=n^2-1$, 
the Bregman divergence system $(\mathbb{R}^{d}, \mu, D^\mu)$ 
describes the set ${\cal S}(\H)$ 
of density matrices on $\H$ with full support and
the relative entropy.
\end{example}

\begin{example}
When $\H$ is given as $\H_1 \otimes \H_2$ with $n_i=\dim \H_i$,
$X_i$ is an Hermitian matrices with the form 
$A \otimes I$ or $I \otimes B$, and $d=n_1^2+n_2^2-2$,
the Bregman divergence system $(\mathbb{R}^{d}, \mu, D^\mu)$ 
describes the set ${\cal S}(\H_1) \otimes {\cal S}(\H_2)$ 
of product density matrices on $\H_1 \otimes \H_2$.
\end{example}

When the parameter $\theta$ is limited to 
$(\bar\theta,\underbrace{0, \ldots, 0}_{d-l})$ with 
$\bar\theta \in \mathbb{R}^l$,
the set of distributions $\rho_\theta$ forms an exponential family.
Also, when the linear space spanned by $d-k$ linearly independent Hermitian matrices
$Y_1, \ldots, Y_{d-k}$
does not contain a constant function,
for $d-k$ constants $a_1, \ldots, a_{d-k}$, 
the following set of distributions forms a mixture family; 
\begin{align}
\Big\{\rho_\theta \Big|
\Tr Y_i \rho_\theta=a_i \hbox{ for } i=1, \ldots, d-k \Big\}.
\end{align}

For the possibility of the projection, we have the following lemma.
\begin{lemma}\Label{LOS3}
The Bregman divergence system $(\mathbb{R}^{d}, \mu, D^\mu)$ defined in this section satisfies Conditions (E4) and (M4).
\end{lemma}

To show this lemma, we prepare the following lemma in a way similar to Lemma \ref{LOS2}.
\begin{lemma}\Label{LOS4}
For $(\theta^1, \ldots, \theta^{d-l})\in \mathbb{R}^{d-l}$
and $\xi:=(\xi^{1}, \ldots, \xi^{l})\in \mathbb{R}^{l}$,
we define 
\begin{align}
\tau_{(\theta^1, \ldots, \theta^{d-l})}(\xi):=
\Bigg(\frac{
\Tr X_{d-j} \exp \Big(  \sum_{i=1}^{d} \theta^j X_{d-i}  \Big) }{\mu(\theta)}
\Bigg)_{j=1}^{l}
\end{align}
with $\theta^{d-l+i}=\xi^i$.
Then, the set $\mathcal{T}_{1,(\theta^1, \ldots, \theta^{d-l})}:=
\{ \tau_{(\theta^1, \ldots, \theta^{d-l})}(\xi) | \xi \in \mathbb{R}^{l} \}$
equals the inner $\mathcal{T}_2$ of the convex full of 
$\{ (  \Tr X_{d-j} \rho)_{j=1}^{l}\}_{\rho \in \mathcal{P}}$,
where $\mathcal{P}$ is the set of pure states.
\end{lemma}
\begin{proofof}{Lemma \ref{LOS4}}
In this proof, $\mathcal{T}_{1,(\theta^1, \ldots, \theta^{d-l})}$ 
and $\tau_{(\theta^1, \ldots, \theta^{d-l})}(\xi)$
are simplified to $\mathcal{T}_1$ and $\tau(\xi)$.
Since $\mathcal{T}_1\subset \mathcal{T}_2$ is trivial, we show only the opposite relation.

\noindent{\bf Step 1:}\quad
Any element in the boundary of the convex full of 
$\{ (\Tr X_{d-j}\rho)_{j=1}^{l}\}_{\rho \in \mathcal{P}}$
is written as \par
\noindent$\Big(\sum_{i=1}^{l'} p_i \Tr \rho_i X_{n-1-j}\Big)_{j=1}^{l}$ with extremal points 
$(\Tr \rho X_{d-j})_{j=1}^{l}$
with at most $l$ orthogonal elements $\rho_i \in \mathcal{P}$ and at most  $l$ positive numbers $p_i$, 
where $i=1, \ldots, l'\le l$.
There exists an element $\xi_* \in \mathbb{R}^{l}$ such that
$ \max_{\rho \in \mathcal{P}} \sum_{j=1}^l \xi_*^j \Tr \rho X_{d-j}=
\sum_{j=1}^l \xi_*^j \Tr \rho_i X_{d-j}=1$ and
$\sum_{j=1}^l \xi_*^j \Tr \rho X_{d-j}<1$ unless 
$(\Tr \rho X_{d-j})_{j=1}^{l}$ is written as a convex combination of 
$\{( \Tr \rho_i X_{d-j})_{j=1}^{l}\}_{i=1}^{l'}$.
For any $\rho_i$, there exists an element $\xi(\rho_i)_* \in \mathbb{R}^{l}$ such that
$ \max_{\rho \in \mathcal{P}} \sum_{j=1}^l \xi(\rho_i)_*^j \Tr \rho X_{d-j}=
\sum_{j=1}^l \xi(\rho_i)_*^j \Tr \rho_i X_{d-j}=1$,
$\sum_{j=1}^l \xi(x_i)_*^j \Tr \rho X_{d-j}<1$ for $\rho (\neq \rho_i) \in \mathcal{P}$
and
$\sum_{j=1}^l \xi(x_i)_*^j \Tr \rho_{i'}X_{d-j} >0$ for $i' \neq i$.
Then, there exists elements $t_i>0$ such that
\begin{align}
\frac{
\Tr \rho_i \exp 
\Big(\sum_{i'=1}^{l'}\sum_{j=1}^l t_{i'}\xi(x_{i'})_*^j  X_{n-1-j} 
+ \sum_{j=1}^{d-l}\theta^j X_j \Big)}
{\sum_{i''=1}^{l'}
\Tr \rho_{i''} \exp\Big(\sum_{i'=1}^{l'}\sum_{j=1}^l t_{i'}\xi(x_{i'})_*^j X_{d-j} 
+ \sum_{j=1}^{d-l}\theta^j X_j \Big)}
=p_i.
\end{align}
Hence, we have
\begin{align}
\tau \Big(t \xi(x_0) + \sum_{i'=1}^{l'} t_{i'}\xi(x_{i'})_* \Big)
\to \Big(\sum_{i=1}^{l'} p_i \Tr \rho_i X_{d-j} \Big)_{j=1}^{l}
\end{align}
as $t \to \infty$.

\noindent{\bf Step 2:}\quad
Conversely, for any $\xi \in \mathbb{R}^{l}$, we can choose 
at most $l$ orthogonal pure states $\rho_1, \ldots, \rho_{l'} \in \mathcal{X}$ such that
$\sum_{j=1}^l \xi^j X_{d-j}$ is commutative with $\rho_1, \ldots, \rho_{l'}$,
$ \max_{\rho \in \mathcal{P}} \sum_{j=1}^l \xi^j \Tr \rho X_{d-j}=
\sum_{j=1}^l \xi^j \Tr \rho_i X_{d-j}$ and 
$\sum_{j=1}^l \xi^j \Tr \rho X_{d-j}<\sum_{j=1}^l \xi^j \Tr \rho_i X_{d-j}$ unless
$(\Tr \rho X_{d-j})_{j=1}^{l}$ is written as a convex combination of 
$\{( \Tr \rho_i X_{d-j})_{j=1}^{l}\}_{i=1}^{l'}$.
Then, we have 
\begin{align}
\tau (t \xi )
\to \Bigg(\sum_{i=1}^{l'} 
\frac{\Tr \rho_i \exp\Big(\sum_{j=1}^{d-l}\theta^j X_j   \Big)}
{\sum_{i'=1}^{l'}\Tr \rho_{i'} \exp\Big(\sum_{j=1}^{d-l}\theta^j X_j   \Big)}
 X_{d-j}\Bigg)_{j=1}^{l}
\end{align}
as $t \to \infty$.

\noindent{\bf Step 3:}\quad
We consider the compact set 
$\mathcal{T}(t):=
\{\tau (\xi )\}_{ \| (\sum_{j=1}^l \xi^j X_{d-j})_+\| =t }$ for large real number $t>0$, where
$(X)_+$ is an operator composed of the positive part.
The analysis on Steps 1 and 2 shows that 
the set $\mathcal{T}(t)$ approaches to the boundary of the convex full of 
$\{ ( \Tr \rho X_{d-j})_{j=1}^{l}\}_{\rho \in \mathcal{P}}$
when $t$ approaches infinity.
Since the map $\tau $ is continuous,
the image $\mathcal{D}(t)$
of $\{\xi \in \mathbb{R}^{l} |~
\| (\sum_{j=1}^l \xi^j X_{d-j})_+\| \le t
\}$ for the map $\tau$
is a compact subset whose boundary is close to the boundary of the convex full of 
$\{ ( \Tr \rho X_{d-j})_{j=1}^{l}\}_{\rho \in \mathcal{P}}$.
Therefore, 
$\cup_{t>0}\mathcal{D}(t)$ equals the convex full of 
$\{ ( \Tr \rho X_{d-j})_{j=1}^{l}\}_{\rho \in \mathcal{P}}$.

\end{proofof}

\begin{proofof}{Lemma \ref{LOS3}}
Lemma \ref{LOS3} can be shown in the same way as Lemma \ref{LOS}
by replacing the role of Lemma \ref{LOS2} by Lemma \ref{LOS4}.
\end{proofof}

\section{em-algorithm}\Label{Sec:em}
\subsection{Basic description for algorithm}\Label{Sec:emS1}
In this section, we address a minimization problem for
a pair of a $k$-dimensional mixture subfamily $\mathcal{M}$ and an $l$-dimensional  exponential subfamily $\mathcal{E}$ although the paper \cite{Fujimoto} discussed a similar problem setting based on Bregman divergence.
Here, we employ notations $u^i_{k+j}$, $a_j$, etc,
for a $k$-dimensional mixture subfamily $\mathcal{M}$ and an $l$-dimensional  exponential subfamily $\mathcal{E}$
that are introduced in Subsections \ref{S2-1} and \ref{S2-2}.
We assume the following condition;
\begin{description}
\item[(B0)]
The Bregman divergence system $(\Theta,F,D^F)$
satisfies Conditions (E4) and (M4). 
\end{description}
Hence, the minimums 
$\min_{\theta' \in \mathcal{E}} D^{F}(\theta \| \theta')$ 
and
$\min_{\theta \in \mathcal{M}} D^{F}(\theta \| \theta')$ 
exist.
We consider the following minimization problem;
\begin{equation}\Label{min1}
C_{\inf}(\mathcal{M},\mathcal{E})
:=
\inf_{\theta \in \mathcal{M}} D^{F}(\theta \| \Pro^{(e),F}_{\mathcal{E}} (\theta))
=
\inf_{\theta \in \mathcal{M}} 
\min_{\theta' \in \mathcal{E}} 
D^{F}(\theta \| \theta').
\end{equation}

The first task is to clarify whether the minimum exists in \eqref{min1}.
If the minimum exists, our second task is to find the minimization point
\begin{equation}\Label{min2}
\theta^*(\mathcal{M},\mathcal{E})
:=
\argmin_{\theta \in \mathcal{M}} D^{F}(\theta \| \Pro^{(e),F}_{\mathcal{E}} (\theta)).
\end{equation}
When we define 
$\theta_*(\mathcal{M},\mathcal{E})
:=\Pro^{(e),F}_{\mathcal{E}} (\theta^*(\mathcal{M},\mathcal{E}))$,
we have the opposite relation 
$\theta^*(\mathcal{M},\mathcal{E})
=\Pro^{(m),F}_{\mathcal{M}} (\theta_*(\mathcal{M},\mathcal{E}))$
because $\theta^*(\mathcal{M},\mathcal{E})$ achieves the maximum.
Hence, we have the relation 
$\mathcal{M}_{\theta_*\to \mathcal{E}}
=\mathcal{E}_{\theta^*\to \mathcal{M}}$.
If there is no risk of confusion, 
$\theta^*(\mathcal{M},\mathcal{E})$ and $\theta_*(\mathcal{M},\mathcal{E})$
are simplified to $\theta^*$ and $\theta_*$, respectively.
If the minimum does not exit, our second task is 
 to find a sequence of elements $\{\theta_{n}^*(\mathcal{M},\mathcal{E})\}
$ in $\mathcal{M}$ to achieve the infimum \eqref{min1}.

\begin{figure}[htbp]
\begin{center}
  \includegraphics[width=0.7\linewidth]{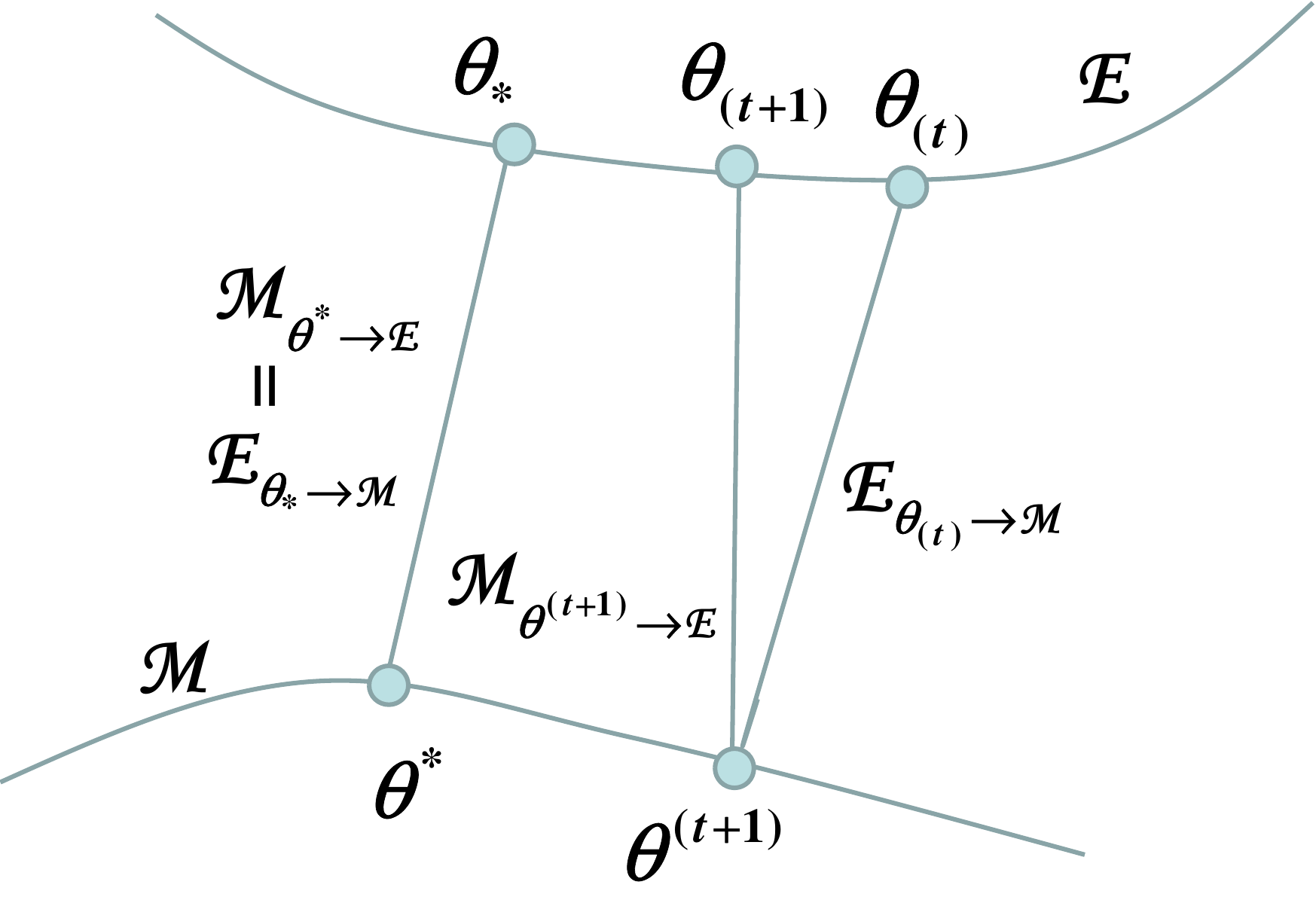}
  \end{center}
\caption{Algorithms \ref{protocol1-0} and \ref{protocol1}: 
This figure shows the topological relation among
$\theta_*$, $\theta^*$,
$\theta_{(t+1 )}$, $\theta^{(t+1 )}$, and $\theta_{(t)}$,
which is used in the
application of Phythagorean theorem (Proposition \ref{MNL}).
$\mathcal{M}_{\theta_*\to \mathcal{E}}
=\mathcal{E}_{\theta^*\to \mathcal{M}}$
and 
$\mathcal{M}_{\theta^{(t+1)} \to \mathcal{E}}$
are the mixture subfamilies to project $\theta(\epsilon_1)$
and $\theta^{(t+1)}$
to the exponential subfamily $\mathcal{E}$, respectively.
$\mathcal{E}_{\theta_{(t)}\to \mathcal{M}}$
is the exponential subfamily to project $\theta_{(t)}$
to the mixture subfamily $\mathcal{M}$.
}
\Label{FFY}
\end{figure}

Although the above minimization problem is very common in machine learning and statistics,
many kinds of minimization problems in information theory can be written in the above form as explained in Section \ref{S1}.
The above minimization asks to minimize the divergence between 
two points in the mixture and exponential subfamilies ${\cal E}$ and ${\cal M}$.
Algorithm \ref{protocol1-0} shows an algorithm to calculate the element 
$\theta^*(\mathcal{M},\mathcal{E})$ to achieve the minimum.
This algorithm is called the em algorithm, and is illustrated in Fig. \ref{FFY}.
By describe the m-step in a concrete form,
Algorithm \ref{protocol1-0} is rewritten as Algorithm \ref{protocol1}, which follows from 
(A3) of Lemma \ref{Th6}.

\begin{algorithm}
\caption{em-algorithm}
\Label{protocol1-0}
\begin{algorithmic}
\STATE {Assume that $\mathcal{M}$ is characterized by \eqref{BO1}.
Choose the initial value ${\theta}_{(1)} \in \mathcal{E}$;} 
\REPEAT 
\STATE {\bf m-step:}\quad 
Calculate ${\theta}^{(t+1)}:=\Pro^{(m),F}_{{\cal M}} ({\theta}_{(t)})$.
That is, ${\theta}^{(t+1)}$ is given as 
$\argmin_{\theta \in \mathcal{M}} D^{F}({\theta} \| {\theta}_{(t)})$, i.e.,
the unique element in ${\cal M}$ to realize the minimum
of the smooth convex function $\theta \mapsto D^{F}(\theta\| {\theta}_{(t)})$.
\STATE {\bf e-step:}\quad 
Calculate ${\theta}_{(t+1)}:=\Pro^{(e),F}_{{\cal E}} (\theta^{(t+1)})$.
That is, ${\theta}_{(t+1)}$ is given as 
$\argmin_{\theta' \in \mathcal{E}} D^{F}({\theta}^{(t+1)} \| \theta')$, i.e.,
the unique element in ${\cal E}$ to realize the minimum
of the smooth convex function $\theta'\mapsto D^{F}({\theta}^{(t+1)} \| \theta')$.
\UNTIL{convergence.} 
\end{algorithmic}
\end{algorithm}

When the mixture family $\mathcal{M}$ has to many parameters,
the optimization in m-step takes a long time.
In this case, m-step can be replaced by another optimization problem
with $d-k$ parameters. 
This replacement is useful when $k > d-k$.

\begin{algorithm}
\caption{em-algorithm}
\Label{protocol1}
\begin{algorithmic}
\STATE {Assume that $\mathcal{M}$ is characterized by \eqref{BO1}.
Choose the initial value ${\theta}_{(1)} \in \mathcal{E}$;} 
\REPEAT 
\STATE {\bf m-step:}\quad 
Calculate ${\theta}^{(t+1)}:=\Pro^{(m),F}_{{\cal M}} ({\theta}_{(t)})$.
That is, ${\theta}^{(t+1)}$ is given as 
${\theta}_{(t)}+ \sum_{j=k+1}^{d} {\tau}^{j}_o u_{j}$,
where $({\tau}_o^{k+1}, \ldots, {\tau}_o^{d})$ is the unique element to satisfy 
\begin{align}
\frac{\partial}{\partial \tau^{\bar{j}}} F \Big({\theta}_{(t)}+ \sum_{j=1}^{l} \tau^{j} u_{j} \Big) \Big|_{{\tau}^{j}={\tau}^{j}_o}=a_{\bar{j}}\Label{const1-2}
\end{align}
for $\bar{j}=k+1, \ldots, d$.
The above choice is equivalent to the following;
\begin{align}
({\tau}_o^{k+1}, \ldots, {\tau}_o^{d}):= \argmin_{\bar{\tau}^{k+1}, \ldots, \bar{\tau}^{d}}
 F \Big({\theta}_{(t)}+ \sum_{j=1}^{l} \bar{\tau}^{j} u_{j} \Big) 
- \sum_{{j}=k+1}^d \bar{\tau}_j a_{j}.
 \Label{const1-2-0}
\end{align}
\STATE {\bf e-step:}\quad 
Calculate ${\theta}_{(t+1)}:=\Pro^{(e),F}_{{\cal E}} (\hat{\theta}^{(t+1)})$.
That is, ${\theta}_{(t+1)}$ is given as 
$\argmin_{\theta' \in \mathcal{E}} D^{F}({\theta}^{(t+1)} \| \theta')$, i.e.,
the unique element in ${\cal E}$ to realize the minimum
of the smooth convex function $\theta'\mapsto D^{F}({\theta}^{(t+1)} \| \theta')$.
\UNTIL{convergence.} 
\end{algorithmic}
\end{algorithm}

The em-algorithm repetitively applies the function $\Pro^{(m),F}_{{\cal M}}\circ \Pro^{(e),F}_{{\cal E}}|_{{\cal M}}$
for an element $\theta \in {\cal M}$.
Since the application of $\Pro^{(m),F}_{{\cal M}}\circ \Pro^{(e),F}_{{\cal E}}|_{{\cal M}}$
monotonically decreases the minimum Bregman divergence from the exponential family $\mathcal{E}$, 
when we apply the updating rule ${\theta}^{(t+1)}:= 
\Pro^{(m),F}_{{\cal M}}\circ \Pro^{(e),F}_{{\cal E}}|_{{\cal M}} ({\theta}^{(t)})$,
it is expected that the outcome ${\theta}^{(t)}$ of the repetitive application converges to 
$\theta^*(\mathcal{M},\mathcal{E})$.
However, 
it is not guaranteed that the converged point gives the global minimum in general \cite{Amari,Fujimoto,Allassonniere}.
To get a global minimum by this algorithm, we introduce the following condition for an exponential subfamily $\mathcal{E}$.
\begin{description}
\item[(B1)]
The relation 
\begin{align}
 D^{F}(\theta'\|\theta)
\ge
 D^{F}( \Pro^{(e),F}_{{\cal E}}  (\theta')\|  \Pro^{(e),F}_{{\cal E}}  (\theta) ) \Label{MLA2I}
\end{align}
holds for any $\theta,\theta' \in \Theta$.
\end{description}
Also, as its weak version, we consider the following condition.
\begin{description}
\item[(B1${\cal M}$)]~
The relation 
\begin{align}
 D^{F}(\theta'\|\theta)
\ge
 D^{F}( \Pro^{(e),F}_{{\cal E}}  (\theta')\|  \Pro^{(e),F}_{{\cal E}}  (\theta) ) \Label{MLA2}
\end{align}
holds for any $\theta,\theta' \in {\cal M}$.
\end{description}

Then, we have the following theorem.
\begin{theorem}\Label{theo:conv2}
Assume Conditions (B0), (B1${\cal M}$), 
and $\sup_{\theta \in \mathcal{E}} D^F(\theta \| \theta_{(1)})<\infty$
for a pair of a $k$-dimensional mixture subfamily $\mathcal{M}$ and an 
$l$-dimensional  exponential subfamily $\mathcal{E}$.
Then, in Algorithms \ref{protocol1-0} and \ref{protocol1},
the quantity $ D^{F}(\theta^{(t)}\| \Pro^{(e),F}_{{\cal E}}  (\theta^{(t)}) ) 
$ converges to the minimum $C_{\inf}(\mathcal{M},\mathcal{E})$ with the speed 
\begin{align}
 D^{F}(\theta^{(t)}\| \Pro^{(e),F}_{{\cal E}}  (\theta^{(t)}) ) 
-C_{\inf}(\mathcal{M},\mathcal{E})
=o(\frac{1}{t}).\Label{mma2}
\end{align}
Also, we have another type evaluation
\begin{align}
 D^{F}(\theta^{(t)}\| \Pro^{(e),F}_{{\cal E}}  (\theta^{(t)}) ) 
-C_{\inf}(\mathcal{M},\mathcal{E})
\le \frac{\sup_{\theta \in \mathcal{M}} D^F(\theta \| \theta_{(1)})}{t-1}
.\Label{mma2BYY}
\end{align}
Further, when 
$t-1 \ge \frac{ \sup_{\theta \in \mathcal{M}} D^F(\theta \| \theta_{(1)})}{\epsilon}$,
the parameter $ \theta^{(t)}$ satisfies
\begin{align}
 D^{F}(\theta^{(t)}\| \Pro^{(e),F}_{{\cal E}}  (\theta^{(t)}) ) 
- C_{\inf}(\mathcal{M},\mathcal{E})
 \le \epsilon.
\Label{NHG2}
\end{align}
In particular, when the minimum in \eqref{min1} exists, i.e., 
$\theta^*(\mathcal{M},\mathcal{E})$ exists,
the supremum $\sup_{\theta \in \mathcal{E}} D^F(\theta \| \theta_{(1)})$ in the above evaluation is replaced by
$D^F(\theta^*(\mathcal{M},\mathcal{E}) \| \theta_{(1)})$.
\end{theorem}
The proof of Theorem \ref{theo:conv2} is given in Appendix \ref{A2}.

To improve the above evaluation, we introduce a strength version of Condition (B1)
as a condition for $\mathcal{M},\mathcal{E}$, and $\theta'\in \mathcal{E}$.
\begin{description}
\item[(B1+)]
The minimizer $\theta^*=\theta^*(\mathcal{M},\mathcal{E})$ exists.
There exists a constant $\beta(\theta')<1$ to satisfy the following condition.
When an element $\theta \in \im \Pro^{(m),F}_{{\cal M}}|_{\mathcal{E}}
\subset {\cal M} $ satisfies the condition
$ D^{F}(\theta^*\|\theta)\le  D^{F}( \theta_{*}\|  \theta' ) $,
the relation
\begin{align}
 \beta(\theta') D^{F}(\theta^*\|\theta)
\ge
 D^{F}( \theta_{*}\|  \Pro^{(e),F}_{{\cal E}}  (\theta) ) \Label{MLA2+}
\end{align}
holds. 
\end{description}
Then, we have the following theorem.

\begin{theorem}\Label{theo:conv2+}
Assume that Conditions (B0) and (B1+) hold
for a pair of a $k$-dimensional mixture subfamily $\mathcal{M}$, an $l$-dimensional  exponential subfamily $\mathcal{E}$, and 
$\theta'=\theta_{(1)} \in \mathcal{E}$.
Then, in Algorithms \ref{protocol1-0} and \ref{protocol1},
the quantity $ D^{F}(\theta^{(t)}\| \Pro^{(e),F}_{{\cal E}}  (\theta^{(t)}) ) 
$ converges to the minimum $C_{\inf}(\mathcal{M},\mathcal{E})$ with the speed 
\begin{align}
 D^{F}(\theta^{(t)}\| \Pro^{(e),F}_{{\cal E}}  (\theta^{(t)}) ) 
-C_{\inf}(\mathcal{M},\mathcal{E})
=
\beta(\theta_{(1)})^{t-2} D^F(\theta_{*} \| \theta_{(1)}).
\Label{mma2+}
\end{align}
Further, when 
$t-2 \ge \frac{ \log  D^F(\theta_{*} \| \theta_{(1)}) -\log \epsilon}{-\log\beta(\theta_{(1)})}$,
the parameter $ \theta^{(t)}$ satisfies
\begin{align}
 D^{F}(\theta^{(t)}\| \Pro^{(e),F}_{{\cal E}}  (\theta^{(t)}) ) 
- C_{\inf}(\mathcal{M},\mathcal{E})
 \le \epsilon.
\Label{NHG2+}
\end{align}
\end{theorem}
The proof of Theorem \ref{theo:conv2+} is given in Appendix \ref{A3}.

In fact, it is not so easy to find $\theta_{(1)}$ to satisfy Condition (B1+).
However, when we apply Algorithm \ref{protocol1},
$\theta_{(t)}$ becomes close to $\theta_{*}$ with sufficiently large $t$.
When $\theta_{(t)}\in \mathcal{E}$ belongs to the neighborhood
of $\theta_{*}$,
Condition (B1+) holds by substituting $\theta_{(t)}$ into $\theta_{(1)}$
so that 
Theorem \ref{theo:conv2+} can be applied with sufficiently $t$.
That is, once $\theta_{(t)}\in \mathcal{E}$ belongs to the neighborhood
of $\theta_{*}$, we have an exponential convergence. 

Further, it is not easy to implement $e$- and $m$- projections perfectly, in general.
Hence, we need an alternative algorithm instead of 
Algorithms \ref{protocol1-0} and \ref{protocol1}.
Now, we consider the case when only e-step can be perfectly implemented and 
m-step is approximately done with $\epsilon$ error.
Examples of such a case will be discussed later sections.
Algorithm \ref{protocol1-0} is modified as follows.

\begin{algorithm}
\caption{em-algorithm with $\epsilon$ approximated m-step in the mixture subfamily $\mathcal{M}$}
\Label{protocol1-error}
\begin{algorithmic}
\STATE {Assume that $\mathcal{M}$ is characterized by \eqref{BO1}.
Choose the initial value ${\theta}_{(1)} \in \mathcal{E}$;} 
\REPEAT 
\STATE {\bf m-step:}\quad 
Calculate ${\theta}^{(t+1)}$.
That is, we choose an element ${\theta}^{(t+1)}  \in \mathcal{M}$
such that 
\begin{align}
D^{F}({\theta}^{(t+1)} \| {\theta}_{(t)})
\le \min
\Big(D^{F}({\theta}^{(t)} \| {\theta}_{(t)}),
\min_{\theta \in \mathcal{M}} D^{F}({\theta} \| {\theta}_{(t)})+\epsilon \Big),
\Label{NLT}
\end{align}
where $D^{F}({\theta}^{(1)} \| {\theta}_{(1)})$ is defined as $\infty$.

\STATE {\bf e-step:}\quad 
Calculate ${\theta}_{(t+1)}:=\Pro^{(e),F}_{{\cal E}} (\hat{\theta}^{(t+1)})$.
That is, ${\theta}_{(t+1)}$ is given as 
$\argmin_{\theta' \in \mathcal{E}} D^{F}({\theta}^{(t+1)} \| \theta')$, i.e.,
the unique element in ${\cal E}$ to realize the minimum
of the smooth convex function $\theta'\mapsto D^{F}({\theta}^{(t+1)} \| \theta')$.
\UNTIL{convergence.} 
\end{algorithmic}
\end{algorithm}

Then, we have the following theorem.
\begin{theorem}\Label{theo:conv2BB}
Assume Conditions (B0), (B1), 
and the existence of the minimizer $\theta^*=\theta^*(\mathcal{M},\mathcal{E})$ in \eqref{min2} 
for a pair of a $k$-dimensional mixture subfamily $\mathcal{M}$ and an 
$l$-dimensional  exponential subfamily $\mathcal{E}$.
In addition, we define the set $\mathcal{E}_0:=\{\theta \in \mathcal{E} |
D^F(\theta_{*}\| \theta ) \le D^F(\theta_{*}\| \theta_{(1)} ) \} \subset \mathcal{E}$.

Then, in Algorithm \ref{protocol1-error}, 
the quantity $ D^{F}(\theta^{(t)}\| \Pro^{(e),F}_{{\cal E}}  (\theta^{(t)}) ) 
$ converges to the minimum $C_{\inf}(\mathcal{M},\mathcal{E})$ with the speed 
\begin{align}
& D^{F}(\theta^{(t+1)}\| \Pro^{(e),F}_{{\cal E}}  (\theta^{(t+1)}) ) 
-C_{\inf}(\mathcal{M},\mathcal{E}) \nonumber \\
\le &
\frac{D^F(\theta_{*}\| \theta_{(1)}  )}{t}
 + 2\gamma \sqrt{D^F(\theta_{*}\|\theta_{(1)})
\epsilon}+
(\gamma+1) \epsilon  .
\Label{mma2BA}
\end{align}
where $\gamma:=\gamma(\mathcal{E}_0|\mathcal{E})$.
Further, when 
$t \ge \frac{ 2 D^F(\theta_{*,1} \| \theta_{(1)})}{\epsilon'}+1$
and 
$\epsilon \le \frac{{\epsilon'}^2}{4 (3 \gamma+1)^2 D^F(\theta_{*}\|\theta_{(1)})}$,
the parameter $ \theta^{(t)}$ satisfies
\begin{align}
 D^{F}(\theta^{(t)}\| \Pro^{(e),F}_{{\cal E}}  (\theta^{(t)}) ) 
- C_{\inf}(\mathcal{M},\mathcal{E})
 \le \epsilon'.
\Label{NHG2T}
\end{align}
\end{theorem}
The proof of Theorem \ref{theo:conv2BB} is given in Appendix \ref{A4}.

\begin{algorithm}
\caption{em-algorithm
with $\epsilon$ approximated m-step in the exponential subfamily}
\Label{protocol1Berror}
\begin{algorithmic}
\STATE {Assume that $\mathcal{M}$ is characterized by \eqref{BO1}.
We choose two parameters $\epsilon_1 < \epsilon_2 $.
Choose the initial value ${\theta}_{(1)} \in \mathcal{E}$;} 
\REPEAT 
\STATE {\bf m-step:}\quad 
We choose ${\theta}^{(t+1)}\in {\cal M}$ and $\bar{\theta}^{(t+1)}={\theta}_{(t)}+ \sum_{j=k+1}^{d} {\tau}^{j}_o u_{j}
$
such that 
\begin{align}
& F \Big({\theta}_{(t)}+ \sum_{j=k+1}^{d} {\tau}_o^{j} u_{j} \Big) 
- \sum_{{j}=k+1}^d {\tau}_o^j a_{j} \nonumber \\
\le &
\min_{\bar{\tau}^{k+1}, \ldots, \bar{\tau}^{d}}
 F \Big({\theta}_{(t)}+ \sum_{j=k+1}^{d} \bar{\tau}^{j} u_{j} \Big) 
- \sum_{{j}=k+1}^d \bar{\tau}^j a_{j}+\epsilon_1.\Label{NKD}
\end{align}
and
\begin{align}
D({\theta}^{(t+1)} \| \bar{\theta}^{(t+1)}) \le \epsilon_2.
\Label{NKDB}
\end{align}
\STATE {\bf e-step:}\quad 
Calculate ${\theta}_{(t+1)}: =\Pro^{(e),F}_{{\cal E}} (\bar{\theta}^{(t+1)})$.
\UNTIL{
$t=t_1-1$.} 
\STATE {\bf final step:}\quad 
We output the final estimate $\theta_f^{(t_1)} :=\theta^{(t_2)} \in \mathcal{M}$
by using  $t_2:= \argmin_{t=2, \ldots, t_1} 
D^F(\theta^{(t)} \| \theta_{(t-1)})-D^F(\theta^{(t)} \| \bar\theta^{(t)})$.
\end{algorithmic}
\end{algorithm}

Since m-step has two conditions, Algorithm \ref{protocol1Berror} seems complicated.
This step can be realized as follows.
The condition \eqref{NKD} simply shows the error for the minimization 
of the convex function $ F \Big({\theta}_{(t)}+ \sum_{j=k+1}^{d} \bar{\tau}^{j} u_{j} \Big) 
- \sum_{{j}=k+1}^d \bar{\tau}^j a_{j}$.
The condition \eqref{NKDB} is related to the choice of $ {\theta}^{(t+1)}\in {\cal M}$.
As one possible choice,
we choose ${\theta}^{(t+1)}$ as follows. 
Next, we choose the element $\kappa^{j'}$ 
by solving the equations
\begin{align}
u_{j'}^{i'} \Big(\eta_{i'}(\bar{\theta}^{(t+1)})+ 
\sum_{i=1}^d J_{i',i}(\bar{\theta}^{(t+1)}) \sum_{j=k+1}^d u_{j}^i \kappa^{j}\Big)=
a_{j'}
\end{align}
for $j'=k+1, \ldots, d$.
Then, we choose the element ${\theta}^{(t+1)} $ 
by $ \eta_j({\theta}^{(t+1)})=\eta_j(\bar{\theta}^{(t+1)})+ 
\sum_{i=1}^d \sum_{j'=k+1}^d J_{j,i}(\bar{\theta}^{(t+1)})u_{j'}^i \kappa^{j'}$ for $j=1, \ldots, d$.
If ${\theta}^{(t+1)}$ does not satisfy \eqref{NKDB}, we retake $\bar{\theta}^{(t+1)} $
such that  the value
$ F \Big({\theta}_{(t)}+ \sum_{j=k+1}^{d} {\tau}_o^{j} u_{j} \Big) 
- \sum_{{j}=k+1}^d {\tau}_o^j a_{j} $ is smaller than the previous one.

In this way,
the m-step of Algorithm \ref{protocol1-error}
requires the approximate calculation of the minimum 
$\min_{\theta \in \mathcal{M}} D^{F}({\theta} \| {\theta}_{(t)})$,
which can be done as the convex minimization with respect to the mixture parameter 
in $\mathcal{M}$.
However, this minimization needs to handle $d-k$ parameters.
If $k < d-k$, the alternative minimization given in \eqref{const1-2-0}
has a smaller number of parameters.
As an approximate version of Algorithm \ref{protocol1}, we have 
Algorithm \ref{protocol1Berror}.
Indeed, if we can calculate the derivative of the convex function 
$ F \Big({\theta}_{(t)}+ \sum_{j=k+1}^{d} \bar{\tau}^{j} u_{j} \Big) 
- \sum_{{j}=k+1}^d \bar{\tau}^j a_{j}$,
we can employ algorithms explained in Appendix \ref{AA1}.

In Algorithm \ref{protocol1Berror},
we use the relation 
\begin{align}
\Pro^{(e),F}_{{\cal E}} ({\theta}^{(t+1)})
=\Pro^{(e),F}_{{\cal E}} (\bar{\theta}^{(t+1)}).\Label{XLQ}
\end{align}
In fact, the point $\Pro^{(e),F}_{{\cal E}} ({\theta}^{(t+1)})$ is characterized 
by the intersection between the exponential subfamily ${\cal E}$
and the mixture subfamily whose mixture parameters $\eta_1, \ldots, \eta_l$ are fixed to 
$\eta_1({\theta}^{(t+1)}), \ldots, \eta_l({\theta}^{(t+1)})$.
Hence, the above relation \eqref{XLQ} holds.

Then, we have the following theorem.
\begin{theorem}\Label{theo:conv2BC}
Assume Conditions (B0), (B1), 
and the existence of the minimizer $\theta^*:=\theta^*(\mathcal{M},\mathcal{E})$ in \eqref{min2} 
for a pair of a $k$-dimensional mixture subfamily $\mathcal{M}$ and an 
$l$-dimensional  exponential subfamily $\mathcal{E}$.
Then, in Algorithm \ref{protocol1Berror}, 
we have
\begin{align}
D^F({\theta}^{(t+1),*}\|\bar{\theta}^{(t+1)})\le \epsilon_1 \Label{XP8}
\end{align}
for $t=1, \ldots, t_1-1$,
where
${\theta}^{(t+1),*}$ is defined as 
${\theta}_{(t)}+ \sum_{j=k+1}^{d} {\tau}^{j}_* u_{j}$ by
using 
$(\tau^{k+1}_{*}, \ldots, \tau^d_{*}):=
\argmin_{\bar{\tau}^{k+1}, \ldots, \bar{\tau}^{d}}
 F \Big({\theta}_{(t)}+ \sum_{j=k+1}^{d} \bar{\tau}^{j} u_{j} \Big) 
- \sum_{{j}=k+1}^d \bar{\tau}^j a_{j}$.
Also, 
the quantity $ D^{F}(\theta_f^{(t_1)}\| \Pro^{(e),F}_{{\cal E}} 
 (\theta_f^{(t_1)}) ) 
$ converges to the minimum $C_{\inf}(\mathcal{M},\mathcal{E})$ with the speed 
\begin{align}
& D^{F}(\theta_f^{(t_1)}\| \Pro^{(e),F}_{{\cal E}} (\theta_f^{(t_1)}) ) 
-C_{\inf}(\mathcal{M},\mathcal{E}) \nonumber\\
\le & \frac{1}{t_1-1} D^F( \theta_{*} \| \theta_{(1)})+  \epsilon_1 + \epsilon_2
\Label{Qma2BA}.
\end{align}
\end{theorem}
The proof of Theorem \ref{theo:conv2BC} is given in Appendix \ref{A5}.

Considering Taylor expansion, we have 
\begin{align}
a_{j'}=& \sum_{j'} u_{j'}^{i'}\eta_{i'}({\theta}_{(t_2-1)}+ \sum_{j=k+1}^{d} {\tau}_*^{j} u_{j} ) \nonumber\\
\cong &
u_{j'}^{i'} \Big(\eta_{i'}(\bar{\theta}^{(t_2)})+ 
\sum_{i=1}^d J_{i',i}(\bar{\theta}^{(t_2)}) \sum_{j=k+1}^d u_{j}^i ({\tau}_*^{j}-\bar{\tau}^{j})\Big)
\end{align}
for $j'=k+1, \ldots, d$.
Hence, 
\begin{align}
\kappa^{j}\cong ({\tau}_*^{j}-\bar{\tau}^{j})
\end{align}

Using \eqref{MLA}, we have
\begin{align}
D^F({\theta}^{(t_2),*}\|\bar{\theta}^{(t_2)})
\cong 
\frac{1}{2}\sum_{j=1}^d \sum_{i=1}^d \sum_{j'=k+1}^d 
\sum_{i'=k+1}^d 
J_{j,i}(\bar{\theta}^{(t_2)})u_{j'}^i ({\tau}_*^{j'}-\bar{\tau}^{j'}) 
u_{i'}^j ({\tau}_*^{i'}-\bar{\tau}^{i'}). \Label{XP9}
\end{align}
Using \eqref{MLA}, we have
\begin{align}
&D^F(\theta^{(t_2)} \| \bar\theta^{(t_2)})
=D^{F^*}( \eta(\bar\theta^{(t_2)}) \| \eta(\theta^{(t_2)}) ) \nonumber\\
\cong &
\frac{1}{2}\sum_{j=1}^d \sum_{\bar{j}=1}^d (J(\bar{\theta}^{(t_2)})^{-1})^{j,\bar{j}}
\sum_{\bar{i}=1}^d \sum_{\bar{j}'=k+1}^d J_{\bar{j},\bar{i}}(\bar{\theta}^{(t_2)})u_{\bar{j'}}^i \kappa^{\bar{j'}}
\sum_{i=1}^d \sum_{j'=k+1}^d J_{j,i}(\bar{\theta}^{(t_2)})u_{j'}^i \kappa^{j'} \nonumber\\
=&
\frac{1}{2}\sum_{j=1}^d \sum_{i=1}^d \sum_{j'=k+1}^d 
\sum_{i'=k+1}^d 
J_{j,i}(\bar{\theta}^{(t_2)})u_{j'}^i \kappa^{j'} 
u_{i'}^j \kappa^{i'}.\Label{XP10}
\end{align}
Combining \eqref{XP8}, \eqref{XP10}, and \eqref{XP9}, we have
\begin{align}
D^F(\theta^{(t_2)} \| \bar\theta^{(t_2)})
\lessapprox \epsilon_1
\end{align}
Therefore, \eqref{Qma2BA} is rewritten as
\begin{align}
& D^{F}(\theta_f^{(t_1)}\| \Pro^{(e),F}_{{\cal E}} (\theta_f^{(t_1)}) ) 
-C_{\inf}(\mathcal{M},\mathcal{E}) \nonumber\\
\lessapprox & \frac{1}{t_1-1} D^F( \theta_{*} \| \theta_{(1)})+  2 \epsilon_1 
\Label{Qma2BT}.
\end{align}
Hence, 
when 
$t_1-1 \ge \frac{3 D^F(\theta_{*,1} \| \theta_{(1)})}{\epsilon}$,
and $\epsilon_1 \le \frac{\epsilon}{3}$,
the parameter $ \theta^{(t)}$ satisfies
\begin{align}
 D^{F}(\theta_f^{(t_1)}\| \Pro^{(e),F}_{{\cal E}}  (\theta_f^{(t_1)}) ) 
- C_{\inf}(\mathcal{M},\mathcal{E})
 \lessapprox \epsilon  .
\Label{WHG2T}
\end{align}

\subsection{Closed convex mixture family}\Label{SEC-CC}
In this section, we address a similar minimization problem for
a pair of a $k$-dimensional closed convex mixture subfamily $\mathcal{M}$ and an $l$-dimensional  exponential subfamily $
\mathcal{E}$ under the following condition (B0).
That is, we discuss a closed convex mixture subfamily instead of a mixture subfamily $\mathcal{M}$ while we consider an exponential subfamily 
$\mathcal{E}$.
Under this condition, 
we employ the same $e$-projection $\Pro^{(e),F}_{\mathcal{E}}$ 
defined in Lemma \ref{Lem8} as in the previous subsection,
but, we use
the $m$-projection $\Pro^{(m),F}_{\mathcal{M}}$ defined in Lemma \ref{LMGO}.
Hence, we consider Algorithm \ref{protocol1B} instead of Algorithm \ref{protocol1}.

\begin{algorithm}
\caption{em-algorithm with closed convex mixture family}
\Label{protocol1B}
\begin{algorithmic}
\STATE {Assume that $\mathcal{M}$ is characterized by the mixture parameter $\eta$.
Choose the initial value ${\theta}_{(1)} \in \mathcal{E}$;} 
\REPEAT 
\STATE {\bf m-step:}\quad 
Calculate ${\eta}^{(t+1)}$.
That is, ${\eta}^{(t+1)}$ is given as 
$\argmin_{\eta \in {\Xi_{\cal M}}} D^{F}(
\phi_{\mathcal{M}}^{(m)}({\eta}) \| {\theta}_{(t)})$, i.e.,
the unique element in ${\cal M}$ to realize the minimum
of the smooth convex function $\eta\mapsto D^{F}(
\phi_{\mathcal{M}}^{(m)}({\eta})\| {\theta}_{(t)})$.
\STATE {\bf e-step:}\quad 
Calculate ${\theta}_{(t+1)}:=\Pro^{(e),F}_{{\cal E}} (
\phi_{\mathcal{M}}^{(m)}({\eta}^{(t+1)}))$.
That is, ${\theta}_{(t+1)}$ is given as 
$\argmin_{\theta' \in \mathcal{E}} D^{F}(
\phi_{\mathcal{M}}^{(m)}({\eta}^{(t+1)}) \| \theta')$, i.e.,
the unique element in ${\cal E}$ to realize the minimum
of the smooth convex function $\theta'\mapsto D^{F}(
\phi_{\mathcal{M}}^{(m)}({\eta}^{(t+1)})\| \theta')$.
\UNTIL{convergence.} 
\end{algorithmic}
\end{algorithm}


When the boundary $\partial \mathcal{M}$ is composed of a finite number of closed mixture families, 
due to Lemmas \ref{Th6} and \ref{LO10},
Algorithm \ref{protocol1B} can be simplified to Algorithm \ref{protocol1C}
because Lemma \ref{LO10} guarantees that
$\Pro^{(m),F}_{{\cal M}} (\theta_{(t)})$
is given as $\Pro^{(m),F}_{\hat{\cal M}_{\lambda_0}} (\theta_{(t)})$,
where
we denote the extended mixture subfamily of 
${\mathcal{M}}_{\lambda}$ by $\hat{\mathcal{M}}_{\lambda}$ for $\lambda \in \Lambda_*$,
and $\lambda_0 $ is given in \eqref{XPA9}.

\begin{algorithm}
\caption{em-algorithm with closed convex mixture family whose boundary is composed 
of finite number of closed mixture families}
\Label{protocol1C}
\begin{algorithmic}
\STATE {Assume the following conditions;
A set of closed convex mixture subfamilies 
$\{ {\mathcal{M}}_{\lambda} \}_{\lambda \in \Lambda}$
covers the boundary $\partial \mathcal{M}$ of a closed convex mixture family $\mathcal{M}$
with subsets $\Lambda_\lambda \subset \Lambda$
and $\lambda \in \Lambda_*:=\Lambda \cup \{0\}$.
Each closed convex mixture subfamily $\mathcal{M}_{\lambda}$ is generated by the constraint by 
$\sum_{i=1} ^d u_{j,\lambda}^i \partial_i F(\theta)= a_{j,\lambda}$ for 
$j=k_{\lambda}+1, \ldots, d$ for $\lambda \in \Lambda_*$.
Choose the initial value ${\theta}_{(1)} \in \mathcal{E}$;} 

\REPEAT 
\STATE {\bf m-step:}\quad 
Calculate ${\theta}^{(t+1)}:=\Pro^{(m),F}_{{\cal M}} ({\theta}_{(t)})$ in the following way.
For $\lambda\in \Lambda_*$,
we calculate ${\theta}^{(t+1),\lambda}$ is given as 
${\theta}_{(t)}+ \sum_{j=k_{\lambda}+1}^{d} {\tau}^{j,\lambda} u_{j}$,
where $({\tau}^{k_{\lambda}+1,\lambda}, \ldots, {\tau}^{d,\lambda})$ is the unique element to satisfy 
\begin{align}
\frac{\partial}{\partial \tau^{\bar{j},\lambda}} F \Big({\theta}_{(t)}
+ \sum_{j=k_{\lambda}+1}^{d} \tau^{j,\lambda} u_{j} \Big) =a_{\bar{j},\lambda}\Label{const1-2-U}
\end{align}
for $\bar{j}=k_{\lambda}+1, \ldots, d$.
We set ${\theta}^{(t+1)}$ as ${\theta}^{(t+1),\lambda_0}$, where
\begin{align}
\lambda_0:= \argmin_{\lambda \in \Lambda_*}
\{ D^F({\theta}^{(t+1),\lambda} \| {\theta}_{(t)}) |
{\theta}^{(t+1),\lambda}\in \mathcal{M} \}.\Label{XPA9}
\end{align}

\STATE {\bf e-step:}\quad 
Calculate ${\theta}_{(t+1)}:=\Pro^{(e),F}_{{\cal E}} (\hat{\theta}^{(t+1)})$.
That is, ${\theta}_{(t+1)}$ is given as 
$\argmin_{\theta' \in \mathcal{E}} D^{F}({\theta}^{(t+1)} \| \theta')$, i.e.,
the unique element in ${\cal E}$ to realize the minimum
of the smooth convex function $\theta'\mapsto D^{F}({\theta}^{(t+1)} \| \theta')$.
\UNTIL{convergence.} 
\end{algorithmic}
\end{algorithm}

Then, in the same way as Theorem \ref{theo:conv2}, we have the following theorem.
\begin{theorem}\Label{theo:conv2T}
Assume Conditions (B0), (B1), and 
$\sup_{\theta \in \mathcal{E}} D^F(\theta \| \theta_{(1)})<\infty$
for a pair of a $k$-dimensional closed convex mixture subfamily $\mathcal{M}$ and 
an $l$-dimensional  exponential subfamily $\mathcal{E}$.
Then, Algorithms \ref{protocol1B} and \ref{protocol1C} have 
the same conclusion as Theorem \ref{theo:conv2}.
\end{theorem}

Also, in the same way as Theorem \ref{theo:conv2+}, we have the following theorem;
\begin{theorem}\Label{theo:conv2+T}
Assume that Conditions (B0) and (B1+) hold
for a pair of a $k$-dimensional close convex mixture subfamily $\mathcal{M}$, an $l$-dimensional  exponential subfamily $\mathcal{E}$, and $
\theta'=\theta_{(1)} \in \mathcal{E}$.
Then, the quantity $ D^{F}(\theta^{(t)}\| \Pro^{(e),F}_{{\cal E}}  (\theta^{(t)}) ) 
$ converges to the minimum $C_{\inf}(\mathcal{M},\mathcal{E})$ with the speed 
\begin{align}
 D^{F}(\theta^{(t)}\| \Pro^{(e),F}_{{\cal E}}  (\theta^{(t)}) ) 
-C_{\inf}(\mathcal{M},\mathcal{E})
=
\beta(\theta_{(1)})^{t-2} D^F(\theta_{*} \| \theta_{(1)}).
\Label{mma2+T}
\end{align}
Further, when 
$t-2 \ge \frac{ \log  D^F(\theta_{*} \| \theta_{(1)}) -\log \epsilon}{\log \beta(\theta_{(1)})}$,
the parameter $ \theta^{(t)}$ satisfies
\begin{align}
 D^{F}(\theta^{(t)}\| \Pro^{(e),F}_{{\cal E}}  (\theta^{(t)}) ) 
- C_{\inf}(\mathcal{M},\mathcal{E})
 \le \epsilon.
\Label{NHG2+T}
\end{align}
\end{theorem}

Theorems \ref{theo:conv2T} and \ref{theo:conv2+T} are shown in Appendix \ref{A5+1}.

When we need to care the error in the m-step,
as an error version of Algorithm \ref{protocol1B}, we have  
Algorithm \ref{protocol1-error2} in the same way as 
Algorithm \ref{protocol1-error}. 

\begin{algorithm}
\caption{em-algorithm with $\epsilon$ approximated m-step}
\Label{protocol1-error2}
\begin{algorithmic}
\STATE {Assume that $\mathcal{M}$ is characterized by the mixture parameter $\eta$.
Choose the initial value ${\theta}_{(1)} \in \mathcal{E}$;} 
\REPEAT 
\STATE {\bf m-step:}\quad 
Calculate ${\eta}^{(t+1)}$.
That is, we choose ${\eta}^{(t+1)} \in \mathcal{M}$ such that
\begin{align}
D^{F}({\theta}^{(t+1)} \| {\theta}_{(t)})
\le \min
\Big(D^{F}({\theta}^{(t)} \| {\theta}_{(t)}),
\min_{\theta \in \mathcal{M}} D^{F}({\theta} \| {\theta}_{(t)})+\epsilon \Big),
\Label{NLT2}
\end{align}
where $D^{F}({\theta}^{(1)} \| {\theta}_{(1)})$ is defined as $\infty$.
\STATE {\bf e-step:}\quad 
Calculate ${\theta}_{(t+1)}:=\Pro^{(e),F}_{{\cal E}} (
\phi_{\mathcal{M}}^{(m)}({\eta}^{(t+1)}))$.
That is, ${\theta}_{(t+1)}$ is given as 
$\argmin_{\theta' \in \mathcal{E}} D^{F}(
\phi_{\mathcal{M}}^{(m)}({\eta}^{(t+1)}) \| \theta')$, i.e.,
the unique element in ${\cal E}$ to realize the minimum
of the smooth convex function $\theta'\mapsto D^{F}(
\phi_{\mathcal{M}}^{(m)}({\eta}^{(t+1)})\| \theta')$.
\UNTIL{convergence.} 
\end{algorithmic}
\end{algorithm}

Then, we have the following theorem.
\begin{theorem}\Label{theo:conv2BB2}
Assume Conditions (B0), (B1), and the existence of the minimizer $\theta^*:=\theta^*(\mathcal{M},\mathcal{E})$ in \eqref{min2} 
for a pair of a $k$-dimensional mixture subfamily $\mathcal{M}$ and an $l$-dimensional  exponential subfamily $\mathcal{E}$.
In addition, we define the set $\mathcal{E}_0:=\{\theta \in \mathcal{E} |
D^F(\theta_{*}\| \theta ) \le D^F(\theta_{*}\| \theta_{(1)} ) \} \subset \mathcal{E}$
and $\theta_{*}:=\Pro^{(e),F}_{{\cal E}}(\theta^*)$.
Then, Algorithm \ref{protocol1-error2} has 
the same conclusion as Theorem \ref{theo:conv2BB}.
\end{theorem}

Theorem \ref{theo:conv2BB2} is shown in Appendix \ref{A5+1}.

When $k < d-k$, we need an alternative minimization for 
the $m$-step for Algorithm \ref{protocol1-error2}
in a way similar to Algorithm \ref{protocol1Berror}.
However, although we can consider a modification of  
Algorithm \ref{protocol1C} in a way similar to Algorithm \ref{protocol1Berror},
it is not so easy to evaluate the error or the modified algorithm.
Hence, to take into account the error in the $m$-step,
we propose another method to modify Algorithm \ref{protocol1C} 
as Algorithm \ref{protocol1Berror2}. 

\begin{algorithm}
\caption{em-algorithm
with $\epsilon$ approximated m-step in the exponential subfamily}
\Label{protocol1Berror2}
\begin{algorithmic}
\STATE {We assume the same conditions as Algorithm \ref{protocol1C}.
We denote the extended mixture subfamily of 
${\mathcal{M}}_{\lambda}$ by $\hat{\mathcal{M}}_{\lambda}$ for $\lambda \in \Lambda_*$.} 
\STATE {\bf 1st-step:}\quad 
For $\lambda \in \Lambda_*$, we apply Algorithm \ref{protocol1Berror}
to the pair of the exponential subfamily $\mathcal{E}$ and 
the mixture subfamily $\hat{\mathcal{M}}_{\lambda}$.
As the result with $t$ iteration, 
we denote the number $t_2$ in this application of Algorithm \ref{protocol1Berror}
by $t_2(\lambda)$.
Then, we denote 
$\theta^{(t_2(\lambda))}$, $\bar\theta^{(t_2(\lambda))}$, and 
$\theta_{(t_2(\lambda)-1)}$ in this application
by 
$\theta^{(t_2(\lambda)),\lambda}$, $\bar\theta^{(t_2(\lambda)),\lambda}$, and 
$\theta_{(t_2(\lambda)-1),\lambda}$, respectively


\STATE {\bf 2nd-step:}\quad 
We output the final estimate $\theta_f^{(t)} :=\theta^{(t_2(\lambda_0)),\lambda_0} \in \mathcal{M}$, where
\begin{align}
\lambda_0:= \argmin_{\lambda \in \Lambda_*}
\Big\{ 
D^F \Big({\theta}^{(t_2(\lambda)),\lambda} \Big\| 
{\theta}_{(t_2(\lambda)-1),\lambda} \Big)
  \Big|
{\theta}^{(t_2(\lambda)),\lambda}\in \mathcal{M}_\lambda \Big\}.
\end{align}
\end{algorithmic}
\end{algorithm}

To evaluate the error of Algorithm \ref{protocol1Berror2}, we prepare
the following lemma.
Therefore, using Theorem \ref{theo:conv2BC}, 
we obtain the following theorem for the error evaluation of 
Algorithm \ref{protocol1Berror2}.

\begin{theorem}\Label{theo:conv2BC2}
Assume the same assumption as Algorithm \ref{protocol1Berror2} and Conditions (B0) and (B1) for $\mathcal{E}$.
Also, we assume the existence of the minimizer $\theta^*:=\theta^*(\mathcal{M}_{\lambda},\mathcal{E})$ in \eqref{min2} 
for $\lambda \in \Lambda_*$.
Then, in Algorithm \ref{protocol1Berror2}, 
the quantity $ D^{F}(\theta_f^{(t)}\| \Pro^{(e),F}_{{\cal E}} 
(\theta_f^{(t)}) ) 
$ converges to the minimum $C_{\inf}(\mathcal{M},\mathcal{E})$ with the speed 
\begin{align}
& D^{F}(\theta_f^{(t)}\| \Pro^{(e),F}_{{\cal E}}  (\theta_f^{(t)}) ) 
-C_{\inf}(\mathcal{M},\mathcal{E}) \nonumber \\
\le &
(D(0)+1)
\max_{\lambda\in \Lambda_*}\Big(
\frac{1}{t_1-1} D^F(\theta_{*}(\mathcal{M}_\lambda,\mathcal{E})\| \theta_{(1)}  ) 
+ \epsilon_1 +D^F(\theta^{(t_2(\lambda)),\lambda} \| \bar\theta^{(t_2(\lambda)),\lambda}) 
\Big)\Label{Qma2BA2}.
\end{align}
Notice that $D(\lambda)$ is defined before Lemma \ref{CLP}.
\end{theorem}
The proof of Theorem \ref{theo:conv2BC2} is given in Appendix \ref{A6}.

\section{Classical rate distortion}\Label{S5}
\subsection{Classical rate distortion without side information}\Label{S5S1}
Let $\X:=\{1,\ldots,n_1 \}$ and $\Y:= \{1,\ldots, n_2 \}$ be finite sets. 
We call a map  $ W  :\X  \rightarrow  {\cal P}_{\Y}$
a channel from $\X$ to $\Y$. 
We denote the set of the above maps by ${\cal P}_{\Y|\X}$.
We use the notation $W_x(y):=W(y|x)$.
For $q \in {\cal P}_{\X}$ and $r \in {\cal P}_{\Y}$, 
$W \cdot q \in {\cal P}_{\Y}$,
$W \times q \in {\cal P}_{\X \times \Y}$, and $q \times r \in \mathcal{P}_{\X \times \Y}$ 
are defined by 
$(W \cdot q)(y):=\sum_{x \in \X}W(y|x)q(x)$, 
$(W \times q)(x,y):=W(y|x)q(x)$, and $(q \times r)(x,y) := q(x)r(y)$ respectively. 

Given a distortion measure $d(x, y)$ on $\X \times \Y$ and a distribution $P_X$ on $\X$, 
we define the following sets;
\begin{align}
{\cal P}_{\Y|\X}^{d,P_X,D}:=& \Big\{ W \in {\cal P}_{\Y|\X}\Big|  \sum_{x \in \X,y\in \Y}
d(x,y) W\times P_X (x,y) =D \Big\} \\
{\cal P}_{\Y|\X}^{d,P_X,D,\le}:=& \Big\{ W \in{\cal P}_{\Y|\X}\Big|  \sum_{x \in \X,y\in \Y}
d(x,y) W\times P_X (x,y) \le D \Big\}.
\end{align}
We define $\bar{d}(x,y)$ as
\begin{align}
\bar{d}(x,y):={d}(x,y) -{d}(x,n_2),\quad
\bar{d}(x,n_2):=0
\end{align}
for $x\in \X$ and $y =1, \ldots, n_2-1$.
Then, the condition
\begin{align}
\sum_{x \in \X,y\in \Y} d(x,y) W\times P_X (x,y) \le D 
\end{align}
is equivalent to 
\begin{align}
\sum_{x \in \X,y\in \Y} \bar{d}(x,y) W\times P_X (x,y) \le D- 
\sum _{x \in \X} P_X(x) {d}(x,n_2).
\end{align}
Hence, for simplicity, we assume that 
$d(x, n_2)=0$ in the following.
Also, we define the vector $d=(d_j)_{j=1}^{n_1(n_2-1)}$ as
$d_{(x-1) (n_2+1)+y}:= d(i,j)$ for $x\in \X$ and $y=1, \ldots, n_2-1$.

The standard rate distortion function is given as
\begin{align}
\min_{W \in {\cal P}_{\Y|\X}^{d,P_X,D,\le} }I(X;Y)_{W \times P_X}=&
\min_{W \in {\cal P}_{\Y|\X}^{d,P_X,D,\le} }D(W \times P_X
\| (W \cdot P_X ) \times P_X ) \nonumber \\
=&\min_{W \in {\cal P}_{\Y|\X}^{d,P_X,D,\le} } \min_{  q \in {{\cal P}_{\Y}} }
 D(W \times P_X \| q \times P_X ).\Label{MOT}
\end{align}
In the following, we use the notation 
$W_*:= \argmin_{W \in {\cal P}_{\Y|\X}^{d,P_X,D,\le} }I(X;Y)_{W \times P_X}$.

When there exists a distribution $Q_Y$ on $\Y$ such that
\begin{align}
\sum_{x,y}P_X(x)Q_Y(y) d(x,y)\le D,\Label{CXP}
\end{align}
the above minimum \eqref{MOT} is zero.
The existence of $Q_Y$ to satisfy the condition \eqref{CXP} is equivalent to 
\begin{align}
\min_y d_Y(y) \le D,\Label{CXP2}
\end{align}
where 
$d_{Y}(y):= \sum_{x \in {\cal X}}P_X(x) d(x,y)$.

Then,  we consider the Bregman divergence system 
$(\mathbb{R}^{n_1(n_2-1)}, \bar\mu, D^{\bar\mu})$ defined in Subsection \ref{4A2}, which coincides with 
the set of distributions $W \times P_X$.
The set of distributions $q \times P_X$ forms an exponential subfamily ${\cal E}$, and 
the subset ${\cal P}_{\Y|\X}^{d,P_X,D}\times P_X$
forms a mixture subfamily ${\cal M}$.

Then, we have the following theorem.
\begin{lemma}\Label{TXPO}
When \eqref{CXP2} holds, 
$\min_{W \in {\cal P}_{\Y|\X}^{d,P_X,D,\le} }I(X;Y)_{W \times P_X}=0$.
Otherwise,
\begin{align}
\min_{W \in {\cal P}_{\Y|\X}^{d,P_X,D,\le} }I(X;Y)_{W \times P_X}
=&
\min_{W \in {\cal P}_{\Y|\X}^{d,P_X,D} }I(X;Y)_{W \times P_X}
\nonumber \\
=&\min_{W \in {\cal P}_{\Y|\X}^{d,P_X,D} } \min_{  q \in {{\cal P}_{\Y}} }
 D(W \times P_X \| q \times P_X ).\Label{MOTR}
\end{align}
\end{lemma}
\begin{proof}
The first statement has been already shown.
We show the second statement by contradiction. 
Assume that \eqref{CXP2} nor the first equation in \eqref{MOTR}
does not hold. 
We define $Q_{Y,1}$ as
$Q_{Y,1}\times P_X=\Pro^{(e),\bar{\mu}}_{{\cal E}}(W_*\times P_X)$.

Since $Q_{Y,1}\times P_X$ does not belong to ${\cal P}_{\Y|\X}^{d,P_X,D,\le}$,
applying \eqref{MGOF} in Lemma \ref{LMGO} to 
the closed convex mixture subfamily ${\cal P}_{\Y|\X}^{d,P_X,D,\le}$,
we find that
$\Pro^{(m),\bar{\mu}}_{{\cal M}}  (Q_{Y,1}\times P_X)
=\Pro^{(m),\bar{\mu}}_{{\cal M}}  
\circ \Pro^{(e),\bar{\mu}}_{{\cal E}}  
(W_*\times P_X)$
belongs to ${\cal P}_{\Y|\X}^{d,P_X,D}$.
We choose $W_1$ such that $(W_1\times P_X)
=\Pro^{(m),\bar{\mu}}_{{\cal M}}  (Q_{Y,1}\times P_X)$.
Hence, we have
\begin{align}
I(X;Y)_{W_* \times P_X}
=& D( W_*\times P_X  \| Q_{Y,1}\times P_X)\nonumber\\
\ge &
D( W_1\times P_X  \| Q_{Y,1}\times P_X)
\ge
I(X;Y)_{W_1 \times P_X},
\end{align}
which contradicts 
$W_*= \argmin_{W \in {\cal P}_{\Y|\X}^{d,P_X,D,\le} }I(X;Y)_{W \times P_X}$.
\end{proof}

Due to Lemma \ref{TXPO}, when \eqref{CXP2} does not hold,
it is sufficient to address the minimization \eqref{MOTR}.
In the following, we address the minimization problem \eqref{MOTR}, which  is a special case of the minimization \eqref{min1}
with the formulation given in Subsection \ref{Sec:emS1}.
The mixture family $\mathcal{M}$ has $ n_1(n_2-1)-1$ parameters.

Since the total dimension is $n_1(n_2-1)$,
we employ Algorithm \ref{protocol1} instead of Algorithm \ref{protocol1-0}.
Since Lemma \ref{LOST} guarantees Condition (B0) for this problem, 
Algorithm \ref{protocol1} works and is rewritten as Algorithm \ref{protocol1-2}. 

\begin{algorithm}
\caption{em-algorithm for rate distortion}
\Label{protocol1-2}
\begin{algorithmic}
\STATE {Choose the initial distribution $P_Y^{(1)}$ on $\Y$.
Then, we define the initial joint distribution $P_{XY,(1)}$ as 
$P_Y^{(1)}\times P_X $;} 
\REPEAT 
\STATE {\bf m-step:}\quad 
Calculate $P_{XY}^{(t+1)}$ as $P_{XY}^{(t+1)}(x,y):=
P_X(x) P_{Y}^{(t)}(y)e^{\bar{\tau} d(x,y)}
\Big(\sum_{y'}P_{Y}^{(t)}(y') e^{\bar{\tau} d(x,y')} \Big)^{-1}$,
where $\bar{\tau}$ is the unique element
$\tau$ to satisfy 
\begin{align}
\frac{\partial}{\partial \tau} 
\sum_{x}P_{X}(x)
\log \Big(\sum_{y} P_{Y}^{(t)}(y) e^{{\tau} d(x,y)}\Big)
&=D \Label{const1-4} 
\end{align}
This choice can be written in the way as \eqref{const1-2-0}.
\STATE {\bf e-step:}\quad 
Calculate $P_{Y}^{(t+1)}(y)$ as 
$\sum_{x \in \X}P_{XY}^{(t+1)}(x,y)$.
\UNTIL{convergence.} 
\end{algorithmic}
\end{algorithm}

To check Condition (B1),
we set $\theta$ and $\theta'$ be elements of $\mathbb{R}^{n_1n_2-1}$ corresponding to 
$W \times P_X $ and $  W' \times P_X$
in the sense of the Bregman divergence system 
$(\mathbb{R}^{n_1(n_2-1)}, \bar\mu, D^{\bar\mu})$ defined in Subsection \ref{4A2}.
Then, the relation
\begin{align}
& D^{\bar{\mu}}( \Pro^{(e),\bar{\mu}}_{{\cal E}}  (\theta')\|  \Pro^{(e),\bar{\mu}}_{{\cal E}}  (\theta) ) 
= D( (W' \cdot P_X)\times P_X\|  (W' \cdot P_X)\times P_X)\nonumber \\
=& D( W' \cdot P_X \|  W' \cdot P_X) 
\le  D( W' \times P_X \|  W' \times P_X)
= D^{\bar{\mu}}(\theta'\|\theta)\Label{MKP}
\end{align}
guarantees condition (B1).
When the initial value $\theta_{(1)}$ is chosen as the case that
$W$ has full support, 
$\sup_{\theta \in \mathcal{E}}D^{\bar{\mu}}(\theta \| \theta_{(1)})$ has a finite value.
Hence, 
Theorem \ref{theo:conv2T} guarantees the convergence to the global minimum.
Now, we set $\theta_{(1)}$ to be the product of $P_X$ and the uniform distribution on $\Y$.
Then, we have
\begin{align}
D^{\bar{\mu}}(\theta_{*}\| \theta_{(1)} )
\le \sup_{\theta \in {\cal M}} D^{\bar{\mu}}(\theta\| \theta_{(1)} )
=\log n_2 \Label{XPZ}.
\end{align}
Hence, the inequality \eqref{mma2BYY} is rewritten as
\begin{align}
I(X;Y)_{P_{XY}^{(t)}}
-\min_{W \in {\cal P}_{\Y|\X}^{d,P_X,D,\le} }I(X;Y)_{W \times P_X}
\le \frac{\log n_2}{t-1}
\Label{mma2BYY2}
\end{align}
In particular, when $t \ge \frac{\log n_2}{\epsilon}+1$, the 
above value is bounded by $\epsilon$.

The original problem \eqref{MOTR} is written as a concave optimization with respect to $n_1(n_2-1) $ mixture parameters
because the mutual information is concave with respect to the conditional distribution.
Although our protocol contains a convex optimization in m-step,
the convex optimization in m-step has only one variable.
Therefore, our method is considered to convert a complicated concave optimization with a larger size
to iterative applications of a convex optimization with one variable.

Next, we consider the case when 
we cannot exactly calculate the unique element
$\bar{\tau}$ to satisfy 
\eqref{const1-4}. Alternatively, we need to use $\epsilon$ approximation for the solution.
We employ Algorithm \ref{protocol1Berror}, which requires to solve
the minimization of
the one-variable smooth convex function
$\hat{F}[P_{Y}](\tau) := \sum_{x}P_{X}(x)
\log \Big(\sum_{y} P_{Y}(y) e^{{\tau}(D- d(x,y)))}\Big)
$.
That is, it  is needed to find the minimizer
$\tau_*[P_Y]:=\argmin_{\tau} \hat{F}[P_{Y}](\tau) $.

To consider this minimization, we focus on the one-parameter exponential subfamily 
$P_{X,Y|\tau}[P_Y](x,y):= 
P_{X}(x) \frac{ P_{Y}(y) e^{{\tau} (D-d(x,y))}}{\sum_{y'} P_{Y}(y') e^{{\tau}(D- d(x,y'))}}$.
The first and second derivatives are calculated as 
\begin{align}
\frac{d}{d\tau}\hat{F}[P_{Y}](\tau)=&
\mathbb{E}_{P_{X,Y|\tau}[P_Y]} [D-d(X,Y)] \\
\frac{d^2}{d\tau^2}\hat{F}[P_{Y}](\tau)=&
\mathbb{E}_{P_{X,Y|\tau}[P_Y]} [(D-d(X,Y))^2]-
\mathbb{E}_{P_{X,Y|\tau}[P_Y]} [D-d(X,Y)] ^2.
\end{align}
Defining $ \zeta_{+}:= \max_{x,y}|D-d(x,y)|^2$, we have
\begin{align}
\frac{d^2}{d\tau^2}\hat{F}[P_{Y}](\tau) \le \zeta_+.\Label{XNP2}
\end{align}
The condition \eqref{XNP2} guarantees that
\begin{align}
\hat{F}[P_{Y}](\tau)
\le
\hat{F}[P_{Y}](0)
+  \frac{d}{d\tau}\hat{F}[P_{Y}](0) \tau+ \frac{1}{2}\zeta_+ \tau^2 \Label{XNP3}
\end{align}
for $\tau>0$.
To solve $\min_{\tau} \hat{F}[P_{Y}](\tau)$, 
we employ the bisection method explained in Appendix \ref{AA1A}.
Since \eqref{CXP2} holds, the relation 
$\mathbb{E}_{P_X\times {P}_{Y}}[d(X,Y)] >D$, i.e., 
$\frac{d}{d\tau}\hat{F}[P_{Y}](0)< 0$
 holds for any distribution $P_Y$.
Hence, $
\frac{d}{d \tau}\hat{F}[P_{Y}](-\frac{\frac{d}{d \tau}\hat{F}[P_{Y}](0)}{\zeta_-})
\ge 0$.

For the application of the bisection method, we consider the following condition for 
the convex function $\hat{F}[P_{Y}](\tau)$;
\begin{align}
\frac{d^2}{d \tau^2}\hat{F}[P_{Y}](\tau) \ge \zeta_-
 \hbox{ for } \tau \in [0, \tau_*[P_Y]].
\Label{XNP}
\end{align}
Since the condition \eqref{XNP} guarantees $0 \le \tau_*[P_Y] \le -\frac{\frac{d}{d \tau}\hat{F}[P_{Y}](0)}{\zeta_-}$,
we can apply the bisection method, Algorithm \ref{Bisection} with 
$a=0$ and $b= -\frac{\frac{d}{d \tau}\hat{F}[P_{Y}](0)}{\zeta_-}$.
Under the condition \eqref{XNP}, we have
\begin{align}
\hat{F}[P_{Y}](0)- \hat{F}[P_{Y}](\tau_*[P_Y])
\le -\frac{d}{d \tau}\hat{F}[P_{Y}](0) \tau_*[P_Y] 
\le  \frac{1}{\zeta_-}\big(\frac{d}{d \tau}\hat{F}[P_{Y}](0) \big)^2.
\Label{XNP5}
\end{align}

We choose the estimate $\tau_k[P_Y]$ as $b_k$ of Algorithm \ref{Bisection}, which requires $k$ iterations.
Then, we have $\frac{d}{d \tau}\hat{F}[P_{Y}](\tau_k[P_Y])>0 $.
The relation \eqref{GAT2} guarantees that 
\begin{align}
&\hat{F}[P_{Y}](\tau_k[P_Y])-\hat{F}[P_{Y}](\tau_*[P_Y]) \nonumber \\
\le & \frac{1}{2^{k-1}}
\max \Big(\hat{F}[P_{Y}](0)- \hat{F}[P_{Y}](\tau_*[P_Y]),
\hat{F}[P_{Y}](-\frac{\frac{d}{d \tau}\hat{F}[P_{Y}](0)}{\zeta_-})- \hat{F}[P_{Y}](\tau_*[P_Y])
\Big)\nonumber\\
\le & \frac{1}{2^{k-1}}
\max \Big( \frac{1}{\zeta_-}\Big(\frac{d}{d \tau}\hat{F}[P_{Y}](0) \Big)^2, \nonumber\\
&
\frac{1}{\zeta_-}\Big(\frac{d}{d \tau}\hat{F}[P_{Y}](0) \Big)^2
+  \frac{d}{d\tau}\hat{F}[P_{Y}](0)  
\cdot \frac{-\frac{d}{d \tau}\hat{F}[P_{Y}](0)}{\zeta_-}
+\frac{1}{2} \zeta_+ \cdot \Big(\frac{-\frac{d}{d \tau}\hat{F}[P_{Y}](0)}{\zeta_-}\Big)^2
\Big) \nonumber\\
= & \frac{1}{2^{k}}
 \Big(\frac{d}{d \tau}\hat{F}[P_{Y}](0)\Big)^2
\frac{ \zeta_+}{\zeta_-^2} .\Label{ZMO1}
\end{align}
The relation \eqref{GAT4} guarantees that 
\begin{align}
0\le  \tau_k[P_Y]- \tau_*[P_Y] 
\le  - \frac{1}{2^{k}}\frac{\frac{d}{d \tau}\hat{F}[P_{Y}](0)}{\zeta_-}.
\end{align}
Hence, 
\begin{align}
0 \le 
\frac{d}{d \tau}\hat{F}[P_{Y}](\tau_k[P_Y])
\le   -\frac{\zeta_+}{2^{k}}\frac{\frac{d}{d \tau}\hat{F}[P_{Y}](0)}{\zeta_-}.
\end{align}
We can choose $\kappa[P_Y] \ge 0$ as 
$0=(1-\kappa[P_Y]) \mathbb{E}_{P_{X,Y|\tau_k[P_Y] }[P_Y]}[D-d(X,Y)] 
+\kappa[P_Y] \mathbb{E}_{P_X\times {P}_{Y}}[D-d(X,Y)]
$.
Then, 
\begin{align}
0\le \kappa[P_Y] 
=& \frac{\frac{d}{d \tau}\hat{F}[P_{Y}](\tau_k[P_Y])}{
\frac{d}{d \tau}\hat{F}[P_{Y}](\tau_k[P_Y])
-\frac{d}{d \tau}\hat{F}[P_{Y}](0)}
\le
\frac{\frac{d}{d \tau}\hat{F}[P_{Y}](\tau_k[P_Y])}{
-\frac{d}{d \tau}\hat{F}[P_{Y}](0)} 
\le \frac{\zeta_+}{2^{k}\zeta_-}.
\end{align}

Then, we choose $P_{XY|k}[P_Y]$
as follows.
\begin{align}
P_{XY|k}[P_Y]:=
(1-\kappa[P_Y])P_{X,Y|\tau_k[P_Y] }[P_Y]
+\kappa[P_Y] P_X \times P_Y.
\end{align}

Since 
\begin{align}
& D(P_Y \times P_X\|P_{X,Y|\tau_k[P_Y] }[P_Y]) =
\hat{F}[P_{Y}](\tau_k[P_Y])-\hat{F}[P_{Y}](0) 
-\frac{d}{d\tau}\hat{F}[P_{Y}](0) \tau_k[P_Y] \nonumber\\
\le &
-\frac{d}{d\tau}\hat{F}[P_{Y}](0) \tau_k[P_Y] 
\le 
\frac{\big(\frac{d}{d \tau}\hat{F}[P_{Y}](0)\big)^2}{\zeta_-},
\end{align}
we have
\begin{align}
& D(P_{XY|k}[P_Y]\|P_{X,Y|\tau_k[P_Y] }[P_Y]) \nonumber\\
\le &
(1-\kappa[P_Y])D(P_{X,Y|\tau_k[P_Y] }[P_Y]\|P_{X,Y|\tau_k[P_Y] }[P_Y])
+\kappa [P_Y]D(P_X \times P_Y\|P_{X,Y|\tau_k[P_Y] }[P_Y])\nonumber\\
=&
\kappa [P_Y]D(P_X \times P_Y\|P_{X,Y|\tau_k[P_Y] }[P_Y]) \nonumber\\
\le &
\kappa [P_Y]\frac{\big(\frac{d}{d \tau}\hat{F}[P_{Y}](0)\big)^2}{\zeta_-}
\le
\frac{\zeta_+}{2^{k}\zeta_-^2}\big(\frac{d}{d \tau}\hat{F}[P_{Y}](0)\big)^2.\Label{ZMO2}
\end{align}
Given $\epsilon'>0$, we choose $k$ as
\begin{align}
k[P_Y,\epsilon']:= 
\log_2 \Big(
 \Big(\frac{d}{d \tau}\hat{F}[P_{Y}](0)\Big)^2
\frac{ \zeta_+ }{\zeta_-^2} 
\Big) 
-\log_2 \epsilon' 
\le
\log_2 \Big(
\frac{ \zeta_+^2 }{\zeta_-^2} 
\Big) 
-\log_2 \epsilon' .
\end{align}
The relations \eqref{ZMO1} and \eqref{ZMO2} guarantee 
\begin{align}
\hat{F}[P_{Y}](\tau_k[P_Y])-\hat{F}[P_{Y}](\tau_*[P_Y]) \le & \epsilon'\\
D(P_{XY|k}[P_Y]\|P_{X,Y|\tau_k[P_Y] }[P_Y]) \le & \epsilon' .
\end{align}
Combining the above discussion for the bisection method
and Algorithm \ref{protocol1Berror}, we obtain
Algorithm \ref{protocol1-2BB}.

\begin{algorithm}
\caption{em-algorithm for rate distortion}
\Label{protocol1-2BB}
\begin{algorithmic}
\STATE {Choose the initial distribution $P_Y^{(1)}$ on $\Y$.
Then, we define the initial joint distribution $P_{XY,(1)}$ as 
$P_Y^{(1)} \times P_X $;}
\REPEAT 
\STATE {\bf m-step:}\quad 
Calculate $P_{XY}^{(t+1)}$ and $\bar{P}_{XY}^{(t+1)}$ as follows.
We apply Algorithm \ref{Bisection} with $a:=0$ and $b:=
-\frac{\frac{d}{d \tau}\hat{F}[P_{Y}^{(t)}](0)
}{\zeta}$ with $k=k[P_Y^{(t)},\frac{\epsilon}{3}]$ iterations;
We choose $\bar{P}_{XY}^{(t+1)}$ and $\bar{P}_{XY}^{(t+1)}$
as 
$P_{X,Y|\tau_k[P_Y] }[P_Y]$ and 
$P_{XY|k}[P_Y]$, respectively.
\STATE {\bf e-step:}\quad 
Calculate $P_{Y}^{(t+1)}(y)$ as 
$\sum_{x \in \X}\bar{P}_{XY}^{(t+1)}(x,y)$.
\UNTIL{
$t=t_1-1$.} 
\STATE {\bf final step:}\quad 
We output the final estimate 
$P_{XY,f}^{(t_1)} :=P_{XY}^{(t_2)} \in \mathcal{M}$
by using  $t_2:= \argmin_{t=2, \ldots, t_1} 
D(P_{XY}^{(t)} \| P_{X} \times P_Y^{(t-1)})-
D(P_{XY}^{(t)} \| \bar{P}_{XY}^{(t)}) $.
\end{algorithmic}
\end{algorithm}

Since $\epsilon'$ is chosen as $\epsilon/2$ in Algorithm \ref{protocol1-2BB},
and the conditions (B0) and (B1) hold,
Theorem \ref{theo:conv2BC} guarantees
the precision \eqref{Qma2BA} with $\epsilon_1=\epsilon_2= \epsilon/3$. 
For its calculation complexity,
we have the following lemma .
\begin{lemma}\Label{XM9}
Assume the conditoins  $\zeta_-=O(1)$, \eqref{XNP}, and $\zeta_+=O(n_2^2)$.
We choose $P_Y^{(1)}$ as the uniform distribution on ${\cal Y}$.
To guarantee 
\begin{align}
I(X;Y)_{P_{XY}^{(t)}}
-\min_{W \in {\cal P}_{\Y|\X}^{d,P_X,D,\le} }I(X;Y)_{W \times P_X}
 \le \epsilon,\Label{XZT}
\end{align}
Algorithm \ref{protocol1-2BB} needs calculation complexity 
$O(\frac{n_1n_2\log n_2}{\epsilon} (\log_2 n_2+\log _2 \epsilon))$.
\end{lemma}
\begin{proof}
Each iteration in the bisection method
needs calculation complexity $O(n_1 n_2)$.
Each application of the bisection method has $O(\log_2 n_2+\log _2 \epsilon)$ iterations.
Hence, one application of the bisection method has $O(n_1n_2(\log_2 n_2+\log _2 \epsilon))$
calculation complexity.

Since $D(\theta_*\|\theta_1   )= D( W_* \times P_X  \| P_Y \times P_X ) \le \log n_2$,
the number $t_1= \frac{3 \log n_2}{\epsilon}+1$ satisfies
\begin{align}
\frac{1}{t_1-1}D(\theta_*\|\theta_1   ) \le \frac{\epsilon}{3}.
\end{align}
Since $\epsilon_1$ and $\epsilon_2$ are chosen as
$\epsilon_1=\epsilon_2= \epsilon/3$ in Algorithm \ref{protocol1-2BB},
the RHS of \eqref{Qma2BA} is upper bounded by $\epsilon$, which implies \eqref{XZT}.
In this case, the calculation complexity of
Algorithm \ref{protocol1-2BB} is $(\frac{3 \log n_2}{\epsilon}+1)\cdot
O(n_1n_2(\log_2 n_2+\log _2 \epsilon))=
O(\frac{n_1n_2\log n_2}{\epsilon} (\log_2 n_2+\log _2 \epsilon))$.
\end{proof}

Next, we compare Algorithm \ref{Bisection} and a simple application of accelerated proximal gradient method
whose performance is evaluated as \eqref{XMU2}. 
In this application of accelerated proximal gradient method,
we treat $I(X;Y)_{W \times P_X}$ as a convex function for the mixture parameter, which is composed of 
$(n_2-1)n_1$ parameters.
In this case, $L$ in \eqref{XMU} is $\zeta_+^{\frac{1}{2}}$ and
$\|x_0-x_*\|^2$ in \eqref{XMU2} is $O((n_2-1)n_1)$.
Hence, to achieve the same precision as \eqref{XZT}, 
the number of iteration is $O(n_2^{\frac{1}{2}}n_1^{\frac{1}{2}} \zeta_+^{\frac{1}{4}}\frac{1}{\epsilon} )$.
Each iteration has calculation complexity $O(n_1n_2)$.
Hence, in total, this method has calculation complexity 
$O(\frac{1}{\epsilon}n_2^{\frac{3}{2}}n_1^{\frac{3}{2}} \zeta_+^{\frac{1}{4}})
=O(\frac{1}{\epsilon}n_2^{2}n_1^{\frac{3}{2}} )$.
This is larger than the calculation complexity given in Lemma \ref{XM9}.

\begin{remark}
Next, we see what Blahut algorithm \cite{Blahut} solved in the relation to \eqref{MOT}.
For this aim, we focus on the function 
$f(D):=\min_{W \in {\cal P}_{\Y|\X}^{d,P_X,D} }I(X;Y)_{W \times P_X}$
Instead of $f(D)$, using Lagrange multiplier $\tau_0$,  
Blahut \cite{Blahut} focused on the minimization
\begin{align}
\min_{W \in {\cal P}_{\Y|\X}} \tau_0 D + I(X;Y)_{W \times P_X}
- \tau_0  \sum_{x \in \X,y\in \Y} d(x,y) (W\times P_X) (x,y) .\Label{XOZ}
\end{align}
When $\frac{d}{d D}f(D)=\tau_0$, the minimum \eqref{XOZ} equals $f(D)$.
However, finding such $\tau_0$ is not so easy.
The algorithm to find such $\tau_0$ was not given in \cite{Blahut}.
The algorithm by \cite{Blahut} to solves \eqref{XOZ} is the same as
Algorithm \ref{protocol1-2} with replacing $\bar{\tau}$ by $\tau_0$.
That is, his algorithm does not consider the condition \eqref{const1-4}.
Attaching the condition \eqref{const1-4}, our algorithm guarantees the following constraint condition \eqref{BPA}
in each iteration.
\begin{align}
\sum_{x \in \X,y\in \Y} d(x,y) P_{Y|X}\times P_X (x,y) = D \Label{BPA}.
\end{align}
The algorithm by \cite{Blahut} has calculation complexity 
$O( \frac{ n_1 n_2 \log n_2}{\epsilon})$.
While our algorithm has the additional factor $-\log \epsilon$,
this factor can be considered as the additional cost to satisfy \eqref{BPA}.
\end{remark}

\subsection{Numerical analysis for classical rate distortion without side information}\Label{S5S1-R}
To see how our algorithm works, we make numerical analysis for the case when 
$n_1=n_2=3$ and $D=1.5$.
We choose the cost function $d$ as
\begin{align}
\left(
\begin{array}{ccc}
d(1,1) & d(1,2) & d(1,3)
\\
d(2,1) & d(2,2) & d(2,3)
\\
d(3,1) & d(3,2) & d(3,3)
\end{array}
\right)
=\left(
\begin{array}{ccc}
0 & 1 & 2\\
1 & 2 & 0\\
3 & 0 & 1
\end{array}
\right),
\end{align}
and the distribution $P_X$ as
\begin{align}
P_X(1)=0.5,~
P_X(2)=0.3,~
P_X(3)=0.2.
\end{align}
We set the initial marginal distribution $P_Y^{(1)}$ to be the uniform distribution.
By applying Algorithm \ref{protocol1-2},
the mutual information 
$I(X;Y)_{P_{XY}^{(t)}}$ converges to
\begin{align}
I(X;Y)_{P_{XY}^{*}}:= 0.100039,
\end{align}
and
the conditional distribution $P_{Y|X}^{(t)}$ converges to
\begin{align}
P_{Y|X}^{*}=
\left(
\begin{array}{ccc}
0.0855598 &0.188594& 0.430983 \\
0.22431&0.494433 &0.139579 \\
0.69013& 0.316974&0.429438
\end{array}
\right).
\end{align}
In particular, the marginal distribution $P_Y^{(t)}$
converges to
\begin{align}
P_{Y}^{*}=
\left(
\begin{array}{c}
0.185555\\
0.288401 \\
0.526045
\end{array}
\right).
\end{align}
Also, the parameter $\bar{\tau}$ appearing in Algorithm \ref{protocol1-2} converges to 
$0.522814$.
Fig. \ref{Tau} shows the behavior of the parameter $\bar{\tau}$.
In addition, Fig. \ref{Error} shows that the error 
$I(X;Y)_{P_{X,Y}^{(t)}}-I(X;Y)_{P_{X,Y}^{*}}$ is much smaller than the upper bound given in 
\eqref{mma2BYY2}, which suggests the existence of a much better evaluation than  
\eqref{mma2BYY2}.

\begin{figure}[htbp]
\begin{center}
  \includegraphics[width=0.7\linewidth]{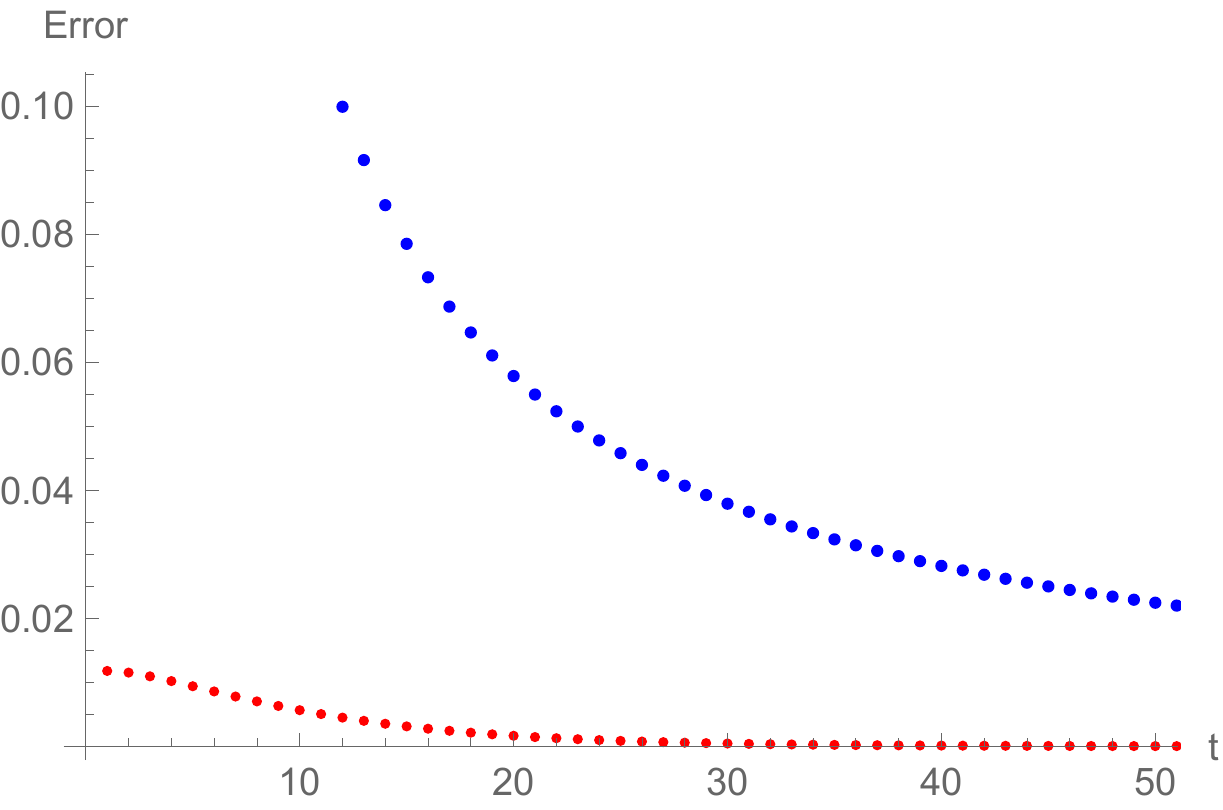}
  \end{center}
\caption{Behavior of the error 
$I(X;Y)_{P_{X,Y}^{(t)}}-I(X;Y)_{P_{X,Y}^{*}}$ of the minimum mutual information.
Red points show the value of the $I(X;Y)_{P_{X,Y}^{(t)}}-I(X;Y)_{P_{X,Y}^{*}}$
depending on the number of iteration $t$.
The blue points show its upper bound given in \eqref{mma2BYY2}.}
\Label{Error}
\end{figure}   

\begin{figure}[htbp]
\begin{center}
  \includegraphics[width=0.7\linewidth]{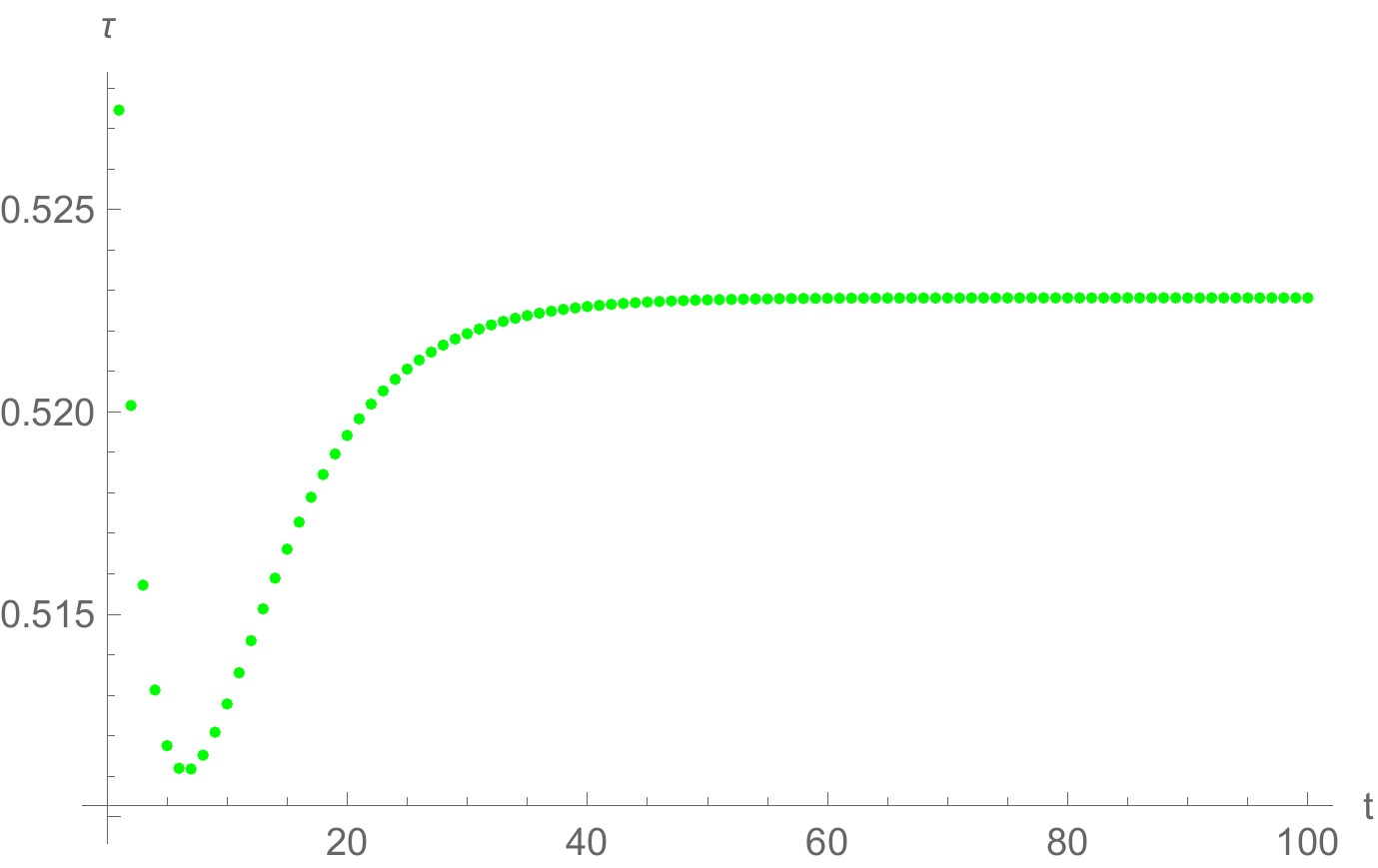}
  \end{center}
\caption{
The behavior of the parameter $\tau$ depending on the number of iteration $t$.
The green points show the parameter $\tau$ in algorithm \eqref{protocol1-2}.
}
\Label{Tau}
\end{figure}   

\subsection{Another approach to classical rate distortion without side information}\Label{S5S1-2}
To see the exponential decay, we discuss 
another approach to classical rate distortion without side information.
To apply Theorem \ref{theo:conv2+}, we need to satisfy Condition (B1+) holds.
For this aim, we apply the model given in Section \ref{4A} to the case when $\X$ is $\X\times \Y$.
Then,  we consider the Bregman divergence system $(\mathbb{R}^{n_1n_2-1}, \mu, D^\mu)$ given in Section \ref{4A}.
The set of distributions $q \times P_X$ forms an exponential family ${\cal E}$ and 
the set of distributions $W \times P_X$ forms a mixture family ${\cal M}$.
Hence, the minimization problem \eqref{MOT} is a special case of the minimization \eqref{min1}
with the formulation given in Subsection \ref{SEC-CC}.

Since the mixture family ${\cal M}$ has $ n_1(n_2-1)-1$ parameters, and the total dimension is $n_1n_2-1$,
Algorithm \ref{protocol1} is rewritten as
Algorithm \ref{protocol1-2U}.
In this case Conditions (B0) and (B1) hold in the same way as Subsection \ref{S5S1}.

\begin{algorithm}
\caption{em-algorithm for rate distortion}
\Label{protocol1-2U}
\begin{algorithmic}
\STATE {Choose the initial distribution $P_Y^{(1)}$ on $\Y$.
Then, we define the initial joint distribution $P_{XY,(1)}$ as 
$P_Y^{(1)}\times P_X $;} 
\REPEAT 
\STATE {\bf m-step:}\quad 
Calculate $P_{XY}^{(t+1)}$ as $P_{XY}^{(t+1)}(x,y):=
P_{XY,(t)}(x,y)e^{\bar{\tau}_x+ \bar{\tau}_0d(x,y))}
\Big(\sum_{x',y'}P_{XY,(t)}(x',y')e^{\bar{\tau}_{x'}+ \bar{\tau}_0d(x',y')}\Big)^{-1}$,
where $(\bar{\tau}_x)_{x \in \X}$ and $\bar{\tau}_0$ are the unique elements 
$({\tau}_x)_{x \in \X}$ and ${\tau}_0$ to satisfy 
\begin{align}
\frac{\partial}{\partial \tau_x} 
\log \Big(\sum_{x',y'}P_{XY,(t)}(x',y')e^{{\tau}_{x'}+ {\tau}_0d(x',y')}\Big)
&=P_X(x) \Label{const1-3} \\
\frac{\partial}{\partial \tau_0} 
\log \Big(\sum_{x',y'}P_{XY,(t)}(x',y')e^{{\tau}_{x'}+ {\tau}_0d(x',y')}\Big)
&=D \Label{const1-4T} 
\end{align}
for $x \in X\setminus \{n_1\}$ and 
${\tau}_{n_1} =\bar{\tau}_{n_1} $ is fixed to $0$.
This choice can be written in the way as \eqref{const1-2-0}.
\STATE {\bf e-step:}\quad 
Calculate $P_{XY,(t+1)} $ as 
$P_{Y}^{(t+1)}\times P_X$ where $P_{Y}^{(t+1)}(y):= \sum_{x \in \X}P_{XY}^{(t+1)}(x,y)$.
\UNTIL{convergence.} 
\end{algorithmic}
\end{algorithm}

The m-step in Algorithm \ref{protocol1-2U}
has optimization with $n_1$-variable convex function
\par
\noindent$\log \Big(\sum_{x',y'}P_{XY,(t)}(x',y')e^{{\tau}_{x'}+ {\tau}_0d(x',y')}\Big)$.
However, this case can satisfy Condition (B1+), which leads the exponential  
decay as follows.
That is,
the above evaluation for the convergence can be improved by using Theorem \ref{theo:conv2+}, i.e.,
the same precision \eqref{NHG2+} can be realized with $t-2 \ge \frac{ (\log \log n ) -\log \epsilon}{\log \beta}$.
In fact, when an element $\theta'$ close to $\theta^{*}$ satisfies Condition (B1+),
the iterated point $\theta^{(t)}$ converges to the true value exponentially
after the iterated point $\theta^{(t)}$ is close to the true value.

In the following, we discuss a necessary condition for (B1+) with an element $\theta'$
close to $\theta^{*}$.
When two elements are close to each other, 
the divergence can be approximated by the Fisher information. 
Hence, we consider the Fisher information version of (B1+).
For this aim, we consider the exponential family $\{P_{\theta,Y}\}$ defined in Subsection \ref{4A} with $d=n_2-1$
by replacing $\X$ by $\Y$.
Let $J_{\theta,1}$ and $J_{\theta,2}$
be the Fisher information matrices of 
$\{ \Pro^{(m),\mu}_{{\cal M}}(P_{\theta,Y}\times P_X)\}$
and
$\{ \Pro^{(e),\mu}_{{\cal E}}\circ \Pro^{(m),\mu}_{{\cal M}}(P_{\theta,Y}\times P_X)\}$.
We choose $\theta^* \in \mathbb{R}^{n_1n_2-1}$ 
corresponding to $ \Pro^{(m),\mu}_{{\cal M}}(P_{\theta_0^*,Y}\times P_X)$
in the sense of the Bregman divergence system $(\mathbb{R}^{n_1n_2-1}, \mu, D^\mu)$ given in Section \ref{4A}.
The local version of Condition (B1+) is written as 
\begin{align}
\beta J_{\theta_0^*,1} \ge  J_{\theta_0^*,2} 
\end{align}
with a constant $0<\beta <1$.
In this case, when the iterated point $\theta^{(t)}$ is close to the true minimum point,
the difference $ D^{F}(\theta^{(t)}\| \Pro^{(e),F}_{{\cal E}}  (\theta^{(t)}) ) 
-C_{\inf}(\mathcal{M},\mathcal{E})$ approaches to zero exponential rate $\log \beta^{-1}$.
Therefore, our algorithm has such an exponential convergence at the neighborhood when 
the inequality 
\begin{align}
J_{\theta_0^*,1} >  J_{\theta_0^*,2} \Label{NAP}
\end{align}
holds, i.e., $J_{\theta_0^*,1} -  J_{\theta_0^*,2} $ is strictly positive-semidefinite.

Then, we have the following theorem.
\begin{theorem}\Label{BL12}
The matrix $J_{\theta,1}-J_{\theta,2}$ 
is a strictly positive semi-definite matrix
when 
the linear space spanned by the distributions $\{W_{\theta,x}\}_{x \in X}$ 
has dimension at least $n_2$ as a function space on $\Y$.
\end{theorem}

Therefore, when the condition for Theorem \ref{BL12} holds,
Algorithm \ref{protocol1-2U} has such an exponential convergence at the neighborhood.

\begin{proofof}{Theorem \ref{BL12}}
To show Theorem \ref{BL12},
we define the parametric family $\{P_{XY,\theta,\tau}\}_{\theta,\tau}$ 
with $\theta=(\theta^i)_{i=1} ^{n_2-1}$
and $\tau =(\tau^i)_{i=0}^{n_1-1}$
as
\begin{align}
P_{XY,\theta,\tau}(x,y):=
\frac{P_{\theta,Y}(y) P_X(x) e^{\sum_{i=0}^{n_1-1} g_i(x,y) \tau_i}}
{\sum_{x'y'}P_{\theta,Y}(y') P_X(x') e^{\sum_{i=0}^{n_1-1} g_i(x',y') \tau_i}},
\end{align}
where $g_i(x,y):=\delta_{i,x}$ and $g_0(x,y):=d(x,y)$.
We define the Fisher information matrix $J_{\theta,\tau,3}$ of 
the parametric family $\{P_{XY,\theta,\tau}\}_{\theta,\tau}$.
We define the channel $W_\theta$ from $\X$ to $\Y$
as $ W_\theta\times P_X= \Pro^{(m),\mu}_{{\cal M}}(P_{\theta,Y}\times P_X)$.
Also, we choose $\tau(\theta)$ as
$W_\theta \times P_X=P_{XY,\theta,\tau(\theta)}$.
In the $n_1+n_2-1$ dimensional vector space,
we denote the projections 
to the first $n_2-1$-dimensional space corresponding to $\theta$ 
and the latter $n_1$-dimensional space corresponding to $\tau$ 
by $P_1$ and $P_2$, respectively.

Then, Theorem \ref{BL12} follows from the following two lemmas.

\begin{lemma}\Label{BL11}
The relation $\Ker P_2 J_{\theta,\tau(\theta),3} P_1=\{0\}$
holds when 
the linear space spanned by the distributions $\{W_{\theta,x}\}_{x \in X}$ 
has dimension at least $n_2$ as a function space on $\Y$.
\end{lemma}

\begin{lemma}\Label{BL10}
The matrix $J_{\theta,1}-J_{\theta,2}$ 
is a strictly positive semi-definite matrix
if and only if $\Ker P_2 J_{\theta,\tau(\theta),3} P_1=\{0\}$.
\end{lemma}
Lemmas \ref{BL11} and \ref{BL10} is shown in Appendix \ref{BBO}.
\end{proofof}

\begin{remark}
Now, we can explain why we cannot show Condition (B1+) for Algorithm \ref{protocol1-2}
in the above method.
If we apply the same discussion to 
Algorithm \ref{protocol1-2}
the projection $P_2$ is the projection to the one-dimensional space.
Hence, the condition $\Ker P_2 J_{\theta,\tau(\theta),3} P_1=\{0\}$
does not hold unless $n_2=2$.
\end{remark}

\subsection{Classical rate distortion with multiple distortion constraint without side information}
Recently, 
the paper \cite[Theorem 1]{LZP} considers a rate-distortion problem motivated by the consideration
of semantic information.
That is, it considers two sets $\hat{\X}$ and $\hat{\mathcal{S}}$ in addition to the set $\X$, 
and focus on two distortion measures $d_{\rs}(x,\hat{s})$ and  $d_{\ra}(x,\hat{x})$ for $x \in \X, \hat{x} \in \hat{\X}$ and $\hat{s}\in \hat{\mathcal{S}}$.  Then, we define the following set for channels $W: \X \to \hat{\X} \times \hat{\mathcal{S}}$ as
\begin{align}
&{\cal P}_{\hat{\X}\times \hat{\S}|\X}^{d_{\ra},d_{\rs},P_{X},D_{\ra},D_{\rs},\le}
\nonumber \\
:=& 
\Bigg\{ W \Bigg|  
\sum_{x \in \X,\hat{s} \in \hat{\S},\hat{x}\in \hat{\X}}
d_i(x,\hat{x}) W\times P_{X} (x,\hat{x},\hat{s}) \le D_i
\hbox{ for }i={\ra},{\rs}\Bigg\}.
\end{align} 
The paper \cite[Theorem 1]{LZP} addresses the following minimization problem;
\begin{align}
\min_{W \in {\cal P}_{\hat{\X}\times \hat{\S}|\X}^{d_{\ra},d_{\rs},P_{X},D_{\ra},D_{\rs},\le}}
I(X;\hat{X},\hat{S})_{W \times P_X}.\Label{NOP}
\end{align}
For its generalization, we consider a set $\Y$ and $m$ distortion measures $d_i(x,y)$ for $x \in \X, y \in \Y$ and $i=1, \ldots, m$. 
We define the following set for channels $W: \X \to \Y$ as
\begin{align}
&{\cal P}_{\Y|\X}^{ (d_i)_{i=1}^m,P_{X},(D_i)_{i=1}^m,\le}
\nonumber \\
:=& 
\Bigg\{ W \Bigg|  
\sum_{x \in \X,y\in \Y}
d_i(x,y) W\times P_{X} (x,y) \le D_i
\hbox{ for }i=1, \ldots, m\Bigg\}.
\end{align} 
Then, the following minimization problem can be regarded a generalization of \eqref{NOP}
by considering  the case with  $Y=(\hat{X},\hat{S})$ and $m=2$;
\begin{align}
&\min_{W \in {\cal P}_{\Y|\X}^{ (d_i)_{i=1}^m,P_{X},(D_i)_{i=1}^m,\le}} I(X;Y)_{W \times P_X}\nonumber \\
=&\min_{W \in {\cal P}_{\Y|\X}^{ (d_i)_{i=1}^m,P_{X},(D_i)_{i=1}^m,\le}} \min_{q \in {\cal P}_Y}
 D(W \times P_X \| q \times P_X ).
\Label{NOP2}
\end{align}
The minimization problem \eqref{NOP2} can be considered as rate distortion 
with multiple distortion functions.
Now, we focus on the Bregman divergence system 
$(\mathbb{R}^{n_1(n_2-1)}, \bar\mu, D^{\bar\mu})$ 
defined in Subsection \ref{4A2}, which coincides with 
the set of distributions $W \times P_X$.
The set of distributions $q \times P_X$ forms an exponential subfamily ${\cal E}$, and 
the subset ${\cal P}_{\Y|\X}^{ (d_i)_{i=1}^m,P_{X},(D_i)_{i=1}^m,\le}\times P_X$
forms a closed convex mixture subfamily ${\cal M}$.
Then, the minimization problem \eqref{NOP2} is a special case of the minimization \eqref{min1}
with the formulation given in Subsection \ref{SEC-CC}

Since Lemma \ref{LOS} guarantees Condition (B0) for this problem, 
Algorithm \ref{protocol1C} works for the minimization problem \eqref{NOP2}
and is rewritten as Algorithm \ref{protocol1-2B}. 
Condition (B1) can be checked in the same way as \eqref{MKP}.
In addition, similar to Algorithm \ref{protocol1Berror}
in Subsection \ref{S5S1}, 
when we cannot solve the equations \eqref{const1-4N},
Algorithm \ref{protocol1Berror2} works in this model.

\begin{algorithm}
\caption{em-algorithm for rate distortion with multiple distortion functions}
\Label{protocol1-2B}
\begin{algorithmic}
\STATE {Choose the initial distribution $P_Y^{(1)}$ on $\Y$.
Then, we define the initial joint distribution $P_{XY,(1)}$ as 
$P_Y^{(1)}\times P_X $;} 
\REPEAT 
\STATE {\bf m-step:}\quad 
For any subset $A \subset \{1, \ldots, m\}$,
Calculate $P_{XY}^{(t+1),A}$ as $P_{XY}^{(t+1),A}(x,y):=
P_X(x)P_{Y|X,(t)}(y|x)e^{\sum_{i \in A}\bar{\tau}_{A,i} d_i(x,y))}
\Big(\sum_{y'}P_{Y|X,(t)}(y'|x)e^{ \sum_{i \in A}\bar{\tau}_{A,i} d_i(x,y')}\Big)^{-1}$,
where $(\bar{\tau}_{A,i})_{i \in A}$ are the unique elements 
$({\tau}_{A,i})_{i\in A}$ to satisfy 
\begin{align}
\frac{\partial}{\partial \tau_{A,i}} 
\sum_x P_X(x)
\log \Big(\sum_{y'}P_{Y|X,(t)}(y'|x)e^{{\tau}_{x'}+ \sum_{i' \in A}{\tau}_{A,i'}d_{i'}(x,y')}\Big)
&=D_i \Label{const1-4N} 
\end{align}
for $i \in A$.
Choose $P_{XY}^{(t+1)}$ to be $P_{XY}^{(t+1),A_0}$, where
\begin{align}
A_0:= \argmin_{ A \subset \{1, \ldots, m\}} 
\left\{D(  P_{XY}^{(t+1),A} \| P_{XY,(t+1)}) 
\left|
\begin{array}{l}
 \sum_{x,y} P_{XY}^{(t+1),A} (x,y)d_i(x,y) \le D_i \\
 \hbox{ for } i=1, \ldots, m
\end{array}
\right.\right\}.
\end{align}
\STATE {\bf e-step:}\quad 
Calculate $P_{XY,(t+1)} $ as 
$P_{Y}^{(t+1)}\times P_X$ where $P_{Y}^{(t+1)}(y):= \sum_{x \in \X}P_{XY}^{(t+1)}(x,y)$.
\UNTIL{convergence.} 
\end{algorithmic}
\end{algorithm}

\subsection{Classical rate distortion with side information}
Next, we consider the rate distortion problem
when the side information state $S \in \S=\{1, \ldots, n_3\}$ is available to both the encoder and the
decoder \cite{Cover}.
Hence, our channel $W$ is given as 
a map  $\X \times \mathcal{S} \rightarrow  {\cal P}_{\Y}$.
Given a distortion measure $d(x, y)$ on $\X \times \Y$ and a distribution $P_{XS}$ on $\X\times \S$, 
we define the following sets;
\begin{align}
{\cal P}_{\Y|\X\times \S}^{d,P,_{XS}D}:=& \Big\{ W \Big|  \sum_{x \in \X,s \in \S,y\in \Y}
d(x,y) W\times P_{XS} (x,s,y) =D \Big\} \\
{\cal P}_{\Y|\X\times \S}^{d,P_{XS},D,\le}:=& \Big\{ W \Big|  \sum_{x \in \X,s \in \S,y\in \Y}
d(x,y) W\times P_{XS} (x,s,y) \le D \Big\}.
\end{align}
We define the set ${{\cal P}_{X-S-Y}  }$  of distributions on ${\cal X}\times {\cal S}\times {\cal Y}$
to satisfy the Markov chain $X-S-Y$ with the marginal distribution $P_{XS}$.
The rate distortion function is given as
\begin{align}
& \min_{W \in {\cal P}_{\Y|\X\times \S}^{d,P,D,\le} }
I(Y;X|S)_{W \times P_{XS}}\nonumber \\
=&
\min_{W \in {\cal P}_{\Y|\X\times \S}^{d,P,D,\le} }
\sum_{s \in \S}P_S(s) D(W \times P_{X|S=s}
\| (W \cdot P_{X|S=s} ) \times P_{X|S=s} )\nonumber \\
=&\min_{W \in {\cal P}_{\Y|\X\times \S}^{d,P,D,\le} } \min_{  Q \in {{\cal P}_{X-S-Y}} }
 D(W \times P_{XS} \| Q ),\Label{MOT3}
\end{align}
where $P_{X|S=s}$ is the conditional distribution on $X$ with the condition $S=s$ of $P_{XS}$.
$P_S$ is the marginal distribution on $S$ of $P_{XS}$.
Now, we apply the discussion  in Subsection \ref{4A2} to the joint system $(\X\times\S)\times \Y$.
Then,  we consider the Bregman divergence system 
$(\mathbb{R}^{n_1 n_3 (n_2-1)}, \bar{\mu}, D^{\bar{\mu}})$, which coincides with 
the set of distributions $W \times P_{XS}$.
The set ${{\cal P}_{X-S-Y}}$ forms an exponential subfamily ${\cal E}$, and 
the subset $ {\cal P}_{\Y|\X\times \S}^{d,P,_{XS}D}\times P_{XS}$ forms a mixture subfamily ${\cal M}$.
Similar to \eqref{CXP2}, 
there exists a distribution $P_{XSY} \in {\cal E}$ such that
\begin{align}
\sum_{x,y,s}P_{XYS}(x,y,s)d(x,y) \le D
\end{align}
if and only if 
\begin{align}
\sum_s P_S(s)\min_{y}d_{YS}(y,s) \le D
\Label{CXP3}
\end{align}
where $d_{YS}(y,s):= \sum_{x}P_{X|S=s}d(x,y)$.
Therefore, in the same way as Lemma \ref{TXPO}, we can show the following lemma.
\begin{lemma}\Label{TXPO2}
When \eqref{CXP3} holds, 
$\min_{W \in {\cal P}_{\Y|\X\times \S}^{d,P_{XS},D,\le} }I(X;Y|S)_{W \times P_{XS}}=0$.
Otherwise,
\begin{align}
\min_{W \in {\cal P}_{\Y|\X\times \S}^{d,P_{XS},D,\le} }I(X;Y|S)_{W \times P_{XS}}
=&
\min_{W \in {\cal P}_{\Y|\X\times \S}^{d,P_{XS},D} }I(X;Y|S)_{W \times P_{XS}}
\nonumber \\
=&
\min_{W \in {\cal P}_{\Y|\X\times \S}^{d,P,D} } \min_{  Q \in {{\cal P}_{X-S-Y}} }
 D(W \times P_{XS} \| Q ).\Label{MOTR2}
\end{align}
\end{lemma}

Due to Lemma \ref{TXPO2}, when \eqref{CXP3} does not hold,
it is sufficient to address the minimization \eqref{MOTR2}.
In the following, we discuss the minimization problem \eqref{MOTR2}, 
which is a special case of the minimization \eqref{min1}
with the formulation given in Subsection \ref{Sec:emS1}.
The mixture family $\mathcal{M}$ has $ n_1 n_3(n_2-1)-1$ parameters.
Since the total dimension is $n_1n_3(n_2-1)$,
we employ Algorithm \ref{protocol1} instead of Algorithm \ref{protocol1-0}.
Since Lemma \ref{LOST} guarantees Condition (B0) for this problem, 
Algorithm \ref{protocol1} works and is rewritten as Algorithm \ref{protocol1-3}.

\begin{algorithm}
\caption{em-algorithm for rate distortion with side information}
\Label{protocol1-3}
\begin{algorithmic}
\STATE {Choose the initial conditional distribution $P_{Y|S}^{(1)}$ on $\Y$ with the condition on ${\cal S}$.
Then, we define the initial joint distribution $P_{XYS,(1)}$ as 
$P_{Y|S}^{ (1)}\times P_{XS} $;} 
\REPEAT 
\STATE {\bf m-step:}\quad 
Calculate $P_{XYS}^{(t+1)}$ as $P_{XYS}^{(t+1)}(x,y,s):=
P_{XS}(x,s)
P_{Y|S,(t)}(y|s)e^{\bar{\tau}d(x,y))}
\Big(\sum_{y'}
P_{Y|S,(t)}(y'|s)e^{\bar{\tau}d(x,y'))}
\Big)^{-1}$,
where $\bar{\tau}$ is the unique element ${\tau}$ to satisfy 
\begin{align}
\frac{\partial}{\partial \tau} 
\sum_{x,s} P_{XS}(x,s)
\log \Big(
\sum_{y'} P_{Y|S,(t)}(y'|s)e^{{\tau}d(x,y'))}\Big)
&=D \Label{const1-6} .
\end{align}
\STATE {\bf e-step:}\quad 
Calculate $P_{XYS,(t+1)} $ as 
$P_{Y|S}^{(t+1)}\times P_{XS}$ where $P_{Y|S}^{(t+1)}(y|s):= \sum_{x \in \X}P_{XYS}^{(t+1)}(x,y,s)/
\sum_{x' \in \X,y' \in \Y}P_{XYS}^{(t+1)}(x',y',s)$.
\UNTIL{convergence.} 
\end{algorithmic}
\end{algorithm}

To check condition (B1),
we set $\theta$ and $\theta'$ be elements of $\mathbb{R}^{n_1n_3(n_2-1)}$ corresponding to 
$W \times P_{XS} $ and $  W' \times P_{XS}$.
We define the distribution $Q_{W} $ on $\X\times\S\times \Y$
as $Q_W(x,s,y)=\sum_{x'}W(y|x',s) P_{X|S=s}(x')P_S(s)$.
In the same way, we define $Q_{W'} $ on $\X\times\S\times \Y$ by replacing $W$ by $W'$. 
Then, the relations
\begin{align}
& D^{\bar{\mu}}( \Pro^{(e),\bar{\mu}}_{{\cal E}}  (\theta')\|  \Pro^{(e),\bar{\mu}}_{{\cal E}}  (\theta) ) 
=D(Q_{W'}\|Q_W) \nonumber \\
=&\sum_{s\in \S}P_S(s) D( (W_{Y|X,S=s}' \cdot P_{X|S=s} )\times P_{X|S=s} \|  
(W_{Y|X,S=s} \cdot P_{X|S=s})\times P_{X|S=s}  ) \nonumber \\
=&\sum_{s\in \S}P_S(s) D( W_{Y|X,S=s}' \cdot P_{X|S=s} \|  W_{Y|X,S=s} \cdot P_{X|S=s} ) \nonumber \\
\le & \sum_{s\in \S}P_S(s) D( W_{Y|X,S=s}' \times P_{X|S=s} \|  W_{Y|X,S=s} \times P_{X|S=s} )\nonumber \\
=& D( W' \times P_{XS} \|  W \times P_{XS})
= D^{\bar{\mu}}(\theta'\|\theta)
\end{align}
guarantee Condition (B1).
When the initial value $\theta_{(1)}$ is chosen as the case that
$W$ has full support, 
$\sup_{\theta \in \mathcal{E}}D^{\bar{\mu}}(\theta \| \theta_{(1)})$ has a finite value.
Hence, Theorem \ref{theo:conv2} guarantees the convergence to the global minimum as follows.
When we choose the initial value $\theta_{(1)}$ 
in the same way as the above case,
the precision \eqref{NHG2} holds with $t \ge \frac{\log n_2}{\epsilon}+1$.
In addition, in the same way as Subsection \ref{Sec:emS1}, we can apply Algorithm \ref{protocol1Berror}.

Next, we consider the case when 
we cannot exactly calculate the unique element
$\bar{\tau}$ to satisfy 
\eqref{const1-6}. Alternatively, we need to use 
Algorithm \ref{protocol1Berror}, which can be rewritten in the same way as Algorithm \ref{protocol1-2BB}.
That is, 
it is sufficient to replace 
$X$ by $XS$ and define 
$\hat{F}^{(t)}(\tau)$ by $\sum_{x,s} P_{XS}(x,s)
\log \Big(
\sum_{y'} P_{Y|S,(t)}(y'|s)e^{{\tau}(D-d(x,y')))}\Big)$
in Algorithm \ref{protocol1-2BB}.
When we fix the precision level $\epsilon>0$ and choose $\epsilon_1:=\frac{\epsilon}{3}$,
this algorithm achieves the precision condition \eqref{WHG2T} with $\frac{2 \log n_2}{\epsilon}+1$ rounds due to \eqref{XPZ}.
The calculation complexity can be evaluated in the same way as Algorithm \ref{protocol1-2BB}.

\section{Quantum entanglement-assisted rate distortion}\Label{S6}
Consider two quantum systems ${\cal H}_A$ and ${\cal H}_B$ with dimension $d_A$ and $d_B$.
Let ${\cal H}_R$ be the reference system of ${\cal H}_A$ with the dimension $d_A$.
We focus on a density matrix $\rho$ on ${\cal H}_A$ and 
a Hermitian matrix $\Delta$ on ${\cal H}_R \otimes {\cal H}_B$, which expresses 
our distortion measure.
Using a purification $\Psi$ of $\rho$ on ${\cal H}_{A}\otimes {\cal H}_R$, we define 
the following sets of TP-CP maps with the input system ${\cal H}_A$ and the output system ${\cal H}_B$.
\begin{align}
{\cal P}_{A\to B}^{\Delta,\rho,D}:=& 
\Big\{ {\cal N} \Big|  \Tr \Delta  (id_R \otimes {\cal N})(|\Psi\rangle \langle \Psi|) =D \Big\} \\
{\cal P}_{A\to B}^{\Delta,\rho,D,\le}:=& 
\Big\{ {\cal N} \Big|  \Tr \Delta  (id_R \otimes {\cal N})(|\Psi\rangle \langle \Psi|) \le D \Big\} .
\end{align}
The entanglement-assisted rate distortion function is given as \cite[Theorem 2]{DHW}
\begin{align}
&\min_{{\cal N} \in {\cal P}_{A\to B}^{\Delta,\rho,D,\le} }
D( (id_R \otimes {\cal N})(|\Psi\rangle \langle \Psi|) \|
(id_R \otimes {\cal N})(|\Psi\rangle \langle \Psi|)_R \otimes (id_R \otimes {\cal N})(|\Psi\rangle \langle \Psi|)_B
)\nonumber \\
=&\min_{{\cal N} \in {\cal P}_{A\to B}^{\Delta,\rho,D,\le} }
\min_{ \sigma_B\in \S({\cal H}_B)}
D( (id_R \otimes {\cal N})(|\Psi\rangle \langle \Psi|) \|
\rho_R\otimes \sigma_B) .\Label{MOT5}
\end{align}
where
$\rho_R:=\Tr_A |\Psi\rangle \langle \Psi|$.
Essentially, the above minimization handles the state
$(id_R \otimes {\cal N})(|\Psi\rangle \langle \Psi|)$.
Hence, we introduce the following sets of states on ${\cal H}_R\otimes {\cal H}_B$;
\begin{align}
{\cal S}_{R B}^{\Delta,\Psi,D}:=& 
\Big\{ \bar\rho_{RB} \Big|  \Tr \Delta \bar\rho_{RB}= D ,\quad
\bar\rho_R= \rho_R\Big\} \\
{\cal S}_{R B}^{\Delta,\Psi,D,\le}:=& 
\Big\{ \bar\rho_{RB} \Big|  \Tr \Delta \bar\rho_{RB}\le D ,\quad
\bar\rho_R= \rho_R\Big\} .
\end{align}
The minimization \eqref{MOT5} is rewritten as
\begin{align}
\min_{\bar\rho_{RB} \in {\cal S}_{R B}^{\Delta,\Psi,D,\le} }
\min_{\sigma_B\in \S({\cal H}_B)}
D( \bar\rho_{RB} \| \rho_R\otimes \sigma_B) .
\Label{MOT7}
\end{align}

Now, we apply the discussion  in Section \ref{4B} to the case when ${\cal H}$ is ${\cal H}_R \otimes {\cal H}_B$.
Then,  we consider the Bregman divergence system $(\mathbb{R}^{d_A^2 d_B^2-1}, \mu, D^\mu)$. 
The set of states $\rho_R \otimes \sigma_B$ forms an exponential family ${\cal E}$, and
the set ${\cal S}_{R B}^{\Delta,\Psi,D}$ forms a mixture family ${\cal M}$.

Similar to \eqref{CXP2}, 
there exists a state $\sigma_B$ such that
\begin{align}
\Tr \Delta \rho_R \otimes \sigma_B  \le D
\end{align}
if and only if 
\begin{align}
\lambda_{\min} (\Delta_B) \le D
\Label{CXP5}
\end{align}
where $\Delta_B:= \Tr_R  \Delta \rho_R \otimes I_B$
and $\lambda_{\min} (\Delta_B)$ expresses the minimum eigenvalue of 
$\Delta_B$.
Therefore, in the same way as Lemma \ref{TXPO}, we can show the following lemma.
\begin{lemma}\Label{TXPO5}
When \eqref{CXP5} holds, the  minimum \eqref{MOT7} equals zero.
Otherwise,
\begin{align}
&\min_{\bar\rho_{RB} \in {\cal S}_{R B}^{\Delta,\Psi,D,\le} }
\min_{\sigma_B\in \S({\cal H}_B)}
D( \bar\rho_{RB} \| \rho_R\otimes \sigma_B) \\
=&
\min_{\bar\rho_{RB} \in {\cal S}_{R B}^{\Delta,\Psi,D} }
\min_{\sigma_B\in \S({\cal H}_B)}
D( \bar\rho_{RB} \| \rho_R\otimes \sigma_B) .
\Label{MOT5O}
\end{align}
\end{lemma}

Due to Lemma \ref{TXPO5}, when \eqref{CXP5} does not hold,
it is sufficient to address the minimization \eqref{MOT5O}.
In the following, we discuss the minimization problem \eqref{MOT5O}.
To address it as a special case of the minimization \eqref{min1}
with the formulation given in Subsection \ref{Sec:emS1},
we choose $d_B^2-1$ linearly independent Hermitian matrices 
$X_{1,R},\ldots, X_{d_B^2-1}$ on ${\cal H}_R$,
and set 
${\cal H}$ to be ${\cal H}_R\otimes {\cal H}_B$.
Then, we consider the Bregman divergence system 
$(\mathbb{R}^{n_1(n_2-1)}, \mu, D^{\mu})$ 
defined in Subsection \ref{4B},
where
$X_{d_R^2 (d_B^2-1)}=\Delta$ and
$X_{d_R^2 (d_B^2-1)+1}, \ldots, X_{d_A^2 d_B^2-1}$
are 
$\theta^i X_{1,R}\otimes I_B,\ldots, \theta^i X_{d_R^2-1,R}\otimes I_B$.
Then, ${\cal S}_{R B}^{\Delta,\Psi,D}$ is given as
\begin{align}
{\cal M}:= \{\theta \in \mathbb{R}^{d_R^2d_B^2-1}|
\hbox{Conditions \eqref{AN1} and \eqref{AN2} hold.} 
\},
 \end{align}
where
\begin{align}
\Tr \rho_\theta X_{d_R^2 (d_B^2-1)}&=D, \Label{AN1} \\
\Tr \rho_\theta X_{d_R^2 (d_B^2-1)+j}&=
\Tr \rho_R X_{d_R^2 (d_B^2-1)+j} \Label{AN2} ,
\end{align}
for $j=1, \ldots, d_B^2-1$.
Also, we choose the set ${\cal E}$ as
\begin{align}
{\cal E}:=\{\theta \in \mathbb{R}^{d_R^2d_B^2-1}|
\rho_{\theta}=\rho_R \otimes \sigma_B \}.
\end{align}
The mixture family $\mathcal{M}$ has $ d_A^2(d_B^2-1)$ parameters.
Since the total dimension is $d_A^2d_B^2-1$,
we employ Algorithm \ref{protocol1} instead of Algorithm \ref{protocol1-0}.
Since Lemma \ref{LOS3} guarantees Condition (B0) for this problem, 
Algorithm \ref{protocol1} works and is rewritten as Algorithm \ref{protocol1-3B}.

Since
\begin{align}
D(\rho_\theta \| \rho_R \otimes \sigma_B)
=&
D(\rho_\theta \| \rho_R \otimes \Tr_R \rho_\theta )
+
D(\Tr_R \rho_\theta \| \sigma_B) \nonumber \\
=&
D(\rho_\theta \| \rho_R \otimes \Tr_R \rho_\theta )
+
D(\rho_R\otimes \Tr_R \rho_\theta \| \rho_R\otimes \sigma_B),
\end{align}
we find that
\begin{align}
\rho_{ \Pro^{(e),\mu}_{{\cal E}}  (\theta)}
=\rho_R \otimes \Tr_R \rho_\theta.
\end{align}
Therefore, we have
\begin{align}
 D^{\mu}( \Pro^{(e),\mu}_{{\cal E}}  (\theta')\|  \Pro^{(e),\mu}_{{\cal E}}  (\theta) ) 
=& D(\rho_R \otimes \Tr_R \rho_{\theta'}\|\rho_R \otimes \Tr_R \rho_\theta)
=D(\Tr_R \rho_{\theta'}\| \Tr_R \rho_\theta) \nonumber \\
\le & D( \rho_{\theta'}\|  \rho_\theta) 
= D^{\mu}(\theta'\|\theta),
\end{align}
which guarantees Condition (B1).
Hence, Theorem \ref{theo:conv2T} guarantees the convergence to the global minimum.
Since Conditions (B0) and (B1) hold,
Theorem \ref{theo:conv2BC} guarantees that 
Algorithm \ref{protocol1Berror} works when m-step has an error.
Since m-step of this case has $d_R^2$ parameters,
it requires more calculation amount as a convex optimization
than Algorithms \ref{protocol1-2BB} and \ref{protocol1-3}.
However, it still has small smaller calculation amount of the case when 
the original problem \eqref{MOT7} is treated as a convex optimization
because \eqref{MOT7} has $d_R^2(d_B^2-1)$ variables.

\begin{algorithm}
\caption{em-algorithm for Quantum entanglement-assisted rate distortion}
\Label{protocol1-3B}
\begin{algorithmic}
\STATE {Choose the state $\rho_{B}^{(1)}$, and set 
$\rho_{RB,(1)} $ to be $\rho_R\otimes \rho_{B}^{(1)}$.
}
\REPEAT 
\STATE {\bf m-step:}\quad 
Calculate $\rho_{RB}^{(t+1)}$ as $\rho_{RB}^{(t+1)}:=
\exp ( \log  \rho_{RB,(t)}+ \sum_i \theta^i X_i\otimes I_B
+\theta^0 \Delta
)/ \Tr \exp ( \log  \rho_{RB,(t)}+ \sum_i \theta^i X_{i,R}\otimes I_B+\theta^0 \Delta)$,
where $(\theta^i)$ are the unique elements to satisfy 
\begin{align}
\frac{\partial}{\partial \theta^i} 
\log \Tr \exp ( \log  \rho_{RB,(t)}+ \sum_i \theta^i X_i\otimes I_B+\theta^0 \Delta)
&=\Tr X_i \rho_R  \Label{const1-4B} \\
\frac{\partial}{\partial \theta^0} 
\log \Tr \exp ( \log  \rho_{RB,(t)}+ \sum_i \theta^i X_i\otimes I_B+\theta^0 \Delta)
&=D \Label{const1-6B} 
\end{align}
for $i=1, \ldots, d_R^2-1$.
\STATE {\bf e-step:}\quad 
Calculate $\rho_{RB,(t+1)} $ as $\rho_R\otimes \rho_B^{(t+1)}$, 
where $\rho_B^{(t+1)}:= \Tr_R \rho_{RB}^{(t+1)}$.
\UNTIL{convergence.} 
\end{algorithmic}
\end{algorithm}

\section{Conclusion}\Label{Sec:conc}
We have formulated em algorithm in the general framework of 
Bregman divergence, and have shown the convergence to the true value and the convergence speed
under conditions that match information-theoretical problem settings.
Then, we have applied them to the rate distortion problem
and its variants including the quantum settings.

Our em algorithm in the general framework
contains two types of minimization processes in e- and m- steps.
Due to the above property of our em algorithm,
our em algorithm has merit only when 
the optimizations in the e- and m- step are written in a form without optimization,
or are converted to simpler optimizations with a smaller number of parameters
than the original minimization problem.
Fortunately, rate distortion problem and its variants satisfy this condition.
In particular, classical rate distortion problem with and without side information
need only a one-parameter convex optimization in each iteration.

To remove the constraint \eqref{BPA}, 
existing papers for the rate distortion problem and its variants
changed the objective function by using a Lagrange multiplier,
and no preceding paper showed how to choose the Lagrange multiplier \cite{Blahut,Csiszar,Cheng,YSM}.
Indeed, the number of studies for this topic is limited 
while more papers studied channel capacities \cite{Blahut,Arimoto,Cheng,YSM,Nagaoka,Matz,Yasui,Yu,SSML,Sutter,NWS,Li,Shoji,RISB}.
Since the set of conditional distributions with the linear constraint \eqref{BPA}
forms a mixture family,
our method can be directly applied to the original objective function with
the linear constraint \eqref{BPA}.
To handle the linear constraint, 
each iteration has a convex optimization only with one variable
in m-step.
Due to this convex optimization, our algorithm has a larger calculation complexity
than the algorithm by \cite{Blahut}.
However, this difference is not so large, and can be considered as
the additional cost to exactly solve the original minimization \eqref{MOT}
instead of the modified minimization \eqref{XOZ}.


Further, since our result is written in a form of Bregman divergence,
we can expect large applicability.
That is, our results have the advantage with respect to its generality over existing methods.
To emphasize our advantage, we need to apply our method to other problems because the problems discussed in this paper are limited.
Hence, it is an interesting open problem to apply 
our em algorithm to other optimization problems.
For example, it can be expected to extend our result to the case with memory \cite{Cheng,Heegard,Aharoni} 
because various information quantities in the Markovian setting
can be written in a form of Bregman divergence \cite{NK,Nagaoka2,HW20,WH17,HW16,HW16b}.
As another future problem, it is interesting to extend our method to the optimization of 
the exponential decreasing rate in various settings, which requires the optimization of 
R\'{e}nyi mutual information
by using R\'{e}nyi version of Pythagorean theorem \cite[Lemma 3 in Suppl. Mat.]{Sharma}\cite[Lemma 2.11]{Group2}.

\section*{Acknowledgments}
The author was supported in part by the National Natural Science Foundation of China (Grant No. 62171212) and
Guangdong Provincial Key Laboratory (Grant No. 2019B121203002).
The author is very grateful to Mr. Shoji Toyota for helpful discussions.

\appendices

\section{Review of convex optimization}\Label{AA1}
In this appendix, we review several existing algorithms for the minimization of 
a differentiable convex function $F$ defined on a closed convex set $C$.
In the following, we use the notation $x_*:=\argmin_{x \in C}  F(x)$.

\subsection{Bisection method}\Label{AA1A}
First we consider the bisection method, which works with one-variable 
differentiable convex function $F$ defined on an interval $[a,b]$ \cite{BV}.

\begin{algorithm}
\caption{Bisection method}
\Label{Bisection}
\begin{algorithmic}
\STATE {Set $a_0:=a$ and $b_0:=b$;} 
\REPEAT 
\STATE {\bf $k+1$th-step:}\quad 
Set $x_k:= \frac{a_k+b_k}{2}$.
If $\frac{d}{dx}F(x_k)>0$, we set $a_{k+1}:= a_k$ and $b_{k+1}:=x_k$.
Otherwise, we set $a_{k+1}:= x_k$ and $b_{k+1}:=b_k$.
This construction guarantees the conditions $\frac{d}{dx}F(a_{k+1})\le 0$
$\frac{d}{dx}F(b_{k+1})\ge 0$.
\UNTIL{convergence.} 
\end{algorithmic}
\end{algorithm}
To see the precision, 
we define the parameter
$V_0:= \max_{x,y \in [a,b]}  |F(x)-F(y)|$.

When we use the bisection method, i.e., Algorithm \ref{Bisection}, we have
\begin{align}
F(x_k)- F(x_*) \le &\frac{V_0}{2^k} \Label{GAT1} \\
F(a_k)- F(x_*), F(b_k)- F(x_*)\le &\frac{V_0}{2^{k-1}} \Label{GAT2}\\
|x_k-x_*| \le & \frac{b-a}{2^{k+1}} \Label{GAT3}\\
 x_*-a_k, b_k-x_* \le& \frac{b-a}{2^{k}} .\Label{GAT4}
\end{align}
That is, to guarantee $|F(x_k)-  F(x_*)| \le\epsilon$, 
the number of iteration $k$ needs to satisfy $k \ge \log_2 \frac{V_0}{\epsilon} $.

\subsection{Gradient method}\Label{AA1B}
Next, we consider the gradient method, which works for
a differentiable $d$-variable convex with the uniform Lipschitz condition.
We consider a differentiable $d$-variable convex function $F$ defined on a convex set $C \subset \mathbb{R}^d$,
and assume the uniform Lipschitz condition with a constant $L$;
\begin{align}
\| \nabla F(x)- \nabla F(y)\| \le L \| x-y\| 
\end{align}
for $x,y \in C$.

\begin{algorithm}
\caption{Gradient method}
\Label{Gradient}
\begin{algorithmic}
\STATE {Set an initial value $x_0 \in C$;} 
\REPEAT 
\STATE {\bf $k+1$th-step:}\quad 
Set $x_{k+1}$ as
\begin{align}
x_{k+1}:= x_k -\frac{1}{L}\nabla F(x_k).
\end{align}
\UNTIL{convergence.} 
\end{algorithmic}
\end{algorithm}
When we use the gradient method, i.e., Algorithm \ref{Gradient}, 
we have \cite[Chapter 10]{Beck} \cite{BT,Nesterov}
\begin{align}
|F(x_k)- F(x_*)| \le \frac{L}{2k} \| x_*-x_0\|^2.\Label{XMU}
\end{align}
That is, to guarantee $|F(x_k)-  F(x_*)| \le\epsilon$, 
the number of iteration $k$ needs to satisfy $k \ge \frac{L \| x_*-x_0\|^2}{2\epsilon} $.
When we employ accelerated proximal gradient
methods, the evaluation \eqref{XMU} is improved as \cite{BT,AT,Nesterov,Nesterov2,Nesterov3,Teboulle}
\begin{align}
|F(x_k)- F(x_*)| \le \frac{L}{2(k+2)^2} \| x_*-x_0\|^2.\Label{XMU2}
\end{align}

\section{Proof of Theorem \ref{XAM}}\Label{A1} 
In this proof, we simplify $\gamma(\hat{\Theta}|{\Theta})$ to $\gamma$.
We consider the mixture subfamily
$\mathcal{M}:=
\{ \theta \in \Theta|  \exists \lambda \in \mathbb{R}, 
\eta(\theta)=
(1-\lambda) \eta (\theta_1)+\lambda \eta(\theta_2)\}$.
Due to Condition (M4), we can define the $m$-projection
$\Pro^{(m),F}_{\mathcal{M}} ( \theta_3) \in \mathcal{M}$.
We choose $\lambda$ such that 
$\Pro^{(m),F}_{\mathcal{M}} ( \theta_3)=
(1-\lambda) \eta (\theta_1)+\lambda \eta(\theta_2)$
We consider three cases;
(i) $\lambda<0$.
(ii) $0\le \lambda \le 1$.
(iii) $1<\lambda$.

Case (i);
Since the subset $\hat{\Theta}\subset \Theta$ 
is a star subset for $ \theta_1 \in \hat{\Theta}$, and 
$ \theta_2 \in \hat{\Theta}$, we have
$ \theta(s) \in \hat{\Theta}$ for $s \in [0,1]$.
Hence, we have the matrix inequality
\begin{align}
J(\theta(s))^{-1}
\le
\gamma J(\theta(1-s))^{-1}. \Label{CLP9}
\end{align}
Thus, we have 
\begin{align}
& D^F(\theta_1\|\theta_2)
\stackrel{(a)}{=}
 \int_0^1 \sum_{i=1}^d \sum_{j=1}^d 
(\eta(\theta_2)-\eta(\theta_1))_i
(\eta(\theta_2)-\eta(\theta_1))_j
(J(\theta(s))^{-1})^{i,j} s ds \nonumber \\
\stackrel{(b)}{\le}&
\gamma
 \int_0^1 \sum_{i=1}^d \sum_{j=1}^d 
(\eta(\theta_2)-\eta(\theta_1))_i
(\eta(\theta_2)-\eta(\theta_1))_j
(J(\theta(1-s))^{-1})^{i,j} s ds \nonumber \\
\stackrel{(c)}{=}&
\gamma D^F( \theta_2 \|\theta_1) \Label{BLS}
\end{align}
where $(a)$, $(b)$, and $(c)$ follow from 
\eqref{NXP}, \eqref{CLP9}, and \eqref{NXP}, respectively.

Also, we have
\begin{align}
& D^F( \theta_2 \|\theta_1)
\le 
D^F( \theta_2 \|\Pro^{(m),F}_{\mathcal{M}} (\theta_3))
\nonumber \\
\le &
D^F( \theta_2 \|\Pro^{(m),F}_{\mathcal{M}} (\theta_3))
+
D^F( \Pro^{(m),F}_{\mathcal{M}} (\theta_3)\| \theta_3 ) 
=
D^F( \theta_2 \|\theta_3)\Label{BLE}.
\end{align}
The combination of \eqref{BLS} and \eqref{BLE} yields \eqref{BLT}.

Case (iii);
We have 
\begin{align}
& D^F(\theta_1\|\theta_2)\le D^F(\theta_1 \| \Pro^{(m),F}_{\mathcal{M}} ( \theta_3)) \nonumber \\
\le &
D^F( \theta_1 \|\Pro^{(m),F}_{\mathcal{M}} (\theta_3))
+
D^F( \Pro^{(m),F}_{\mathcal{M}} (\theta_3)\| \theta_3 ) 
=
D^F( \theta_1 \|\theta_3).
\end{align}

Case (ii);
We use the quantity $M:=\Big(\max_{s \in [0,1]}\sum_{i=1}^d \sum_{j=1}^d 
(\eta(\theta_2)-\eta(\theta_1))_i
(\eta(\theta_2)-\eta(\theta_1))_j
(J(\theta(s))^{-1})^{i,j} \Big).$
Then, we have 
\begin{align}
& D^F(\theta_1\| \theta_3) \nonumber \\
= & D^F(\theta_1\|\Pro^{(m),F}_{\mathcal{M}} (\theta_3))+ D^F(\Pro^{(m),F}_{\mathcal{M}} (\theta_3)\| \theta_3) \nonumber \\
\ge & D^F(\theta_1\|\Pro^{(m),F}_{\mathcal{M}} (\theta_3))= D^F(\theta_1\|\theta(\lambda))\nonumber \\
=&
 \int_0^\lambda \sum_{i=1}^d \sum_{j=1}^d 
(\eta(\theta_2)-\eta(\theta_1))_i
(\eta(\theta_2)-\eta(\theta_1))_j
(J(\theta(s))^{-1})^{i,j} s ds \nonumber \\
\ge & \Big( \int_0^\lambda s ds\Big)
\Big(\min_{s \in [0,1]}\sum_{i=1}^d \sum_{j=1}^d 
(\eta(\theta_2)-\eta(\theta_1))_i
(\eta(\theta_2)-\eta(\theta_1))_j
(J(\theta(s))^{-1})^{i,j} \Big)\nonumber  \\
\ge & \frac{\lambda^2}{2 \gamma}M ,
\end{align}
and
\begin{align}
& D^F(\theta_2\| \theta_3) \nonumber \\
= & D^F(\theta_2\|\Pro^{(m),F}_{\mathcal{M}} (\theta_3))+ D^F(\Pro^{(m),F}_{\mathcal{M}} (\theta_3)\| \theta_3) \nonumber \\
\ge & D^F(\theta_2\|\Pro^{(m),F}_{\mathcal{M}} (\theta_3))= D^F(\theta_2\|\theta(\lambda))\nonumber \\
=&
(1-\lambda)^2 \int_0^1 \sum_{i=1}^d \sum_{j=1}^d 
(\eta(\theta_2)-\eta(\theta_1))_i
(\eta(\theta_2)-\eta(\theta_1))_j
(J(\theta( 1 -s (1-\lambda)))^{-1})^{i,j}
 s ds \nonumber \\
\ge & \frac{(1-\lambda)^2}{2}
\Big(\min_{s \in [0,1]}\sum_{i=1}^d \sum_{j=1}^d 
(\eta(\theta_2)-\eta(\theta_1))_i
(\eta(\theta_2)-\eta(\theta_1))_j
(J(\theta(s))^{-1})^{i,j}
 \Big) \nonumber \\
\ge & \frac{(1-\lambda)^2}{2 \gamma} M.
\end{align}
That is,
we obtain
\begin{align}
\lambda \le \sqrt{ \frac{2 \gamma D^F(\theta_1\| \theta_3)}{M}}, \quad
1-\lambda \le \sqrt{ \frac{2 \gamma D^F(\theta_2\| \theta_3)}{M}}. \Label{CPA}
\end{align}
Therefore, we have
\begin{align}
& D^F(\theta_1\|\theta_2)
\stackrel{(a)}{=}
 \int_0^1 \sum_{i=1}^d \sum_{j=1}^d 
(\eta(\theta_2)-\eta(\theta_1))_i
(\eta(\theta_2)-\eta(\theta_1))_j
(J(\theta(s))^{-1})^{i,j}
 s ds \nonumber \\
\stackrel{(a)}{=}&
 \int_0^\lambda \sum_{i=1}^d \sum_{j=1}^d 
(\eta(\theta_2)-\eta(\theta_1))_i
(\eta(\theta_2)-\eta(\theta_1))_j
(J(\theta(s))^{-1})^{i,j}
 s ds \nonumber \\
&+
 \int_\lambda^1 \sum_{i=1}^d \sum_{j=1}^d 
(\eta(\theta_2)-\eta(\theta_1))_i
(\eta(\theta_2)-\eta(\theta_1))_j
(J(\theta(s))^{-1})^{i,j}
 s ds \nonumber \\
\stackrel{(b)}{\le}&
D^F(\theta_1\|\theta(\lambda))
+ \Big(\int_\lambda^1 s ds \Big)
M\nonumber \\
\stackrel{(c)}{\le}&
D^F(\theta_1\|\theta_3)
+ \frac{1-\lambda^2}{2} M\nonumber \\
\stackrel{(c)}{=}&
D^F(\theta_1\|\theta_3)
+ ( (1-\lambda)^2+ \lambda (1-\lambda))
\frac{M}{2} \nonumber \\
\stackrel{(c)}{\le}&
D^F(\theta_1\|\theta_3)
+ \frac{2\gamma}{M}
\big(D^F(\theta_2\|\theta_3)
+\sqrt{D^F(\theta_1\|\theta_3) D^F(\theta_2\|\theta_3)}\big)
\frac{M}{2},
\Label{BTLS}
\end{align}
which implies \eqref{BLT}.

\section{Proof of Theorem \ref{theo:conv2}}\Label{A2}
Remember that $\theta_{(t)} $ is $\Pro^{(e),F}_{{\cal E}}(\theta^{(t)} ) $, which implies
that $\Pro^{(m),F}_{{\cal M}}(\theta_{(t)} )= \theta^{(t+1)}$.
For any $\epsilon_1>0$, we choose an element $\theta(\epsilon_1) $ of ${\cal M}$ 
such that 
$D^F(\theta(\epsilon_1)\| \Pro^{(e),F}_{{\cal E}}(\theta(\epsilon_1) ) )
\le C_{\inf}(\mathcal{M},\mathcal{E})+\epsilon_1$.
Also, let $\theta(\epsilon_1)_* $ be $\Pro^{(e),F}_{{\cal E}}(\theta(\epsilon_1) ) $.

\begin{figure}[htbp]
\begin{center}
  \includegraphics[width=0.7\linewidth]{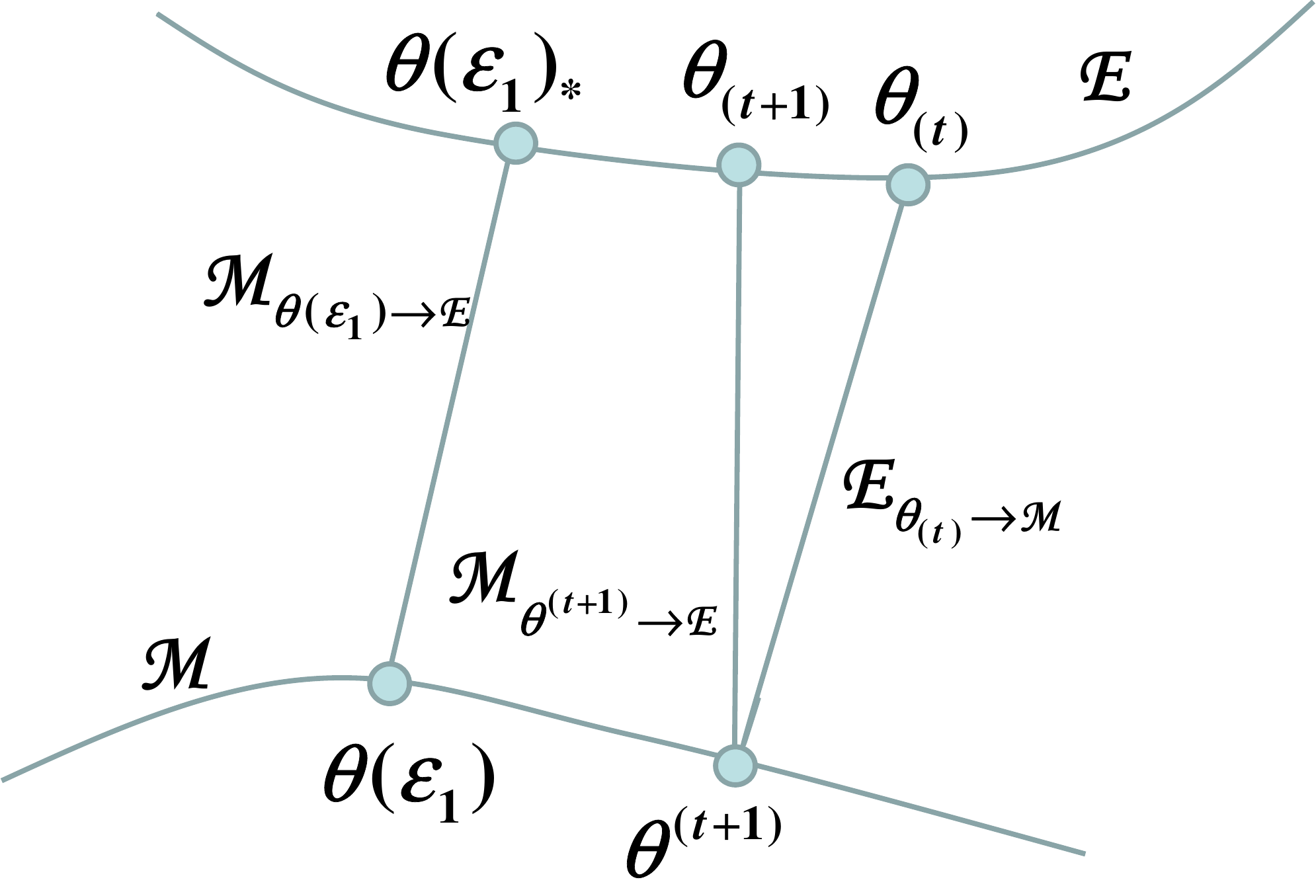}
  \end{center}
\caption{Algorithms \ref{protocol1-0} and \ref{protocol1}: 
This figure shows the topological relation among
$\theta(\epsilon_1 )_*$, $\theta(\epsilon_1 )$,
$\theta_{(t+1 )}$, $\theta^{(t+1 )}$, and $\theta_{(t)}$,
which is used in the
application of Phythagorean theorem (Proposition \ref{MNL}).
$\mathcal{M}_{\theta(\epsilon_1)\to \mathcal{E}}$
and 
$\mathcal{M}_{\theta^{(t+1)} \to \mathcal{E}}$
are the mixture subfamilies to project $\theta(\epsilon_1)$
and $\theta^{(t+1)}$
to the exponential subfamily $\mathcal{E}$, respectively.
$\mathcal{E}_{\theta_{(t)}\to \mathcal{M}}$
is the exponential subfamily to project $\theta_{(t)}$
to the mixture subfamily $\mathcal{M}$.
}
\Label{FF1}
\end{figure}   

As explained in Fig. \ref{FF1}, 
Phythagorean theorem (Proposition \ref{MNL}) guarantees that
the divergence $D^F(\theta(\epsilon_1)\| \theta_{(t)})$ can be written in the following two ways (as two equations $(a)$ and $(b)$);
\begin{align}
& D^F(\theta(\epsilon_1)\| \theta^{(t+1)})
+D^F(\theta^{(t+1)}\| \theta_{(t)})
\stackrel{(a)}{=}
D^F(\theta(\epsilon_1)\| \theta_{(t)}) \nonumber\\
\stackrel{(b)}{=}
&
D^F(\theta(\epsilon_1)\| \theta(\epsilon_1)_*)
+D^F(\theta(\epsilon_1)_*\| \theta_{(t)}).
\Label{VOMT}
\end{align}
Hence,
\begin{align}
& D^{F}(\theta^{(t+1)}\| \theta_{(t)}) 
-C_{\inf}(\mathcal{M},\mathcal{E})-\epsilon_1\nonumber \\
\le &
D^{F}(\theta^{(t+1)}\| \theta_{(t)})  
-D^F(\theta(\epsilon_1)\| \theta(\epsilon_1)_*)
\nonumber \\
\stackrel{(a)}{=}& 
D^{F}(\theta(\epsilon_1)_*\| \theta_{(t)})  
-D^{F}(\theta(\epsilon_1)\| \theta^{(t+1)})  \nonumber \\
\stackrel{(b)}{\le}&
D^{F}(\theta(\epsilon_1)_*\| \theta_{(t)})  
-D^{F}(\Pro^{(e),F}_{{\cal E}}(\theta(\epsilon_1)) \| \Pro^{(e),F}_{{\cal E}}(\theta^{(t+1)})  ) \nonumber \\
=&
D^{F}(\theta(\epsilon_1)_*\| \theta_{(t)})  
-D^{F}(\theta(\epsilon_1)_*\| \theta_{(t+1)})  \Label{VOMT3},
\end{align}
where the steps $(a)$ and $(b)$ follows from \eqref{VOMT} and Condition (B1${\cal M}$), respectively.
Thus,
\begin{align}
& \sum_{i=2}^{t}
D^{F}(\theta^{(i)}\| \theta_{(i-1)}) 
-C_{\inf}(\mathcal{M},\mathcal{E})-\epsilon_1
 \nonumber \\
\le &
\sum_{i=2}^{t}
D^{F}(\theta(\epsilon_1)_*\| \theta_{(i-1)})  
-D^{F}(\theta(\epsilon_1)_*\| \theta_{(i)})  
\nonumber \\
=&
D^{F}(\theta(\epsilon_1)_*\| \theta_{(1)})  
-D^{F}(\theta(\epsilon_1)_*\| \theta_{(t)})  
\le
 D^F(\theta(\epsilon_1)_*\| \theta_{(1)})\nonumber \\
\le &
\sup_{\theta \in \mathcal{E}} D^F(\theta \| \theta_{(1)}).\Label{suma7}
\end{align}
Taking the limit $\epsilon_1\to 0$ in \eqref{suma7}, 
we have
\begin{align}
 \sum_{i=2}^{t}
D^{F}(\theta^{(i)}\| \theta_{(i-1)}) -C_{\inf}(\mathcal{M},\mathcal{E})
\le 
\sup_{\theta \in \mathcal{E}} D^F(\theta \| \theta_{(1)}).\Label{suma8Y}
\end{align}
Since the relations
\begin{align}
D^{F}(\theta^{(i+1)}\| \Pro^{(e),F}_{{\cal E}}  (\theta^{(i+1)})) 
\le& D^{F}(\theta^{(i+1)}\| \theta_{(i)}) 
\le D^{F}(\theta^{(i)}\| \theta_{(i)}) \nonumber \\
=&
D^{F}(\theta^{(i)}\| \Pro^{(e),F}_{{\cal E}}  (\theta^{(i)}) ) 
\Label{AOP9}
\end{align}
for $i=2, \ldots, t$, 
\eqref{suma8Y} implies 
\begin{align}
(t-1) (D^{F}(\theta^{(t)}\| \theta_{(t-1)}) -C_{\inf}(\mathcal{M},\mathcal{E}))
\le 
\sup_{\theta \in \mathcal{E}} D^F(\theta \| \theta_{(1)})\Label{suma9}.
\end{align}
Thus, we have
\begin{align}
D^{F}(\theta^{(t)}\| \theta_{(t)}) -C_{\inf}(\mathcal{M},\mathcal{E})
\le &
D^{F}(\theta^{(t)}\| \theta_{(t-1)}) -C_{\inf}(\mathcal{M},\mathcal{E})
\nonumber \\
\le &
\frac{1}{t-1}\sup_{\theta \in \mathcal{E}} D^F(\theta \| \theta_{(1)})\Label{suma10},
\end{align}
which implies \eqref{mma2BYY} and \eqref{NHG2}.

When the inequality
\begin{align}
D^{F}(\theta^{(t)}\| \theta_{(t-1)}) 
-C_{\inf}(\mathcal{M},\mathcal{E})
\ge D^{F}(\theta^{(t)}\| \theta_{(t)}) 
-C_{\inf}(\mathcal{M},\mathcal{E})
\ge c(\frac{1}{t}) \Label{MU}
\end{align}
holds with a constant $c>0$, 
the relation \eqref{suma8Y} implies 
\begin{align}
\infty = \sum_{t=1}^{\infty} c(\frac{1}{t}) 
\le \sup_{\theta \in \mathcal{E}} D^F(\theta \| \theta_{(1)}),
\end{align}
which yields the contradiction.
Hence, we have 
\begin{align}
D^{F}(\theta^{(t)}\| \theta_{(t)}) 
-C_{\inf}(\mathcal{M},\mathcal{E})
= o(\frac{1}{t}) \Label{MUL}.
\end{align}
Combining \eqref{AOP9}, we obtain \eqref{mma2}.

Indeed, when the minimum in \eqref{min1} exists, i.e., 
$\theta_*(\mathcal{M},\mathcal{E})$ exists,
the supremum $\sup_{\theta \in \mathcal{E}} D^F(\theta \| \theta_{(1)})$ in the above evaluation is replaced by
$D^F(\theta_*(\mathcal{M},\mathcal{E}) \| \theta_{(1)})$
because $\theta(\epsilon_1)$ is replaced by $\theta_*(\mathcal{M},\mathcal{E})$.

\section{Proof of Theorem \ref{theo:conv2+}}\Label{A3}
We use the same notation as the proof of Theorem \ref{theo:conv2}.
We set $\beta:= \beta(\theta^{(1)})$.
The relations \eqref{VOMT3} is rewritten as the folloiwing way for the case with $\epsilon_1=0$;
\begin{align}
0\le & D^{F}(\theta^{(t+1)}\| \theta_{(t)}) 
-C_{\inf}(\mathcal{M},\mathcal{E})\nonumber \\
\le &
D^{F}(\theta_{*}\| \theta_{(t)})  
-D^{F}(\theta^* \| \theta^{(t+1)}  ) \Label{VK2}\\
\le &
D^{F}(\theta_{*}\| \theta_{(t)})  
-D^{F}(\Pro^{(e),F}_{{\cal E}}(\theta^*) \| \Pro^{(e),F}_{{\cal E}}(\theta^{(t+1)})  ) \nonumber \\
=&
D^{F}(\theta_{*}\| \theta_{(t)})  
-D^{F}(\theta_{*}\| \theta_{(t+1)})  \Label{VK1}.
\end{align}
Thus, we have 
$ D^F(\theta^*\| \theta^{(t+1)}) \stackrel{(a)}{\le}
 D^F(\theta_{*}\| \theta_{(t)})\stackrel{(b)}{\le}
 D^F(\theta_{*}\| \theta_{(t-1)}) \stackrel{(c)}{\le}
 D^F(\theta_{*}\| \theta_{(1)})$, where
 $(a)$ and $(b)$ follow from \eqref{VK2} and \eqref{VK1}, respectively, and
 $(c)$ follows from multiple use of \eqref{VK2}.
 Thus, Condition (B1+) implies $\beta D^F(\theta^*\| \theta^{(t+1)}) \ge
D^F( \theta_{*} \| \Pro^{(e),F}_{{\cal E}} (\theta^{(t+1)}))
=D^F( \theta_{*} \| \theta_{(t+1)})$.
Combining \eqref{VK2}, we have
$\beta D^F( \theta_{*} \| \theta_{(t)}) \ge D^F( \theta_{*} \| \theta_{(t+1)})$.
Thus, we have
\begin{align}
D^F(\theta_{*}\| \theta_{(t+1)})
 \le
\beta^{t} D^F(\theta_{*}\| \theta_{(1)}).
 \end{align}
Using \eqref{VK2}, we have 
\begin{align}
& 
D^{F}(\theta^{(t+1)}\| \theta_{(t+1)})  -C_{\inf}(\mathcal{M},\mathcal{E})\nonumber \\
\le &
D^{F}(\theta^{(t+1)}\| \theta_{(t)})  -C_{\inf}(\mathcal{M},\mathcal{E})\nonumber \\
\le &
 D^F(\theta_{*}\| \theta_{(t)})
 \le
\beta^{t-1} D^F(\theta_{*}\| \theta_{(1)}).
 \end{align}
Hence, we obtain \eqref{mma2+}.

\section{Proof of Theorem \ref{theo:conv2BB}}\Label{A4}
\noindent{\bf Step 1:}\quad
In this proof, 
we use the notations ${\theta}^{(t+1),*}:=\Pro^{(m),F}_{{\cal M}} ({\theta}_{(t)})$
and ${\theta}_{(t+1),*}:=\Pro^{(e),F}_{{\cal E}} (\theta^{(t+1),*})$.
From the construction, 
$D^F(\theta^{(t)}\| \theta_{(t)}  )$ is monotonically decreasing for $t$ as
\begin{align}
D^F(\theta^{(t+1)}\| \theta_{(t+1)}  )
\le D^F(\theta^{(t+1)}\| \theta_{(t)}  )
\stackrel{(a)}{\le} D^F(\theta^{(t)}\| \theta_{(t)}  ),\Label{BLY}
\end{align}
where $(a)$ follows from \eqref{NLT}.

\begin{figure}[htbp]
\begin{center}
  \includegraphics[width=0.7\linewidth]{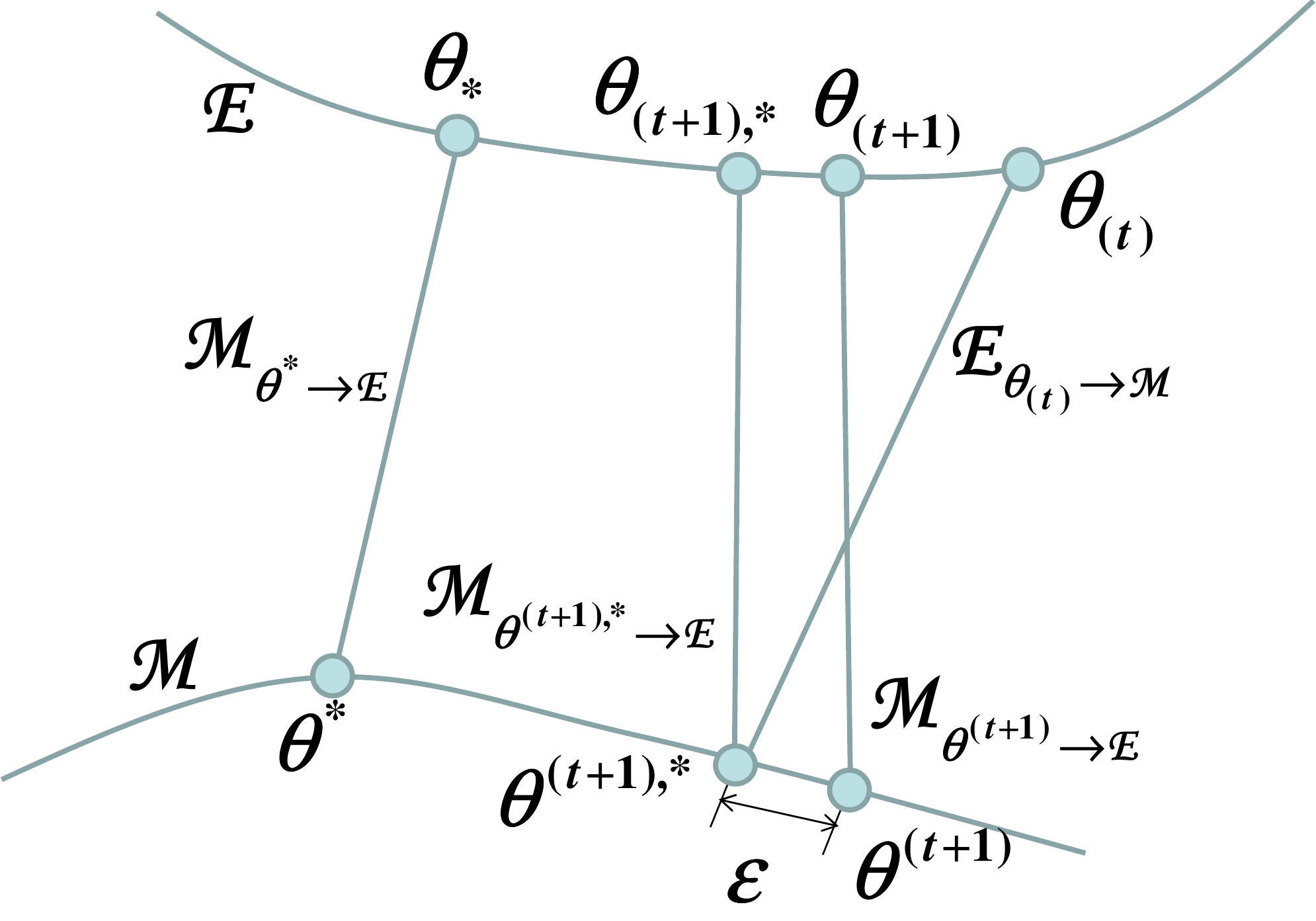}
  \end{center}
\caption{Algorithm \ref{protocol1-error}: 
This figure shows the topological relation among
$\theta_{*}$, $\theta^*$,
$\theta_{(t+1 )}$, $\theta^{(t+1 )}$, $\theta_{(t+1 ),*}$, 
$\theta^{(t+1 ),*}$, and $\theta_{(t)}$,
which is used in the
application of Phythagorean theorem (Proposition \ref{MNL}).
$\mathcal{M}_{\theta^*\to \mathcal{E}}$,
$\mathcal{M}_{\theta^{(t+1),*} \to \mathcal{E}}$,
and 
$\mathcal{M}_{\theta^{(t+1)} \to \mathcal{E}}$
are the mixture subfamilies to project $\theta^* $, $\theta^{(t+1),*}$,
and $\theta^{(t+1)}$
to the exponential subfamily $\mathcal{E}$, respectively.
$\mathcal{E}_{\theta_{(t)}\to \mathcal{M}}$
is the exponential subfamily to project $\theta_{(t)}$
to the mixture subfamily $\mathcal{M}$.
}
\Label{FF1B}
\end{figure}   

\noindent{\bf Step 2:}\quad
The aim of this step is the derivation of the relation;
\begin{align}
& D^F(\theta_{*}\| \theta_{(t)}  )-  D^F(\theta_{*}\|\theta_{(t+1)})
\nonumber \\
\ge &
D^F(\theta^{(t+1)}\| \theta_{(t+1)}  )- D^F(\theta^*\|\theta_{*})
- 2\gamma \sqrt{D^F(\theta_{*}\|\theta_{(t)})
\epsilon}-
(\gamma+1) \epsilon . \Label{CAM}
\end{align}

We notice that
\begin{align}
D^F(\theta^{(t+1)} \|\theta^{(t+1),*})
\stackrel{(a)}{=} D^F(\theta^{(t+1)} \|\theta_{(t)})
-D^F(\theta^{(t+1),*}\|\theta_{(t)})
\stackrel{(b)}{\le} \epsilon,\Label{NBTA}
\end{align}
where $(a)$ and $(b)$ follow from 
Phythagorean theorem (Proposition \ref{MNL}) and \eqref{NLT},
respectively.
Since $\theta^{(t+1),*}= \Pro^{(m),F}_{{\cal M}}  (\theta_{(t)})$, we have
\begin{align}
D^F(\theta^{(t+1),*}\| \theta_{(t)}  )
\le 
D^F(\theta^{(t)}\| \theta_{(t)}  ).\Label{BPA9}
\end{align}

Since the set $\mathcal{E}_0$ is a star subset of $\mathcal{E}$ for $\theta_{*}$,
we can apply Theorem \ref{XAM} to 
the set $\mathcal{E}_0$ as a star subset for $\theta_{*}$,
and obtain
\begin{align}
&D^F(\theta_{*}\|\theta_{(t+1)})
\nonumber \\
\stackrel{(a)}{\le} &
D^F(\theta_{*}\|\theta_{(t+1),*})+ 2\gamma \sqrt{D^F(\theta^*\|\theta_{(t+1),*})
D^F(\theta_{(t+1)} \| \theta_{(t+1),*})}\nonumber \\
& +
\gamma D^F(\theta_{(t+1)} \| \theta_{(t+1),*}) \nonumber \\
\stackrel{(b)}{\le} &
D^F(\theta_{*}\|\theta_{(t+1),*})+ 2\gamma \sqrt{D^F(\theta^*\|\theta_{(t+1),*})
D^F(\theta^{(t+1)} \| \theta^{(t+1),*})}\nonumber \\
&+
\gamma D^F(\theta^{(t+1)} \| \theta^{(t+1),*}) \nonumber \\
\stackrel{(c)}{\le} &
D^F(\theta_{*}\|\theta^{(t+1),*})+ 
2\gamma \sqrt{D^F(\theta_{*}\|\theta_{(t+1),*})
\epsilon}+
\gamma \epsilon ,
 \Label{NBT7}
\end{align}
where $(a)$, $(b)$, and $(c)$ follow from 
Theorem \ref{XAM},
Condition (B1), and \eqref{NBTA}, respectively.

The definition \eqref{min1} implies
\begin{align}
D^F(\theta^{(t+1),*}\| \theta_{(t)}  )
\ge
C_{\inf}(\mathcal{M},\mathcal{E}  )= D^F(\theta^*\|\theta_{*}).
\Label{NBT6}
\end{align}
Also, 
applying Phythagorean theorem (Proposition \ref{MNL}) 
to $ D^F(\theta^*\| \theta_{(t)}  )$, we have
\begin{align}
 D^F(\theta^*\|\theta^{(t+1),*})+
D^F(\theta^{(t+1),*}\| \theta_{(t)}  )
\stackrel{(a)}{=}
 D^F(\theta^*\| \theta_{(t)}  )
\stackrel{(b)}{=}
D^F(\theta_{*}\| \theta_{(t)}  )+ D^F(\theta^*\|\theta_{*}) .
 \Label{NBT5}
\end{align}
That is, steps $(a)$ and $(b)$ in \eqref{NBT5}
follow from Phythagorean theorem.
Using \eqref{NBT5}, we have
\begin{align}
0 \stackrel{(a)}{\le} 
& 
D^F(\theta^{(t+1),*}\| \theta_{(t)}  )- D^F(\theta^*\|\theta_{*}) \nonumber \\
\stackrel{(b)}{=}&
 D^F(\theta_{*}\| \theta_{(t)}  )-  D^F(\theta^*\|\theta^{(t+1),*})\nonumber \\
\stackrel{(c)}{\le}&
 D^F(\theta_{*}\| \theta_{(t)}  )-  D^F(\theta_{*}\|\theta_{(t+1),*})\nonumber \\
\stackrel{(d)}{\le}&
 D^F(\theta_{*}\| \theta_{(t)}  )-  D^F(\theta_{*}\|\theta_{(t+1)})
 + 2\gamma \sqrt{D^F(\theta_{*}\|\theta_{(t+1),*})
\epsilon}+
\gamma \epsilon 
\nonumber  \\
\stackrel{(e)}{\le}&
 D^F(\theta_{*}\| \theta_{(t)}  )-  D^F(\theta_{*}\|\theta_{(t+1)})
 + 2\gamma \sqrt{D^F(\theta_{*}\|\theta_{(t)})
\epsilon}+
\gamma \epsilon  ,
\end{align}
where
$(a)$, $(b)$, $(c)$, and $(d)$ follow from 
\eqref{NBT6}, 
\eqref{NBT5}, Condition (B1), and 
\eqref{NBT7}, respectively.
The final step $(e)$ is derived by the inequality
$ D^F(\theta_{*}\| \theta_{(t)}  )-  D^F(\theta_{*}\|\theta_{(t+1),*})\ge 0$, which can be shown from 
$(a)$ and $(b)$.
Comparing the RHS of $(a)$ and the final term, we have
\begin{align}
& 
D^F(\theta^{(t+1),*}\| \theta_{(t)}  )- D^F(\theta^*\|\theta_{*}) \nonumber \\
{\le}&
 D^F(\theta_{*}\| \theta_{(t)}  )-  D^F(\theta_{*}\|\theta_{(t+1)})
 + 2\gamma \sqrt{D^F(\theta_{*}\|\theta_{(t)})
\epsilon}+
\gamma \epsilon  .\Label{NBT8}
\end{align}
In addition, $D^F(\theta^{(t+1),*}\| \theta_{(t)}  )$ can be evaluated as 
\begin{align}
D^F(\theta^{(t+1)}\| \theta_{(t+1)}  )
\stackrel{(a)}{\le}&
D^F(\theta^{(t+1)}\| \theta_{(t)}  )
\stackrel{(b)}{=}
D^F(\theta^{(t+1)}\|\theta^{(t+1),*})
+D^F(\theta^{(t+1),*}\| \theta_{(t)}  )
\nonumber \\
\stackrel{(c)}{\le}&
\epsilon
+D^F(\theta^{(t+1),*}\| \theta_{(t)}  )\Label{NBT9},
\end{align}
where
$(a)$, $(b)$, and $(c)$ follow from 
the relation $\theta_{(t+1)} = \Pro^{(e),F}_{{\cal E}} (\theta^{(t+1)})$,
the relation $\theta^{(t+1),*} = 
\Pro^{(m),F}_{{\cal M}} (\theta_{(t)})$,
and \eqref{NBTA}, respectively.

Combining the above relations, we have
\begin{align}
& D^F(\theta_{*}\| \theta_{(t)}  )-  D^F(\theta_{*}\|\theta_{(t+1)})\nonumber \\
\stackrel{(a)}{\ge}&
D^F(\theta^{(t+1),*}\| \theta_{(t)}  )- D^F(\theta^*\|\theta_{*})
- 2\gamma \sqrt{D^F(\theta_{*}\|\theta_{(t)})
\epsilon}-
\gamma \epsilon \nonumber \\
\stackrel{(b)}{\ge}&
D^F(\theta^{(t+1)}\| \theta_{(t+1)}  )- D^F(\theta^*\|\theta_{*})
- 2\gamma \sqrt{D^F(\theta_{*}\|\theta_{(t)})
\epsilon}-
(\gamma+1) \epsilon . \Label{CAM2}
\end{align}
where
$(a)$ and $(b)$ follow from 
\eqref{NBT8} and \eqref{NBT9}, respectively.
Hence, we obtain \eqref{CAM}.

\noindent{\bf Step 3:}\quad
The aim of this step is showing 
\begin{align}
 D^F(\theta_{*}\| \theta_{(t)}  )-  D^F(\theta_{*}\|\theta_{(t+1)})
\ge 0 \Label{BAP}
\end{align}
for $t \le t_0$ by induction
when we assume that $t_0$ satisfies the following condition with $t\le t_0$;
\begin{align}
D^F(\theta^{(t)}\| \theta_{(t)}  )-D^F(\theta^*\|\theta_{*}) 
\ge 
2\gamma \sqrt{D^F(\theta_{*}\|\theta_{(1)})
\epsilon}+
(\gamma+1) \epsilon.\Label{BLA}
\end{align}
Due to the assumption of induction,
we have 
\begin{align}
D^F(\theta_{*}\|\theta_{(t)}) \le D^F(\theta_{*}\|\theta_{(1)})\Label{XAO5}.
\end{align}
The combination of \eqref{CAM}, \eqref{BLA}, and \eqref{XAO5} implies the relation \eqref{BAP}.

\noindent{\bf Step 4:}\quad
The aim of this step is showing 
\begin{align}
&
D^F(\theta^{(t_0+1)}\| \theta_{(t_0+1)}  ) - D^F(\theta^*\|\theta_{*}) \nonumber \\
\le & 
 \frac{D^F(\theta_{*}\| \theta_{(1)}  )}{t_0}
 + 2\gamma \sqrt{D^F(\theta_{*}\|\theta_{(1)})
\epsilon}+
(\gamma+1) \epsilon  \Label{XOA}.
\end{align}
If there exists a number $t \le t_0$ that does not satisfy 
the condition \eqref{BLA},
we have \eqref{XOA} as
\begin{align}
&
D^F(\theta^{(t_0+1)}\| \theta_{(t_0+1)}  ) - D^F(\theta^*\|\theta_{*}) \nonumber \\
\le & 
D^F(\theta^{(t)}\| \theta_{(t)}  ) - D^F(\theta^*\|\theta_{*}) \nonumber \\
<&
2\gamma \sqrt{D^F(\theta_{*}\|\theta_{(1)})
\epsilon}+
(\gamma+1) \epsilon \nonumber \\
\le &
 \frac{D^F(\theta_{*}\| \theta_{(1)}  )}{t_0}
 + 2\gamma \sqrt{D^F(\theta_{*}\|\theta_{(1)})
\epsilon}+
(\gamma+1) \epsilon  .
\end{align}
Hence, it is sufficient to show \eqref{XOA}
under the assumption  \eqref{BLA} with $t \le t_0$.

Using the facts shown above, under this assumption, we have
\begin{align}
&D^F(\theta^{(t+1)}\| \theta_{(t+1)}  ) - D^F(\theta^*\|\theta_{*})\nonumber  \\
\stackrel{(a)}{\le}&
 D^F(\theta_{*}\| \theta_{(t)}  )-  D^F(\theta_{*}\|\theta_{(t+1)})
 + 2\gamma \sqrt{D^F(\theta_{*}\|\theta_{(t)})
\epsilon}+
(\gamma+1) \epsilon \nonumber \\
\stackrel{(b)}{\le}&
 D^F(\theta_{*}\| \theta_{(t)}  )-  D^F(\theta_{*}\|\theta_{(t+1)})
 + 2\gamma \sqrt{D^F(\theta_{*}\|\theta_{(1)})
\epsilon}+
(\gamma+1) \epsilon ,\Label{BMO}
\end{align}
where $(a)$ and $(b)$
follow from \eqref{CAM} and \eqref{BAP}, respectively.

Taking the sum for \eqref{BMO}, we have
\begin{align}
&
t_0 \Big(D^F(\theta^{(t_0+1)}\| \theta_{(t_0+1)}  ) - D^F(\theta^*\|\theta_{*}) \Big)\nonumber \\
\stackrel{(a)}{\le}&
\sum_{t=1}^{t_0}
\Big(D^F(\theta^{(t+1)}\| \theta_{(t+1)}  ) - D^F(\theta^*\|\theta_{*})\Big) \nonumber \\
\stackrel{(b)}{\le}&
\sum_{t=1}^{t_0}
\Big( D^F(\theta_{*}\| \theta_{(t)}  )-  D^F(\theta_{*}\|\theta_{(t+1)})
 + 2\gamma \sqrt{D^F(\theta_{*}\|\theta_{(1)})
\epsilon}+
(\gamma+1) \epsilon \Big)\nonumber  \\
= & 
 D^F(\theta_{*}\| \theta_{(1)}  )-  D^F(\theta_{*}\|\theta_{(t_0+1)})
 + 2t_0\gamma \sqrt{D^F(\theta_{*}\|\theta_{(1)})
\epsilon}+
t_0(\gamma+1) \epsilon \nonumber  \\
\le & 
 D^F(\theta_{*}\| \theta_{(1)}  )
 + 2t_0\gamma \sqrt{D^F(\theta_{*}\|\theta_{(1)})
\epsilon}+
t_0(\gamma+1) \epsilon  ,
\end{align}
where $(a)$ and $(b)$
follow from \eqref{BLY} and \eqref{BMO}, respectively.
Hence, we have \eqref{XOA}.

\noindent{\bf Step 5:}\quad
Finally, we derive \eqref{NHG2T} from \eqref{mma2BA}.
The condition $t \ge \frac{ 2 D^F(\theta_{*,1} \| \theta_{(1)})}{\epsilon'}+1$
implies 
$\frac{D^F(\theta_{*}\| \theta_{(1)}  )}{t}
\le \epsilon'$.
The condition 
$\epsilon \le \frac{{\epsilon'}^2}{4 (3 \gamma+1)^2 D^F(\theta_{*}\|\theta_{(1)})}$
implies 
$(3\gamma+1) \sqrt{D^F(\theta_{*}\|\theta_{(1)}) \epsilon}
\le \frac{\epsilon'}{2}
$.
Since $ D^F(\theta_{*}\|\theta_{(1)}) \ge \epsilon$
and $\gamma>1$, 
we have 
$ + 2\gamma \sqrt{D^F(\theta_{*}\|\theta_{(1)})
\epsilon}+
(\gamma+1) \epsilon  \le \frac{\epsilon'}{2}$.
Thus, we obtain \eqref{NHG2T}.

\section{Proof of Theorem \ref{theo:conv2BC}}\Label{A5}
\noindent{\bf Step 1:}\quad
We define 
$\theta_{*}:= \Pro^{(e),F}_{{\cal E}} (\theta^*)$.
The aim of this step is showing the inequality \eqref{XP8}.
The condition \eqref{NKD} implies that 
\begin{align}
 F ( \bar\theta^{(t+1)} ) 
- \sum_{{j}=k+1}^d (\bar\theta^{(t+1)})^j a_{j} 
\le 
 F (\theta^{(t+1,*)}) 
- \sum_{{j}=k+1}^d (\theta^{(t+1,*)})^j a_{j} +\epsilon_1.
\Label{NMA}
\end{align}
Hence, 
\begin{align}
& D^F(\theta^{(t+1),*}\|\bar\theta^{(t+1)} ) 
=\sum_{i=1}^d \eta_i (\theta^{(t+1),*})(\theta^{(t+1),*} - \bar\theta^{(t+1)})^i -
F(\theta^{(t+1),*})+F(\bar\theta^{(t+1)} )
\le \epsilon_1\Label{LFS4}.
\end{align}

\begin{figure}[htbp]
\begin{center}
  \includegraphics[width=0.7\linewidth]{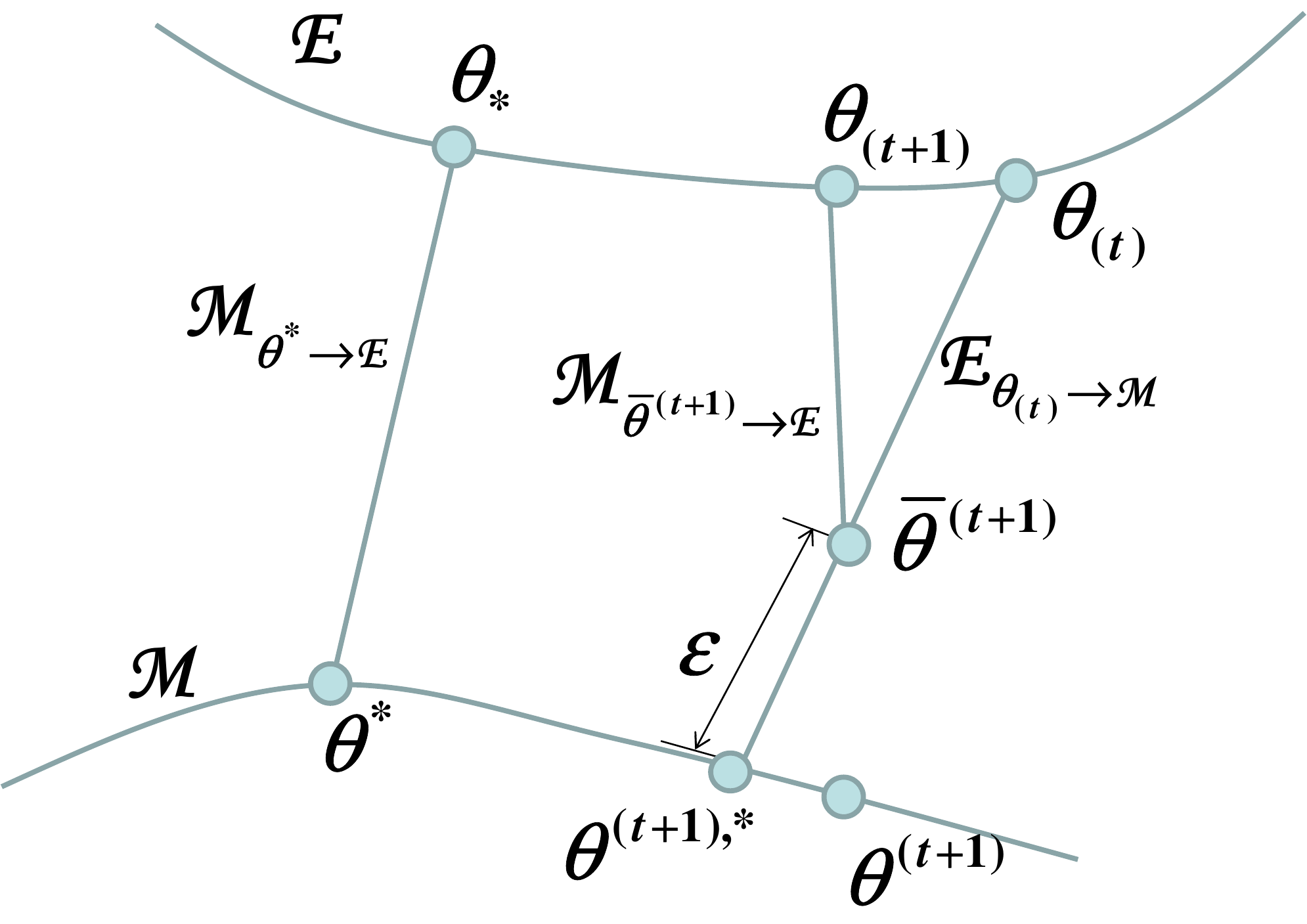}
  \end{center}
\caption{Algorithm \ref{protocol1Berror}: 
This figure shows the topological relation among
$\theta_{*}$, $\theta^*$,
$\theta_{(t+1 )}$, $\theta^{(t+1 )}$, $\bar\theta^{(t+1 )}$, 
$\theta^{(t+1 ),*}$, and $\theta_{(t)}$,
which is used in the
application of Phythagorean theorem (Proposition \ref{MNL}).
$\mathcal{M}_{\theta^*\to \mathcal{E}}$,
$\mathcal{M}_{\theta^{(t+1),*} \to \mathcal{E}}$,
and 
$\mathcal{M}_{\theta^{(t+1)} \to \mathcal{E}}$
are the mixture subfamilies to project $\theta^* $, $\theta^{(t+1),*}$,
and $\theta^{(t+1)}$
to the exponential subfamily $\mathcal{E}$, respectively.
$\mathcal{E}_{\theta_{(t)}\to \mathcal{M}}$
is the exponential subfamily to project $\theta_{(t)}$
to the mixture subfamily $\mathcal{M}$.
}
\Label{FF1C}
\end{figure}   

\noindent{\bf Step 2:}\quad
The aim of this step is showing 
\begin{align}
 D^F( \theta^{(t_3),*} \| \theta_{(t_3-1)})
-D^F(\theta^*\| \theta_{*})
\le \frac{1}{t_1-1} D^F( \theta_{*} \| \theta_{(1)})+  \epsilon_1 \Label{CO4}
\end{align}
under the choice of $t_3:= \argmin_{2\le t \le t_1}D^F( \theta^{(t),*} \| \theta_{(t-1)})$.

Pythagorean theorem (Proposition \ref{MNL}) implies that
\begin{align}
D^F(\theta^*\| \theta^{(t+1),*})
+
D^F( \theta^{(t+1),*} \| \theta_{(t)})
=
D^F(\theta^*\|  \theta_{(t)})
=
D^F(\theta^*\| \theta_{*})
+
D^F( \theta_{*} \| \theta_{(t)}).\Label{CO3}
\end{align}
Using the result of Step 1 and various formulas,  we have
\begin{align}
& D^F( \theta_{*} \| \theta_{(t)})-D^F( \theta_{*} \| \theta_{(t+1)}) \nonumber \\
\stackrel{(a)}{\ge} & D^F( \theta_{*} \| \theta_{(t)})-D^F(\theta^*\| \bar\theta^{(t+1)})
\stackrel{(b)}{=}D^F( \theta_{*} \| \theta_{(t)})-D^F(\theta^*\| \theta^{(t+1),*})-D^F(\theta^{(t+1),*} \| \bar\theta^{(t+1)}) \nonumber \\
\stackrel{(c)}{=} &
D^F( \theta^{(t+1),*} \| \theta_{(t)})
-D^F(\theta^*\| \theta_{*})
-D^F(\theta^{(t+1),*} \| \bar\theta^{(t+1)})
\nonumber \\
\stackrel{(d)}{\ge} &
D^F( \theta^{(t+1),*} \| \theta_{(t)})
-D^F(\theta^*\| \theta_{*})
-\epsilon_1 ,\Label{CO2}
\end{align}
where each step is derived as follows.
Step $(a)$ follows from Condition (B1). 
Step $(b)$ follows from Pythagorean theorem (Proposition \ref{MNL}).
Step $(c)$ follows from \eqref{CO3}.
Step $(d)$ follows from \eqref{LFS4}.

We choose $t_3:= \argmin_{2\le t \le t_1}D^F( \theta^{(t),*} \| \theta_{(t-1)})$.
Hence, for $t \le t_1-1$, we have
\begin{align}
D^F( \theta^{(t_3),*} \| \theta_{(t_3-1)})
-D^F(\theta^*\| \theta_{*}) -\epsilon_1 
\le & D^F( \theta_{*} \| \theta_{(t)})-D^F( \theta_{*} \| \theta_{(t+1)}).\Label{CO1T}
\end{align}
Taking the sum for \eqref{CO1T}, we have
\begin{align}
& D^F( \theta^{(t_3),*} \| \theta_{(t_3-1)})
-D^F(\theta^*\| \theta_{*}) -\epsilon_1 \nonumber \\
\le & \frac{1}{t_1-1}\sum_{t=1}^{t=t_1-1}D^F( \theta_{*} \| \theta_{(t)})-D^F( \theta_{*} \| \theta_{(t+1)}) \nonumber \\
=& \frac{1}{t_1-1} (D^F( \theta_{*} \| \theta_{(1)})-D^F( \theta_{*} \| \theta_{(t_1)}) )
\le \frac{1}{t_1-1} D^F( \theta_{*} \| \theta_{(1)}).
\end{align}
Therefore, we obtain \eqref{CO4}.

\noindent{\bf Step 3:}\quad
The aim of this step is showing the following inequality;
\begin{align}
&
D^F(\theta^{(t_2)}\| \theta_{(t_2-1)}  )
-D^F(\theta^*\| \theta_{*}) \nonumber \\
{\le} & \frac{1}{t_1-1} D^F( \theta_{*} \| \theta_{(1)})+  \epsilon_1 +D^F(\theta^{(t_2)} \| \bar\theta^{(t_2)}).\Label{CO8T}
\end{align}
Remember that the final estimate $\theta_f^{(t_1)}$
is defined as $\theta^{(t_2)} \in \mathcal{M}$ 
by using  $t_2= \argmin_{t=2, \ldots, t_1} 
D^F(\theta^{(t)} \| \theta_{(t-1)})-D^F(\theta^{(t)} \| \bar\theta^{(t)})$.
Then, Eq. \eqref{CO8T} is shown as follows.
\begin{align}
&
D^F(\theta^{(t_2)}\| \theta_{(t_2-1)}  )-D^F(\theta^{(t_2)} \| \bar\theta^{(t_2)})
\nonumber \\
\stackrel{(a)}{\le} &
D^F(\theta^{(t_3)}\| \theta_{(t_3-1)}  )-D^F(\theta^{(t_3)} \| \bar\theta^{(t_3)}) \nonumber \\
\stackrel{(b)}{=} & 
D^F(\theta^{(t_3)}\| \theta^{(t_3),*}  )
+D^F( \theta^{(t_3),*} \| \theta_{(t_3-1)})-D^F(\theta^{(t_3)} \| \bar\theta^{(t_3)})\nonumber \\
{\le} &
D^F(\theta^{(t_3)}\| \theta^{(t_3),*}  )
+D^F( \theta^{(t_3),*} \| \bar\theta^{(t_3)})
+D^F( \theta^{(t_3),*} \| \theta_{(t_3-1)})-D^F(\theta^{(t_3)} \| \bar\theta^{(t_3)})
\nonumber \\
\stackrel{(c)}{=} &
D^F( \theta^{(t_3),*} \| \theta_{(t_3-1)})
\nonumber \\
\stackrel{(d)}{\le} & \frac{1}{t_1-1} D^F( \theta_{*} \| \theta_{(1)})+  \epsilon_1 +D^F(\theta^*\| \theta_{*}),\Label{CO8}
\end{align}
where each step is derived as follows.
Step $(a)$ follows from the definition of $t_2$.
Steps $(b)$ and $(c)$ follow from Pythagorean theorem (Proposition \ref{MNL}) for $D^F(\theta^{(t_3)}\| \theta_{(t_3-1)}  )$ and 
$D^F(\theta^{(t_3)} \| \bar\theta^{(t_3)})$, respectively.
Step $(d)$ follows from \eqref{CO4}.

\noindent{\bf Step 4:}\quad
The aim of this step is showing Eq. \eqref{Qma2BA}.
Eq. \eqref{Qma2BA} is shown as follows;
\begin{align}
&
D^F(\theta_f^{(t_1)}\| \Pro^{(e),F}_{{\cal E}} (\theta_f^{(t_1)})  )
-D^F(\theta^*\| \theta_{*})  \nonumber \\
=&
D^F(\theta^{(t_2)}\| \Pro^{(e),F}_{{\cal E}} (\theta^{(t_2)})  )
-D^F(\theta^*\| \theta_{*}) \nonumber \\
\stackrel{(a)}{\le} &
D^F(\theta^{(t_2)}\| \theta_{(t_2-1)}  )
-D^F(\theta^*\| \theta_{*}) \nonumber \\
\stackrel{(b)}{\le} & \frac{1}{t_1-1} D^F( \theta_{*} \| \theta_{(1)})+  \epsilon_1 +D^F(\theta^{(t_2)} \| \bar\theta^{(t_2)}),
\Label{CO8A}
\end{align}
where
Step $(a)$ follows from the definition of $\Pro^{(e),F}_{{\cal E}} (\theta^{(t_2)})$
and 
Step $(b)$ follows from \eqref{CO8T}.

\section{Proofs of Theorems \ref{theo:conv2T}, \ref{theo:conv2+T}, and \ref{theo:conv2BB2}}\Label{A5+1}

\begin{proofof}{Theorem \ref{theo:conv2T}}
Theorem \ref{theo:conv2}
is shown by application of Phythagorean theorem (Proposition \ref{MNL}) to $m$-projection to $\mathcal{M}$.
We can show Theorem \ref{theo:conv2T} in the same way as
the proof of Theorem \ref{theo:conv2}
by replacing the role of Proposition \ref{MNL} by Lemma \ref{LMGO}
In this case, the proof of Theorem \ref{theo:conv2T} is completed
by replacing the equations 
at $(a)$ of \eqref{VOMT} and $(a)$ of \eqref{VOMT3}
by the inequality $\le$.
\end{proofof}

\begin{proofof}{Theorem \ref{theo:conv2+T}}
In the proof of Theorem \ref{theo:conv2+},
Phythagorean theorem is applied to $m$-projection to $\mathcal{M}$. However, this theorem is used only in the derivation for \eqref{VK2},
which is essentially given in \eqref{VOMT}.
In the current setting, 
the step $(a)$ of \eqref{VOMT} is derived by 
Lemma \ref{LMGO} instead of Proposition \ref{MNL}.
Hence, the proof of Theorem \ref{theo:conv2+T} is completed.
\end{proofof}

\begin{proofof}{Theorem \ref{theo:conv2BB2}}
We can show Theorem \ref{theo:conv2BB2} in the same way as
the proof of Theorem \ref{theo:conv2BB}
by replacing the role of Proposition \ref{MNL} by Lemma \ref{LMGO}
for Phythagorean theorem to the projeciton to $m$-projection to $\mathcal{M}$.
In this case, the proof of Theorem \ref{theo:conv2BB2} is completed
by replacing the equations 
at $(a)$ of \eqref{NBTA}, 
$(a)$ of \eqref{NBT5}, and $(b)$ of \eqref{NBT8}
by the inequality $\le$.
\end{proofof}

\section{Proof of Theorem \ref{theo:conv2BC2}}\Label{A6}
\noindent{\bf Step 1:}\quad
To show Theorem \ref{theo:conv2BC2}, we prepare the following lemma.
\begin{lemma}\Label{L-XLP}
Assume the same assumption as Algorithm \ref{protocol1Berror2}.
Also, we assume Conditions (B0) and (B1) for $\mathcal{E}$. 
When the relation $C_{\inf}(\mathcal{M}_{\lambda},\mathcal{E})
=C_{\inf}(\hat{\mathcal{M}}_{\lambda},\mathcal{E})
$ holds for $\lambda \in \Lambda_*$, 
for $\theta_0 \in \hat{\mathcal{M}}_{\lambda} \setminus {\mathcal{M}}_{\lambda}$, 
we have 
\begin{align}
\min_{\lambda' \in \Lambda_\lambda }
C_{\inf}(\mathcal{M}_{\lambda'},\mathcal{E})
\le
D^F({\theta}_0 \| \Pro^{(e),F}_{{\cal E}}  ({\theta}_0)).
\end{align}
\end{lemma}

\begin{proofof}{Lemma \ref{L-XLP}}
Lemma \ref{theo:conv2} guarantees that there is no local minimum
for the minimization \par
\noindent$\min_{\theta \in \mathcal{M}_{\lambda}}
D^F({\theta} \| \Pro^{(e),F}_{{\cal E}}  ({\theta}))$.
Hence, there exists a one-parameter continuous curve $\theta(s) \in \mathcal{M}_{\lambda}$ such that
$\theta(0)= {\theta}_0$, 
\begin{align}
\lim_{s\to 1} D^F\big({\theta}(s) \big\| \Pro^{(e),F}_{{\cal E}}  ({\theta}(s)) \big)
=C_{\inf}(\mathcal{M}_{\lambda},\mathcal{E}),
\end{align}
and $D^F({\theta}(s) \| \Pro^{(e),F}_{{\cal E}}  ({\theta}(s)))$ is monotonically increasing for $s$.
Then, there exits $s_0 \in (0,1)$ such that 
${\theta}(s_0) \in \partial \mathcal{M}_{\lambda}$.
We choose $\lambda'' \in \Lambda_\lambda $ such that
${\theta}(s_0) \in \mathcal{M}_{\lambda''}$.
Then, we obtain
\begin{align}
 \min_{\lambda' \in \Lambda_\lambda }
C_{\inf}(\mathcal{M}_{\lambda'},\mathcal{E})
\le &
C_{\inf}(\mathcal{M}_{\lambda''},\mathcal{E})
\le
D^F \big({\theta}(s_0) \big\| \Pro^{(e),F}_{{\cal E}}  ({\theta}(s_0)) \big) \nonumber \\
\le &
D^F \big({\theta}_0 \big\| \Pro^{(e),F}_{{\cal E}}  ({\theta}_0)\big).
\end{align}
\end{proofof}

In the following, we show Theorem \ref{theo:conv2BC2} by using Lemma \ref{L-XLP} and
Eq. \eqref{CO8T} in the proof of Theorem \ref{theo:conv2BC}.

\noindent{\bf Step 2:}\quad
The aim of this step is showing the following relation by induction for $D(\lambda)$;
 \begin{align}
& 
\min_{\lambda' \in \bar{\Lambda}_\lambda \cup \{\lambda\}
: \theta^{(t_2(\lambda')),\lambda'} \in \mathcal{M}_{\lambda'}}
D^{F}(\theta^{(t_2(\lambda')),\lambda'}\| \theta_{(t_2(\lambda')-1),\lambda'} ) 
-C_{\inf}(\mathcal{M}_\lambda,\mathcal{E}) \nonumber \\
\le &
\frac{1}{t_1-1} D^F(\theta_{*}(\mathcal{M}_{\lambda},\mathcal{E})\| \theta_{(1)}  ) 
+ \epsilon_1 +D^F(\theta^{(t_2(\lambda)),\lambda} \| \bar\theta^{(t_2(\lambda)),\lambda}) 
\nonumber \\
&+\sum_{k=0}^{D(\lambda)-1}
\max_{\lambda'\in \bar{\Lambda}_\lambda : D(\lambda')=k}
\Big(\frac{1}{t_1-1} D^F(\theta_{*}(\mathcal{M}_{\lambda'},\mathcal{E})\| \theta_{(1)}  )
+ \epsilon_1 +D^F(\theta^{(t_2(\lambda')),\lambda'} \| \bar\theta^{(t_2(\lambda')),\lambda'}) 
\Big).
\Label{Qma2BA6}
\end{align}
Eq. \eqref{CO8T} in the proof of Theorem \ref{theo:conv2BC} implies \eqref{Qma2BA6}
with the condition $D(0)=0$.
In the following, we show \eqref{Qma2BA6} with the condition $D(\lambda)=k$
by assuming \eqref{Qma2BA6} with the condition $D(\lambda)\le k-1$.

When the relation 
\begin{align}
C_{\inf}(\mathcal{M}_{\lambda},\mathcal{E})
=C_{\inf}(\hat{\mathcal{M}}_{\lambda},\mathcal{E})
\Label{QNV}
\end{align}
does not hold,
there exists $\lambda' \in \bar{\Lambda}_{\lambda}$ such that
$C_{\inf}(\mathcal{M}_\lambda,\mathcal{E})
=C_{\inf}(\mathcal{M}_{\lambda'},\mathcal{E})$.
Since $D(\lambda ) \le k-1$, 
the assumption of induction implies \eqref{Qma2BA6}.
When the relation
\begin{align}
\theta^{(t),\lambda} \in \hat{\mathcal{M}}_{\lambda} \setminus {\mathcal{M}}_{\lambda}
\Label{QNMP}
\end{align}
dos not hold, i.e., $\theta^{(t),\lambda} \in
 {\mathcal{M}}_{\lambda}$, 
Theorem \ref{theo:conv2BC} implies \eqref{Qma2BA6}.
Hence, it is sufficient to show \eqref{Qma2BA6} when 
\eqref{QNV} and \eqref{QNMP} hold.

Due to these two conditions, 
Lemma \ref{L-XLP} implies that 
\begin{align}
\min_{\lambda' \in \Lambda_\lambda }
C_{\inf}(\mathcal{M}_{\lambda'},\mathcal{E}) 
\le D^{F}(\theta^{(t_2(\lambda)),\lambda}\| \Pro^{(e),F}_{{\cal E}}  
(\theta^{(t_2(\lambda)),\lambda}) ) .\Label{QNM4}
\end{align}
Thus,
\begin{align}
& 
\min_{\lambda' \in \Lambda_\lambda }
C_{\inf}(\mathcal{M}_{\lambda'},\mathcal{E}) 
-C_{\inf}(\mathcal{M}_\lambda,\mathcal{E}) \nonumber \\
 \le &
D^{F}(\theta^{(t_2(\lambda)),\lambda}\| \Pro^{(e),F}_{{\cal E}}  (\theta^{(t_2(\lambda)),\lambda}) ) 
-C_{\inf}(\mathcal{M}_\lambda,\mathcal{E}) \nonumber \\
\stackrel{(a)}{\le} &
D^{F}(\theta^{(t_2(\lambda)),\lambda}\| \theta_{(t_2(\lambda)-1),\lambda})  
-C_{\inf}(\mathcal{M}_\lambda,\mathcal{E}) \nonumber \\
\stackrel{(b)}{\le} &
\frac{1}{t_1-1} D^F(\theta_{*}(\mathcal{M}_{\lambda},\mathcal{E})\| \theta_{(1)}  ) 
+ \epsilon_1 +D^F(\theta^{(t_2(\lambda)),\lambda} \| \bar\theta^{(t_2(\lambda)),\lambda}) 
 \Label{QNMf},
\end{align}
where $(a)$ follows from the definition of the $e$-projection $\Pro^{(e),F}_{{\cal E}}$
and $(b)$ follows from Eq. \eqref{CO8T} in the proof of Theorem \ref{theo:conv2BC}.
Hence, we have
 \begin{align}
& \min_{\lambda' \in \bar{\Lambda}_\lambda \cup \{\lambda\}
: \theta^{(t),\lambda'} \in \mathcal{M}_{\lambda'}}
D^{F}(\theta^{(t_2(\lambda')),\lambda'}\| \theta_{(t_2(\lambda')-1),\lambda'} ) 
-C_{\inf}(\mathcal{M}_\lambda,\mathcal{E}) \nonumber \\
\stackrel{(a)}{\le} &
 \min_{\lambda' \in {\Lambda}_\lambda }
\Big(\min_{\lambda'' \in {\Lambda}_{\lambda'} \cup \{\lambda'\}
: \theta^{(t),\lambda''} \in \mathcal{M}_{\lambda''}}
D^{F}(\theta^{(t_2(\lambda'')),\lambda''}\| \theta_{(t_2(\lambda'')-1),\lambda''} ) 
-C_{\inf}(\mathcal{M}_\lambda,\mathcal{E}) \Big)\nonumber \\
=& \min_{\lambda' \in {\Lambda}_\lambda }
\Big(\min_{\lambda'' \in {\Lambda}_{\lambda'} \cup \{\lambda'\}
: \theta^{(t),\lambda''} \in \mathcal{M}_{\lambda''}}
D^{F}(\theta^{(t_2(\lambda'')),\lambda''}\| \theta_{(t_2(\lambda'')-1),\lambda''} ) 
-C_{\inf}(\mathcal{M}_{\lambda'},\mathcal{E}) \nonumber \\
&+
\big(C_{\inf}(\mathcal{M}_{\lambda'},\mathcal{E}) 
-C_{\inf}(\mathcal{M}_\lambda,\mathcal{E}) \big)
\Big)\nonumber \\
\le & 
\max_{\lambda' \in {\Lambda}_\lambda }
\Big(\min_{\lambda'' \in {\Lambda}_{\lambda'} \cup \{\lambda'\}
: \theta^{(t),\lambda''} \in \mathcal{M}_{\lambda''}}
D^{F}(\theta^{(t_2(\lambda'')),\lambda''}\| \theta_{(t_2(\lambda'')-1),\lambda''} ) 
-C_{\inf}(\mathcal{M}_{\lambda'},\mathcal{E}) \Big) 
\nonumber \\
&+
\min_{\lambda' \in {\Lambda}_\lambda }
\Big(C_{\inf}(\mathcal{M}_{\lambda'},\mathcal{E}) 
-C_{\inf}(\mathcal{M}_\lambda,\mathcal{E}) 
\Big)\nonumber \\
\stackrel{(b)}{\le} &
\max_{\lambda' \in {\Lambda}_\lambda }\Bigg(
\sum_{k=0}^{D(\lambda')-1}
\max_{\lambda''\in \bar{\Lambda}_{\lambda'} : D(\lambda'')=k}
\Big(\frac{1}{t_1-1} D^F(\theta_{*}(\mathcal{M}_{\lambda''},\mathcal{E})\| \theta_{(1)}  ) 
+ \epsilon_1 +D^F(\theta^{(t_2(\lambda'')),\lambda''} \| \bar\theta^{(t_2(\lambda'')),\lambda''}) 
\Big) \nonumber \\
&+
\frac{1}{t_1-1} D^F(\theta_{*}(\mathcal{M}_{\lambda'},\mathcal{E})\| \theta_{(1)}  )
+ \epsilon_1 +D^F(\theta^{(t_2(\lambda')),\lambda'} \| \bar\theta^{(t_2(\lambda')),\lambda'}) \Bigg)
\nonumber\\
&+
\frac{1}{t_1-1} D^F(\theta_{*}(\mathcal{M}_{\lambda},\mathcal{E})\| \theta_{(1)}  )
+ \epsilon_1 +D^F(\theta^{(t_2(\lambda)),\lambda} \| \bar\theta^{(t_2(\lambda)),\lambda}) \nonumber\\
\stackrel{(c)}{\le} &
\sum_{k=0}^{D(\lambda)-1}
\max_{\lambda'\in \bar{\Lambda}_\lambda : D(\lambda')=k}
\Big(\frac{1}{t_1-1} D^F(\theta_{*}(\mathcal{M}_{\lambda'},\mathcal{E})\| \theta_{(1)}  )
+ \epsilon_1 +D^F(\theta^{(t_2(\lambda')),\lambda'} \| \bar\theta^{(t_2(\lambda')),\lambda'}) 
\Big) \nonumber \\
&+
\frac{1}{t_1-1} D^F(\theta_{*}(\mathcal{M}_{\lambda},\mathcal{E})\| \theta_{(1)}  )
+ \epsilon_1 +D^F(\theta^{(t_2(\lambda)),\lambda} \| \bar\theta^{(t_2(\lambda)),\lambda}) ,
\Label{Qma2BAR}
\end{align}
where Step $(a)$ follows from the definition of $\bar{\Lambda}_\lambda$.
The second line of $(b)$ follows from \eqref{QNMf}.
The first line of $(b)$ follows from 
the substitution of $\lambda'$ into $\lambda$
in the relation \eqref{Qma2BA6} as the assumption of induction.
Step $(c)$ follows from the following fact.
For $\lambda' \in \Lambda_\lambda$, we have
the relations $\bar{\Lambda}_{\lambda'} \subset \bar{\Lambda}_{\lambda}$
and $D(\lambda)-1\ge D(\lambda') > D(\lambda') -1$.
Hence, we obtain \eqref{Qma2BA6} in the general case.

\noindent{\bf Step 3:}\quad
The aim of this step is showing \eqref{Qma2BA2} by using \eqref{Qma2BA6}.
We apply \eqref{Qma2BA6} to the case with $\lambda=0$.
We have
\begin{align}
& D^{F}(\theta_f^{(t)}\| \Pro^{(e),F}_{{\cal E}}  (\theta_f^{(t)}) ) 
-C_{\inf}(\mathcal{M},\mathcal{E}) \nonumber \\
\stackrel{(a)}{=} & D^{F}(\theta^{(t_2(\lambda_0)),\lambda_0}\| \Pro^{(e),F}_{{\cal E}}  (\theta^{(t_2(\lambda_0)),\lambda_0}) ) 
-C_{\inf}(\mathcal{M},\mathcal{E}) \nonumber \\
\stackrel{(b)}{\le} &
D^{F}(\theta^{(t_2(\lambda_0)),\lambda_0}\| \theta_{(t_2(\lambda_0)-1),\lambda_0})  
-C_{\inf}(\mathcal{M},\mathcal{E}) \nonumber \\
\stackrel{(c)}{=} &
\min_{\lambda' \in {\Lambda}_* 
: \theta^{(t_2(\lambda')),\lambda'} \in \mathcal{M}_{\lambda'}}
D^{F}(\theta^{(t_2(\lambda')),\lambda'}\| \theta_{(t_2(\lambda')-1),\lambda'})  
-C_{\inf}(\mathcal{M},\mathcal{E}) \nonumber \\
\stackrel{(d)}{=} &
\min_{\lambda' \in \bar{\Lambda}_0 \cup \{0\}
: \theta^{(t_2(\lambda')),\lambda'} \in \mathcal{M}_{\lambda'}}
D^{F}(\theta^{(t_2(\lambda')),\lambda'}\| \theta_{(t_2(\lambda')-1),\lambda'})  
-C_{\inf}(\mathcal{M},\mathcal{E}) \nonumber \\
\stackrel{(e)}{\le} &
\frac{1}{t_1-1} D^F(\theta_{*}(\mathcal{M}_{0},\mathcal{E})\| \theta_{(1)}  ) 
+ \epsilon_1 +D^F(\theta^{(t_2(0)),0} \| \bar\theta^{(t_2(0)),0}) 
\nonumber \\
&+\sum_{k=0}^{D(0)-1}
\max_{\lambda'\in \bar{\Lambda}_\lambda : D(\lambda')=k}
\Big(\frac{1}{t_1-1} D^F(\theta_{*}(\mathcal{M}_{\lambda'},\mathcal{E})\| \theta_{(1)}  )
+ \epsilon_1 +D^F(\theta^{(t_2(\lambda')),\lambda'} \| \bar\theta^{(t_2(\lambda')),\lambda'}) 
\Big) \nonumber \\
=&
(D(0)+1)
\max_{\lambda\in \Lambda_*}\Big(
\frac{1}{t_1-1} D^F(\theta_{*}(\mathcal{M}_\lambda,\mathcal{E})\| \theta_{(1)}  ) 
+ \epsilon_1 +D^F(\theta^{(t_2(\lambda)),\lambda} \| \bar\theta^{(t_2(\lambda)),\lambda}) 
\Big),
\end{align}
where each step is shown as follows.
$(a)$ follows from the definition of $\theta_f^{(t)}$.
$(b)$ follows from the definition of the $e$-projection $\Pro^{(e),F}_{{\cal E}}$.
$(c)$ follows from the definition of $\lambda_0$.
$(d)$ follows from the relation $\bar{\Lambda}_0 \cup \{0\}=\Lambda_*$.
$(e)$ follows from the application of \eqref{Qma2BA6} to the case with $\lambda=0$.

\section{Proofs of Lemmas \ref{BL11} and \ref{BL10}}\Label{BBO}
\begin{proofof}{Lemma \ref{BL11}}
The assumption implies $n_1 \ge n_2$. 
It is sufficient to show that the matrix 
$((P_2 J_{\theta,\tau(\theta),3} P_1)_{i,j}  )_{i = 1, \ldots, n_1-1, j=1, \ldots, n_2-1 }$ 
has at least rank $n_2-1$ under the given condition.
For $i = 1, \ldots, n_1-1, j=1, \ldots, n_2-1$,
we choose $c_{1,i}$ and $c_{2,j}$ as
\begin{align}
\sum_{x,y} P_X(x)W_{\theta,x}(y) f_j(y) &=c_{1,j} \\
\sum_{x,y} P_X(x)W_{\theta,x}(y) \delta_{i,x}&=c_{2,i},
\end{align}
where $f_j(y) $ is defined in Subsection \ref{4A}. Then, 
we have
\begin{align}
(P_2 J_{\theta,\tau(\theta),3} P_1)_{i,j}
=&\sum_{x,y} P_X(x)W_{\theta,x}(y) (\delta_{i,x}-c_{2,i})( f_j(y)-c_{1,j} ) \nonumber \\
=&\sum_{x} P_X(x)  (\delta_{i,x}-c_{2,i}) \Big(\sum_{y} W_{\theta,x}(y) ( f_j(y)-c_{1,j} ) \Big).
\end{align}
When $( f_j(y)-c_{1,j} )_{y,j}$ is considered as a matrix, its rank is $n_2-1$.
Also, $(W_{\theta,x}(y))_{x,y}$ can be regarded as a rank-$n_2-1$ matrix.
Hence, $\Big(\sum_{y} W_{\theta,x}(y) ( f_j(y)-c_{1,j} ) \Big)_{x,j}$
can be regarded as a rank-$n_2-1$ matrix.
Also, $(P_X(x)  (\delta_{i,x}-c_{2,i}))_{x,i}$ can be regarded as a rank-$n_1-1$ matrix.
Since $n_1 \ge n_2$, 
$(P_2 J_{\theta,\tau(\theta),3} P_1)_{i,j}$ is a rank-$n_1-1$ matrix.
\end{proofof}

\begin{proofof}{Lemma \ref{BL10}}
To show Lemma \ref{BL10}, we prepare the following lemma;
\begin{lemma}\Label{MOF}
We consider 
a one-parameterized family of channels $\{\bar{W}_s \}_{s \in \mathbb{R} }$
We denote the Fisher information of 
$\{\bar{W}_s \times P_X \}_{s }$
by $\bar{J}_{s,1}$.
We denote the Fisher information of 
$\{\bar{W}_s \cdot P_X \}_{s }$
by $\bar{J}_{s,2}$.
Then, 
$ \bar{J}_{s_0,1}\ge \bar{J}_{s_0,2}$.
The equality hold if and only if
the function $(x,y)\mapsto 
\frac{\frac{d }{d s}  \bar{W}_{s}(y|x) \big|_{s=s_0}}{\bar{W}_{s_0}(y|x)}$ 
can be written as a function of $y$.
\end{lemma}

We denote the mixture parameter of the exponential family 
$\{P_{XY,\theta,\tau}\}_{\theta,\tau}$
by $(\eta_1(\theta,\tau),\eta_2(\theta,\tau))$.
The condition \eqref{const1-4T} implies
\begin{align}
\eta_{2,0}(\theta,\tau(\theta))=D,
\end{align}
and the construction of $P_{XY}^{(t+1)}$ implies 
\begin{align}
\eta_{2,x}(\theta,\tau(\theta))=P_X(x)
\end{align}
for $x \in \X\setminus \{n_1\}$.
We choose a one-parameter family $c(t)\in \mathbb{R}^{n_2-1}$ such that
$c(0)=\theta_0$ and $v_1:=\frac{d}{dt} c(t)|_{t=0} \neq 0$.
Then, we have
\begin{align}
\frac{d}{dt}\eta_2(c(t),\tau(\theta_0))
+\frac{d}{dt}\eta_2(\theta_0,\tau(c(t)))=0.\Label{NCP}
\end{align}
We denote 
$ \frac{d}{dt} \tau(c(t))|_{t=0}$ by $v_2$. 
The condition \eqref{NCP} is equivalent to the condition;
\begin{align}
P_2 J_{\theta_0,\tau(\theta_0),3} P_1 v_1
+P_2 J_{\theta_0,\tau(\theta_0),3} P_2 v_2=0.
\end{align}
That is, 
\begin{align}
v_2= 
- (P_2 J_{\theta_0,\tau(\theta_0),3} P_2)^{-1}P_2 J_{\theta_0,\tau(\theta_0),3} P_1 v_1.
\end{align}
Hence, 
the vector $v_2$ is not zero for any $v_1\neq 0$
if and only if $\Ker P_2 J_{\theta,\tau(\theta),3} P_1=\{0\}$.

In addition, 
\begin{align}
&\frac{d}{dt}\Pro^{(m),\mu}_{{\cal M}}(P_{\theta,Y}\times P_X)(x,y)\Big|_{t=0}
\nonumber \\
=&\Pro^{(m),\mu}_{{\cal M}}(P_{\theta_0,Y}\times P_X)(x,y)
( \sum_{i=1}^{n_2-1}v_1^i f_i(y)+ \sum_{j=0}^{n_1-1} v_2^j g_j(x,y)).
\end{align}
Therefore, 
$\frac{\frac{d}{dt}\Pro^{(m),\mu}_{{\cal M}}(P_{\theta,Y}\times P_X)(x,y)\Big|_{t=0}}{
\Pro^{(m),\mu}_{{\cal M}}(P_{\theta_0,Y}\times P_X)(x,y)}$
cannot be written as a function of $y$ for any $v_1 \neq 0$
if and only if $\Ker P_2 J_{\theta,\tau(\theta),3} P_1=\{0\}$.

We define $W_\theta$ as $ W_\theta\times P_X= \Pro^{(m),\mu}_{{\cal M}}(P_{\theta,Y}\times P_X)$.
Applying Lemma \ref{MOF} with substitution of $W_{c(t)}$ into $\bar{W}_s$,
we obtain the desired statement of Lemma \ref{BL10} from the above equivalence relation.
\end{proofof}

\begin{proofof}{Lemma \ref{MOF}}
\begin{align}
& \bar{J}_{s_0,1} \nonumber \\
=&  \sum_{x,y} \Big(\frac{d}{ds} \bar{W}_{s}(y|x)|_{s=s_0}\Big)^2  
\bar{W}_{s_0} (y|x)^{-1} P_X(x) \nonumber \\
= &  \sum_{y} \Big(\sum_{x'} \frac{d}{ds} \bar{W}_{s}(y|x') |_{s=s_0} P_X(x')  \Big)^2  (\sum_{x'} \bar{W}_{s_0} (y|x') P_X(x'))^{-1} \nonumber \\
 &+  \sum_{y} \Big(\sum_{x'} \bar{W}_{s_0} (y|x') P_X(x')\Big)
\nonumber \\
&\quad \cdot  \sum_x \Big( \frac{d}{ds} 
\frac{ P_X(x)  \bar{W}_{s}(y|x) }{\sum_{x'} \bar{W}_{s} (y|x') P_X(x')}
 \Big|_{s=s_0}   \Big)^2   \Big(\frac{ P_X(x)  \bar{W}_{s_0}(y|x) }{\sum_{x'} \bar{W}_{s_0} (y|x') P_X(x')}
\Big)^{-1} .
\end{align}
Hence,
\begin{align}
&
\bar{J}_{s_0,1}-\bar{J}_{s_0,2}\nonumber \\
=&
 \sum_{y} \Big(\sum_{x'} \bar{W}_{s_0} (y|x') P_X(x')\Big)
\nonumber \\
&\quad \cdot \sum_x \Big( \frac{d}{ds} 
\frac{ P_X(x)  \bar{W}_{s}(y|x) }{\sum_{x'} \bar{W}_{s} (y|x') P_X(x')}
 \Big|_{s=s_0}   \Big)^2   
 \Big(\frac{ P_X(x)  \bar{W}_{s_0}(y|x) }{\sum_{x'} \bar{W}_{s_0} (y|x') P_X(x')}
\Big)^{-1}\nonumber \\
=&
 \sum_{y} \Big(\sum_{x'} \bar{W}_{s_0} (y|x') P_X(x')\Big)
\nonumber \\
&\quad \cdot \sum_x \Big( \frac{d}{ds} 
\log \Big(\frac{ P_X(x)  \bar{W}_{s}(y|x) }{\sum_{x'} \bar{W}_{s} (y|x') P_X(x')}\Big)
 \Big|_{s=s_0}   \Big)^2   
 \Big(\frac{ P_X(x)  \bar{W}_{s_0}(y|x) }{\sum_{x'} \bar{W}_{s_0} (y|x') P_X(x')}
\Big)\nonumber \\
=&
 \sum_{x,y} \bar{W}_{s_0} (y|x) P_X(x)
\nonumber \\
&\quad \cdot \Big( \frac{d}{ds} 
(\log  P_X(x)  \bar{W}_{s}(y|x) )- \log 
\Big( \sum_{x'} \bar{W}_{s} (y|x') P_X(x')\Big)
 \Big|_{s=s_0}   \Big)^2   
\nonumber \\
=&
 \sum_{x,y} \bar{W}_{s_0} (y|x) P_X(x)
\nonumber \\
&\quad \cdot \Big( 
\frac{d}{ds} \log   \bar{W}_{s}(y|x) ) \Big|_{s=s_0}
- \frac{d}{ds} \log 
\Big( \sum_{x'} \bar{W}_{s} (y|x') P_X(x')\Big)
 \Big|_{s=s_0}   \Big)^2   
\nonumber \\
=&
 \sum_{x,y} \bar{W}_{s_0} (y|x) P_X(x)
\Big( 
\frac{\frac{d}{ds}  \bar{W}_{s}(y|x) \Big|_{s=s_0}}{\bar{W}_{s_0}(y|x)}
-  
\frac{\frac{d}{ds} \Big( \sum_{x'} \bar{W}_{s} (y|x') P_X(x')\Big)\Big|_{s=s_0}  }
{\sum_{x'} \bar{W}_{s_0} (y|x') P_X(x')} 
  \Big)^2   
\nonumber \\
=&
 \sum_{x,y} \bar{W}_{s_0} (y|x) P_X(x)
\Big( 
l(x,y)-  
\frac{\Big( \sum_{x'} l(x',y) \bar{W}_{s_0} (y|x') P_X(x')\Big) }
{\sum_{x'} \bar{W}_{s_0} (y|x') P_X(x')} 
  \Big)^2   
\nonumber \\
=&
 \sum_{x,y} \bar{W}_{s_0} (y|x) P_X(x)
\Big( 
l(x,y)-  
\frac{\Big( \sum_{x'} l(x',y) \bar{W}_{s_0} (y|x') P_X(x')\Big) }
{\sum_{x'} \bar{W}_{s_0} (y|x') P_X(x')} 
  \Big)^2   ,
\end{align}
where $l(x,y):=
\frac{\frac{d}{ds}  \bar{W}_{s}(y|x) \Big|_{s=s_0}}{\bar{W}_{s_0}(y|x)}$.
Hence, we have $\bar{J}_{s_0,1}-\bar{J}_{s_0,2}\ge 0$.
The equality holds if and only if 
$l(x,y)$ depends only on $y$.
The desired statement is obtained.
\end{proofof}

\end{document}